\documentclass[runningheads,envcountsame,orivec]{llncs}
\usepackage[T1]{fontenc}

\usepackage[
  b5paper,
  textwidth=12.2cm,
  textheight=19.3cm,
  centering
]{geometry}

\usepackage{mathtools}

\usepackage{amsthm}
\usepackage{dsfont}
\usepackage{amssymb}
\usepackage{biolinum}
\DeclareMathAlphabet{\mathsf}{OT1}{LinuxBiolinumT-TLF}{m}{n}
\SetMathAlphabet{\mathsf}{bold}{OT1}{LinuxBiolinumT-TLF}{b}{n}
\usepackage[table]{xcolor}
\usepackage[dvipsnames]{xcolor}
\usepackage{esvect}
\usepackage[normalem]{ulem}
\definecolor{HANADA}{RGB}{0, 98, 132}
\definecolor{KURENAI}{RGB}{203, 27, 69}
\usepackage{breakcites}
\usepackage[
    pdfstartview=FitH,
    pdfpagemode=UseNone,
    colorlinks=true,
    citecolor=KURENAI,
    linkcolor=HANADA,
    linktocpage=true
]{hyperref}
\usepackage{orcidlink}
\renewcommand{\orcidID}[1]{\orcidlink{#1}}
\usepackage[capitalise, nameinlink]{cleveref}
\usepackage[nottoc]{tocbibind}
\usepackage{appendix}
\usepackage{array}
\usepackage{braket}
\usepackage{enumitem}

\newtheorem{fact}[theorem]{Fact}

\newtheorem{claim}[theorem]{Claim}
\newtheorem{notation}[theorem]{Notation}
\newcommand{\fig}[1]{\hyperref[fig:#1]{Figure~\ref*{fig:#1}}}
\newcommand{\eq}[1]{\hyperref[eq:#1]{(\ref*{eq:#1})}}
\newcommand{\lem}[1]{\hyperref[lem:#1]{Lemma~\ref*{lem:#1}}}
\newcommand{\thm}[1]{\hyperref[thm:#1]{Theorem~\ref*{thm:#1}}}
\newcommand{\defi}[1]{\hyperref[def:#1]{Definition~\ref*{def:#1}}}
\newcommand{\app}[1]{\hyperref[sec:#1]{Appendix~\ref*{sec:#1}}}
\newcommand{\fct}[1]{\hyperref[fact:#1]{Fact~\ref*{fact:#1}}}
\newcommand{\clr}[1]{\hyperref[clr:#1]{Corollary~\ref*{clr:#1}}}
\newcommand{\sct}[1]{\hyperref[sec:#1]{Section~\ref*{sec:#1}}}
\newcommand{\subsec}[1]{\hyperref[subsec:#1]{Subsection~\ref*{subsec:#1}}}
\newcommand{\itm}[2]{\hyperref[itm:#1]{#2}}
\newcommand{\clm}[1]{\hyperref[clm:#1]{Claim~\ref*{clm:#1}}}
\newcommand{\rmk}[1]{\hyperref[rmk:#1]{Remark~\ref*{rmk:#1}}}
\definecolor{lightcyan}{RGB}{0.88,1,1}
\definecolor{darkgreen}{RGB}{0, 128, 0}
\definecolor{darkblue}{RGB}{0, 0, 128}

\makeatletter
\newtheorem*{rep@theorem}{\rep@title}
\newcommand{\newreptheorem}[2]{\newenvironment{rep#1}[1]{ \def\rep@title{#2 \ref*{##1}} \begin{rep@theorem}} {\end{rep@theorem}}}
\makeatother

\newreptheorem{theorem}{Theorem}
\newreptheorem{lemma}{Lemma}
\newreptheorem{claim}{Claim}

\newcommand{\N}{\mathbb{N}}   \newcommand{\C}{\mathbb{C}}   \DeclareMathOperator*{\E}{\mathbb{E}}     \renewcommand{\i}{\mathrm{i}} \newcommand{\A}{\mathcal{A}}    \newcommand{\D}{\mathcal{D}}  \renewcommand{\H}{\mathcal{H}} \newcommand{\K}{\mathcal{K}}  \newcommand{\OO}{\mathcal{O}}  \newcommand{\X}{\mathcal{X}} \newcommand{\Y}{\mathcal{Y}} \renewcommand{\S}{\mathcal{S}}      \newcommand{\RR}{\mathfrak{R}}       \newcommand{\bit}[1]{\{0,1\}^{#1}}     \newcommand{\expect}[2]{\E_{\substack{#1}}\!\Br{#2}} \newcommand{\prob}[2]{\underset{#1}{\mathrm{Pr}}\!\Br{#2}}      

\newcommand{\br}[1]{\left(#1\right)} \newcommand{\Br}[1]{\left[#1\right]} \newcommand{\st}[1]{\left\{#1\right\}}  \newcommand{\abs}[1]{\left|#1 \right|} \newcommand{\norm}[1]{\left\lVert #1 \right\rVert}    \newcommand{\poly}[1]{\mathrm{poly}\!\br{#1}} \newcommand{\negl}[1]{\mathrm{negl}\!\br{#1}}  \newcommand{\val}[1]{\mathrm{val}\!\br{#1}}   \newcommand{\ketbratwo}[2]{\ket{#1} \hspace{-0.4em}\bra{#2}} \newcommand{\ketbra}[1]{\ketbratwo{#1}{#1}} \newcommand{\id}{\ensuremath{\mathds{1}}}  \newcommand{\ugroup}[1]{\mathrm{U}\!\br{#1}} \newcommand{\td}{\mathrm{TD}}      \newcommand{\ncopy}{\ell}     \newcommand{\haar}{\ensuremath{\mu}}    \newcommand{\tp}{\otimes}   

\newcommand{\secpar}{\lambda}  \newcommand{\cprim}[1]{\textup{\textsf{#1}}} \newcommand{\prf}{\cprim{PRF}} \newcommand{\prp}{\cprim{PRP}} \newcommand{\pru}{\cprim{PRU}} \newcommand{\prs}{\cprim{PRS}}    \newcommand{\rss}{\cprim{RSS}}                       \newcommand{\prg}{\cprim{PRG}} \newcommand{\qprf}{\cprim{QPRF}} \newcommand{\qprp}{\cprim{QPRP}}   \newcommand{\owf}{\cprim{OWF}}      \newcommand{\oracle}{\mathcal{O}} 

\usepackage[normalem]{ulem}

   \newcommand{\newDB}{\mathsf{DB}}

 \newcommand{\HP}{\ensuremath{\mathsf{HP}}} \newcommand{\kac}{\ensuremath{\mathsf{Kac}}} \newcommand{\hkac}{\ensuremath{\widehat{\mathsf{Kac}}}} \newcommand{\HPC}{\ensuremath{\mathsf{HP}}} \newcommand{\hHPC}{\ensuremath{\widehat{\mathsf{HP}}}} \newcommand{\HPO}{\ensuremath{\mathsf{HPO}}} \newcommand{\reg}[1]{\ensuremath{_{\textcolor{Black}{\mathsf{#1}}}}}

\newcommand{\DB}{\newDB_t}
\newcommand{\DBRt}{\ensuremath{\mathfrak{R}_t^{\mathsf{DB}}}}
\newcommand{\DBR}{\ensuremath{\mathfrak{R}^{\mathrm{DB}}}}
\newcommand{\yDBRt}{\ensuremath{\mathfrak{R}_t^{\mathsf{yDB}}}}
\newcommand{\yDBR}{\ensuremath{\mathfrak{R}^{\mathrm{yDB}}}}
\newcommand{\CPS}{\ensuremath{\mathsf{Compress}}}
\newcommand{\PR}{\ensuremath{\mathsf{PR}}}
\newcommand{\DBproj}{\ensuremath{\Pi\reg{X}^{(t)}}}
\newcommand{\DBp}[1]{\ensuremath{\Pi\reg{#1}^{(t)}}}
\newcommand{\TD}[2]{\ensuremath{\mathsf{TD}\br{#1,#2}}}
\newcommand{\Tr}{\mathrm{Tr}} \newcommand{\DBRproj}{\ensuremath{\widetilde{\Pi}^{(t)}\reg{HP}}}
\renewcommand{\set}[1]{\ensuremath{\left\{#1\right\}}}
\newcommand{\BImR}{\ensuremath{\mathrm{BIm}\br{R}}}
\newcommand{\BDomR}{\ensuremath{\mathrm{BDom}\br{R}}}
\newcommand{\BImS}{\ensuremath{\mathrm{BIm}\br{S}}}
\newcommand{\ImR}{\ensuremath{\mathrm{Im}\br{R}}}
\newcommand{\ImS}{\ensuremath{\mathrm{Im}\br{S}}}
\newcommand{\DomR}{\ensuremath{\mathrm{Dom}\br{R}}}
\newcommand{\DomS}{\ensuremath{\mathrm{Dom}\br{S}}}
\newcommand{\BDom}[1]{\ensuremath{\mathrm{BDom}\br{#1}}}
\newcommand{\BIm}[1]{\ensuremath{\mathrm{BIm}\br{#1}}}
\newcommand{\dom}[1]{\ensuremath{\mathrm{Dom}\br{#1}}}
\newcommand{\im}[1]{\ensuremath{\mathrm{Im}\br{#1}}}
\newcommand{\domProj}[1]{\ensuremath{\Pi^{\dom{#1}}}}
\newcommand{\imProj}[1]{\ensuremath{\Pi^{\im{#1}}}}
\newcommand{\num}[1]{\ensuremath{\mathrm{num}\!\br{#1}}}
\newcommand{\e}{\mathrm{e}} 
\newcommand{\eqProj}{\ensuremath{\Pi^{\mathsf{eq}}}}
\newcommand{\ffbProj}{\ensuremath{\Pi^{\mathsf{ffb}}}}

\newcommand{\rssd}{\ensuremath{\mathcal{R}}}

\newcommand{\re}{\mathrm{Re}}
\newcommand{\cQ}{\ensuremath{\mathsf{cQ}}}
\newcommand{\cC}{\ensuremath{\mathsf{cC}}}
\newcommand{\cD}{\ensuremath{\mathsf{cD}}}
\newcommand{\init}{\ensuremath{\mathsf{init}}}
\newcommand{\HPOp}{\ensuremath{\widetilde{\mathsf{HPO}}}}

\newcommand{\idy}{\mathsf{Id}}
\newcommand{\regi}[1]{\mathsf{#1}}
\newcommand{\widebar}[1]{\overline{#1}}
\begin{document}

\title{Parallel Kac's Walk Generates \pru}

\author{\mbox{Chuhan Lu\inst{1}\orcidID{0009-0000-3359-320X}} \and
\mbox{Minglong Qin\inst{2}\orcidID{0009-0004-8760-5498}} \and
\mbox{Fang Song\inst{3}\orcidID{0000-0002-3098-6451}} \and
\mbox{Penghui Yao\inst{4,5}\orcidID{0000-0002-4104-2069}} \and \\
\mbox{Mingnan Zhao\inst{4}\orcidID{0009-0006-9623-9703}}}

\authorrunning{C. Lu et al.}

\institute{
    Department of Computer Science, Rice University, USA\\
    \email{cl308@rice.edu}
    \and
    Centre for Quantum Technologies, National University of Singapore, Singapore\\
    \email{mlqin6@gmail.com}
	\and
    Computer Science Department, Portland State University, USA\\
    \email{fang.song@pdx.edu}
    \and
    State Key Laboratory for Novel Software Technology, Nanjing University,\\Nanjing 210023, China\\
    \email{phyao1985@gmail.com, mingnanzh@gmail.com}
    \and
    Hefei National Laboratory, Hefei 230088, China
}
\date{ }

\maketitle

\begin{abstract}
Ma and Huang recently proved that the \textsf{PFC}
construction, introduced by Metger, Poremba, Sinha and
Yuen~\cite{MPSY24}, gives an adaptive-secure
pseudorandom unitary family (\pru). Their proof
developed a new \emph{path recording}
technique \cite{MH25}.

In this work, we show that a linear number of
sequential repetitions of the \emph{parallel Kac's
	Walk}, introduced by Lu, Qin, Song, Yao and
Zhao~\cite{LQSY+24}, also forms an adaptive-secure $\pru$,
confirming a conjecture therein. Moreover, it
additionally satisfies strong security against
adversaries making inverse queries. This gives an
alternative $\pru$ construction, and provides another
instance demonstrating the power of the path recording technique. We also discuss some further
simplifications and implications.
 \end{abstract}

\allowdisplaybreaks

\section{Introduction}
\label{sec:intro}
\emph{Pseudorandomness} is a fundamental concept in cryptography. The basic pseudorandom objects, including pseudorandom functions ($\prf$s), pseudorandom permutations ($\prp$s) and pseudorandom generators ($\prg$s), have served as primitives in modern classical cryptography. 

Pseudorandom objects in quantum information have witnessed increasing
influences in recent years. The first example is \emph{pseudorandom states}
($\prs$s), introduced by Ji, Liu and Song~\cite{JLS18}, which are a set
of states that can be prepared by polynomial-sized quantum circuits
and look indistinguishable from Haar random states for any
polynomial-time quantum distinguisher. They gave the first
construction of $\prs$s. $\prs$s have found applications in various areas
including quantum cryptography~\cite{AQY22}, quantum learning
theory~\cite{doi:10.1126/science.abn7293} and quantum
gravity~\cite{BFV20,YE23}. Since their work, a number of different
constructions have been
discovered~\cite{BS19,giurgicatiron2023pseudorandomnesssubsetstates,jeronimo2024pseudorandompseudoentangledstatessubset,brakerski2024realvalued}.
Ji, Liu and Song further introduced the concept of \emph{pseudorandom
	unitaries} ($\pru$s), which are ensembles of unitaries that are efficient
to implement, but are indistinguishable from Haar random unitaries by
any quantum polynomial-time distinguisher. The construction of $\pru$
turns out to be challenging and a large body of work has been
devoted to constructing pseudorandom objects that partially realize
$\pru$'s functions. Examples include scalable pseudorandom
states~\cite{BS20}, pseudorandom function-like quantum state
generators~\cite{AGQY22}, pseudorandom isometries~\cite{AGKL23},
and pseudorandom scramblers~\cite{LQSY+24}.

A giant leap was achieved by Metger, Poremba, Sinha and
Yuen~\cite{MPSY24} where a pseudorandom unitary family is constructed
against \emph{non-adaptive} adversaries. Their construction is a \cprim{PFC}
ensemble, where the circuits sequentially apply a uniformly random
Clifford gate, a diagonal unitary with a (pseudorandom) random phase
$f:\set{0,1}^n\rightarrow\set{-1,1}$, and then a (pseudorandom)
permutation matrix on $n$-qubit computational basis states. Soon
after, Ma and Huang in their very recent work~\cite{MH25} proved that
\cprim{PFC} ensembles are indeed secure against \emph{adaptive} adversaries as well,
provably establishing the feasibility of $\pru$ for the first time.

\subsection{Main result}

In this work, we carry on the exciting advancement on $\pru$ lately. We
show that the \emph{parallel Kac's walk} construction, introduced
in~\cite{LQSY+24}, also produces an \emph{adaptive-secure} $\pru$.  
The original \emph{Kac's walk},
introduced by Kac~\cite{Kac56} in 1956, is a random walk on unitary
groups, which has been extensively studied by mathematical
physicists. Kac's walk has also found applications in the construction of
polynomial designs~\cite{BHH16} and memory-optimal dimension
reduction~\cite{JPSSS22}. To put it in the language of quantum
computing, in one step of Kac's walk, the algorithm randomly selects two
elements of the computational basis $\set{\ket{i},\ket{j}}$ and implements
a random $2\times 2$ unitary in the space spanned by these two basis
elements. Kac's walk can also be viewed as a random walk on a unit
sphere, where all the random unitaries sampled by the walk are applied
to an initial unit vector sequentially. Pillai and Smith~\cite{PS17} showed a
tight $\Theta(N\log N)$ mixing time for Kac's walk on an
$N$-dimensional unit sphere.

Lu, Qin, Song, Yao and Zhao ~\cite{LQSY+24} introduced a variant
called \emph{parallel Kac's walk}. It randomly samples a random matching
among all $N$ computational basis elements, say
$\set{\br{\ket{i_1},\ket{j_1}},\ldots,\br{\ket{i_{N/2}},\ket{j_{N/2}}}}$, 
which is done in a single step via a random permutation of all basis elements.
For each pair $\br{\ket{i},\ket{j}}$, the algorithm implements a 
$2\times 2$ Haar random unitary. By replacing the random functions and
random permutations with their quantum-secure pseudorandom
counterparts, one obtains an efficient implementation of one step of
the parallel Kac's walk. It was shown that a parallel Kac's walk reduces
the mixing time by a factor $N$. Consequently, $O(\log N)$ steps of parallel
Kac's walk, which is linear in the number of qubits $n$, map any input
pure state to a (pseudo)random state.
This is formalized as a primitive called (\emph{pseudo})\emph{random state scrambler} in~\cite{LQSY+24}.
It was left open in their paper whether such a construction can generate a $\pru$. In this
paper, we show an affirmative answer that \emph{linear} steps of parallel
Kac's walk indeed form a $\pru$.

\begin{theorem}[Main Theorem, Informal]
	The distribution of the unitary corresponding to an $O(n)$-step random parallel Kac's walk is computationally indistinguishable from Haar distribution against adaptive adversaries. 
	Moreover, without asymptotically increasing the number of steps, it also remains secure against adversaries capable of making inverse queries.
\end{theorem}

\paragraph{Proof overview.} 
The proof comprises two phases. First we show that
$O(n)$ iterations of parallel Kac's walk effectively project the
adversary's state to what we term the \emph{distinct block
  subspace}. We divide $\bit{n}$ into $N/2$ blocks, each
containing two bit strings. We define $(x_1,\dots,x_t)\in(\bit{n})^t$
to be in the distinct block subspace if they belong to different
blocks. Once the state is promised to reside in the distinct block
subspace, the second phase employs a single step of parallel Kac's
walk to ensure that the entire construction is indistinguishable from a
Haar random unitary.

More specifically, in the initial phase, we leverage the random state
scrambling property of parallel Kac's walk as introduced in
\cite{LQSY+24} to ensure that the adversary's state is approximately in the 
\emph{distinct subspace} after $O(n)$ steps. 
To further show the projection into the distinct block subspace, 
imagine inserting an additional random permutation that does not affect the construction 
(since every iteration inherently includes a random permutation). 
Given that the adversary's state already resides in the
distinct subspace, applying a random permutation will cause it to fall into the
distinct block subspace with high probability.
Note that 2-designs, such as random Clifford, can also project states onto distinct block subspaces. However, parallel Kac’s walk alone is not known to form a 2-design, and hence this does not help us build a $\pru$ purely based on the parallel Kac’s walk. Therefore, we instead directly prove that the parallel Kac’s walk suffices for this purpose by an unexpected application of a random scrambling property established in~\cite{LQSY+24}.

In the second phase, we adapt the techniques in \cite{MH25}. A key ingredient in their method is the
path-recording oracle of~\cite{MH25}, which was designed so that it behaves indistinguishably from both the (purified) $\mathsf{PF}$ and Haar random distribution as long as the queries are restricted to tuples of distinct elements (a.k.a. the distinct subspace). In our case, because of the $2\times 2$ block diagonal structure of the random rotations in parallel Kac's walk, we instead need to work with the distinct block subspace as indicated earlier. We thus design a variant path-recording oracle tailored to the distinct block subspace structure, and show that it serves as an effective bridge to reason about the indistinguishability of one-round parallel Kac's walk and Haar random unitary when restricting queries in the distinct block subspace. This also illustrates the versatility of the path-recording technique and provides a template for further extension.

\subsection{Discussions}
As far as we are aware, existing provably-secure constructions of \pru~all follow the \cprim{PFC} paradigm. For example, Schuster, Haferkamp and
Huang~\cite{schuster2024randomunitariesextremelylow} developed a ``gluing'' lemma which can
compose $\omega(\log n)$-qubit $\pru$ by a two-layer brickwork into an $n$-qubit $\pru$. However, prior to our construction, the only instantiation of the small $\omega(\log n)$-qubit $\pru$ would be by \cprim{PFC}. More recently in~\cite{CSBH25}, it was observed that the (pseudo-)random permutation $P$ in the \cprim{PFC} construction can be substituted by a two-round Luby-Rackoff Feistel network, and several $\pru$ constructions with reduced depth were obtained. This can be seen as a further refinement of the \cprim{PFC} paradigm. Our Kac-walk based \pru~hence gives the first secure construction that explores a new route drastically deviating from the \cprim{PFC} style.\footnote{An earlier version of \cite{BHHP24} claimed a proof for a variant of a candidate $\pru$ proposed in~\cite{JLS18}, but a bug was found and the claim was retracted later.} 

While our \pru~construction does not provide efficiency advantages over
the \cprim{PFC} construction, there are potential benefits in other
regards. First, it is usually valuable to have multiple candidates for
a primitive available, to cater to different use cases and to mitigate
the risk in case of unexpected vulnerabilities in some candidates. It
is also preferable in cryptography, for practical implementation
concerns, to base the construction on as few primitives as
possible. For instance, the popular \cprim{HMAC} instantiates the
Hash-then-\cprim{MAC} paradigm using hash functions only as opposed to the
vanilla instantiation by a hash function and a \cprim{MAC} scheme
separately. We can view each step of the parallel Kac's walk as a basic
module and then the construction consists of repetitions of this basic
module. Finally, since $\pru$s
are the quantum analogue of pseudorandom permutations (i.e., block
ciphers), our \pru~is also reminiscent of the famous Feistel network and variants in practical block ciphers such as \cprim{AES},
where a basic unit is iterated in multiple rounds to achieve desirable
security properties.

The construction and analysis in our work also shed new light on designing $\pru$ and quantum pseudorandomness more broadly. From the perspective of appropriate path-recording oracles, one can argue that $\textsf{PF}$ in the \cprim{PFC} construction and one-round parallel Kac's walk are already indistinguishable from Haar random unitary under respective restrictions on the adversary's queries. Thus to obtain a full-fledged \pru, we can prepend to it a ``filter'' (random Clifford/2-design $C$ or $O(n)$-round parallel Kac's walk) to enforce the restriction. Alternatively, we can view both $\textsf{PF}$ and one-round parallel Kac's walk as a ``randomness-compiler'': as long as the initial object (random Clifford/2-design $C$ in PFC or $O(n)$-round leading parallel Kac's walk) admits some well-specified ``benign'' randomizing property (e.g., approximating Haar under some twirling operation), composing with $\textsf{PF}$ or one-round parallel Kac's walk will convert it to a \pru. We distill such claims in the main body, which would be useful in building more $\pru$s and other quantum (pseudo-)random objects. For example, one can immediately obtain a host of new $\pru$ constructions, such as an arbitrary 2-design or random-state scrambler followed by one-round parallel Kac's walk. 

Some open questions naturally follow from the above discussion. Can we
reduce the number of rounds of the parallel Kac's walk, ideally to
constant rounds?  This appears difficult with the current analysis, and
new techniques may be needed. Following the analogy we draw between
our construction and classical constructions of block ciphers {\`a}
la Feistel network, it is worth exploring quantum analogues of the
wide variations on Feistel network (e.g., unbalanced Luby-Rackoff). 
Can our construction be further simplified? Can we replace all i.i.d. 
random rotations in one step of parallel Kac's walk by the same random 
rotation? Recent research on orthogonal repeated averaging, which is a 
simplified Kac's walk, alludes to an affirmative answer. The other possible 
simplification is replacing the $\prp$s in the construction by random local 
permutations, like practical architecture of \cprim{DES}[2]-brickwork circuits. 
If both simplifications were plausible, we would obtain a construction of local 
random circuits, which is also a $\pru$, answering a longstanding open 
problem in quantum complexity theory. On a more technical note, 
\cite{LQSY+24} identifies a strong \emph{dispersing} property of the parallel Kac's walk.
Does this property carry over and equip our $\pru$ with additional properties? This is both
interesting for the sake of Kac's walk and possible applications of the
resulting $\pru$.
Finally, can Kac's walk lead to pseudorandom operations on other groups?
In particular, since Kac's walk also arises on $\mathrm{SO}(N)$
\cite{Jan2001,Maslen2003,Jan03,Carlen2003,Oliveira09,PS18,PS26},
it would be interesting to investigate
whether it yields other pseudorandom operations, such as
pseudorandom orthogonal operators.

\paragraph{Acknowledgment.}
This work was done while Chuhan Lu was at Portland State University.
MQ was supported by the National Research Foundation, Singapore,
through the National Quantum Office, hosted in A*STAR,
under its Centre for Quantum Technologies Funding Initiative (S24Q2d0009).
PY and MZ were supported by the National Natural Science Foundation of China (Grant Nos. 62332009 and 12347104), the Quantum Science and Technology—National Science and Technology Major Project (Grant No. 2021ZD0302901), the NSFC/RGC Joint Research Scheme (Grant No. 12461160276), the Natural Science Foundation of Jiangsu Province (No. BK20243060), the Fundamental and Interdisciplinary Disciplines Breakthrough Plan of the Ministry of Education of China (No. JYB2025XDXM118), the ``111 Center'' (No. B26023), and the Fundamental Research Funds for the Central Universities (Grant No. 2026300376).
 
\paragraph{Organization.}
The remainder of the paper is organized as follows. \cref{sec:prelim} introduces the notation and cryptographic preliminaries. \cref{sec:DBproj} proves the projection lemma into the distinct block subspace. \cref{sec:construction} presents the parallel Kac's walk construction and proves its adaptive security. \cref{sec:strong_security} extends the analysis to inverse queries and proves strong security. Deferred proofs are given in the appendix.

\section{Preliminaries}
\label{sec:prelim}
\subsection{Notation}
Unless stated otherwise, we use $n$ to denote the number of qubits and $N = 2^n$ to denote the dimension.
We denote the set of unitaries of dimension $N$ by $\ugroup{N}$. The symbol $\mu$ represents the Haar random distribution over quantum states or unitaries, depending on the context.
For finite sets $\mathcal{X}$ and
$\mathcal{Y}$, we use $\mathcal{X}^\mathcal{Y}$ to denote the set of
all functions $\{f:\mathcal{X}\to \mathcal{Y}\}$.
We generally refer to the permutation group over elements in set $\X$ as $\S_\X$. We often write $\S_{N}$ instead of
$\S_{\bit{n}}$ to denote the permutation group over elements in
$\st{0,1}^n$. In the context of unitary matrices, $\S_N$ refers to the group of permutation unitaries of dimension $N$, and $S \gets \S_N$ indicates sampling a permutation unitary uniformly at random.
Given two density operators $\rho, \eta$, the trace distance between them is $\td\br{\rho, \eta} = \norm{\rho-\eta}_1$.

For $x\in\bit{n}$,
we define
$\val{x} = \sum_{i=1}^n 2^{-i}x_i$ and
use $\overline{x}\in \bit{n}$ to denote the binary string
obtained by flipping the first bit of $x$.
We divide the set $\bit{n}$ into $2^{n-1}$ blocks
according to the suffix of each string.
For any $x,y\in\bit{n}$,
we say that $x$ and $y$ belong to the same block
if $x$ and $y$ share the same suffix of length $n-1$
(i.e., $x_2=y_2,\dots,x_n=y_n$).
Conversely, we say that $x$ and $y$ are in \emph{distinct blocks}
if they have different suffixes.
For $t\in \N$, we use $\DB$ to denote the set of all $t$-tuples
consisting of strings from distinct blocks. That is,
\[
	\newDB_t \coloneq \set{(x_1,\dots,x_t)\in(\bit{n})^t:\forall\ i\ne j,\ x_i\text{ and }x_j\text{ are in distinct blocks}}.
\]

We also need the following lemmas:
\begin{lemma}[{\cite[Lemma 2.2]{MH25}}]\label{lem:projdist}
Let $\rho\reg{CD}$ be a density matrix on registers $\mathsf{C},\mathsf{D}$ and let $\Pi\reg{CD}$ be a projector of the form $\Pi\reg{CD}=\mathsf{Id}\reg{C}\otimes\Pi'\reg{D}$, where $\Pi'\reg{D}$ is a projector that acts on register $\mathsf{D}$. Then
\[\TD{\Tr\reg{D}(\rho\reg{CD})}{\Tr\reg{D}\br{\Pi\reg{CD}\cdot\rho\reg{CD}\cdot\Pi\reg{CD}}}=1-\Tr\br{\Pi\reg{CD}\rho\reg{CD}} \enspace.\]
\end{lemma}

\begin{lemma}[Gentle Measurement Lemma {\cite[Lemma 2.3]{MH25}}]
\label{lem:gentle_measurement}
Let $\ket{\psi}$ be a quantum state,
$U_1,\dots,U_t$ be unitary operators,
and $\Pi_1,\dots,\Pi_t$ be projectors.
We have
\[
\norm{U_t\cdots U_1\ket{\psi} - \Pi_t U_t\cdots \Pi_1 U_1\ket{\psi}}_2
\leq
t\cdot\sqrt{1-\norm{\Pi_t U_t\cdots \Pi_1 U_1\ket{\psi}}_2^2}\enspace.
\]
\end{lemma}

\subsection{Adversary with Access to Oracle}
We adopt the model for adversaries with oracle access in \cite{MH25}.
\begin{definition}[Adversary with Access to Oracle]
	An adversary $\A$ with oracle access is a quantum algorithm
	which queries an oracle $\OO$ on its first $n$-qubit register $\mathsf{A}$ without knowing the description of $\OO$. 
	The adversary owns an additional $m$-qubit ancillary register, denoted by $\mathsf{B}$.
	
	A $t$-query adversary $\A$ with oracle access is specified by a $t$-tuple of unitaries $(A^{(1)}\reg{AB},\cdots, A^{(t)}\reg{AB})$.
	The view of the adversary after all the queries is
	\[
		\ket{\A^{\OO}_t}\reg{AB} \coloneq
		\prod_{i=1}^t
		\br{\OO\reg{A}\cdot A^{(i)}\reg{AB}}
		\ket{0}\reg{AB} \enspace.
	\]
	
	We also allow the adversary $\A$ to make both forward queries (i.e., to $\OO$) and inverse queries (i.e., to $\OO^\dag$).
	In this case,
	a $t$-query adversary $\A$
	is specified by a $t$-tuple of unitaries $(A^{(1)}\reg{AB},\cdots, A^{(t)}\reg{AB})$ and a Boolean string $b\in\bit{t}$.
	The view of the adversary after all the queries is
	\[
		\ket{\A^{\OO}_t}\reg{AB} \coloneq
		\prod_{i=1}^t
		\br{\br{(1-b_i)\cdot\OO+b_i\cdot\OO^\dag}\reg{A}\cdot A^{(i)}\reg{AB}}
		\ket{0}\reg{AB} \enspace.
	\]
	Unless otherwise specified,
	we assume that the adversary makes
	only forward queries.
\end{definition}

\begin{definition}[Computational Indistinguishability]
	We say two distributions $\D_1$ and $\D_2$ over $\ugroup{N}$ are computationally indistinguishable if for any $\poly{n}$-time adversary $\A$ with oracle access who makes $t=\poly{n}$ queries,
	\[
		\abs{
			\prob{U\sim \D_1}{ \A^{U} \text{ outputs } 1.}
			- \prob{V\sim \D_2}{ \A^{V} \text{ outputs } 1.}
		} = \negl{n} \enspace .
	\]
\end{definition}

\subsection{Relation States}
For $t\in \N$, $\RR_t$ represents the set of
all size-$t$ relations
$R=\st{(x_1,y_1),\dots,(x_t,y_t)} \subseteq \st{0,1}^n\times \st{0,1}^n$.
We allow the relations to be a multiset.
And let $\RR \coloneq \cup_{t=0}^{N} \RR_t$ be the set of
all relations with size at most $N$.
For $R\in\RR_t$,
the corresponding \emph{relation state} is defined by
	\begin{equation}
		\ket{R}\reg{XY} \coloneq \frac{1}{\sqrt{\gamma_R}} \sum_{\sigma\in \S_t}
		S_\sigma \ket{x_1,\dots,x_t}\reg{X}
		\otimes S_\sigma \ket{y_1, \dots, y_t}\reg{Y} \enspace,
        \label{eq:rstatedef}
	\end{equation}
	where $S_\sigma$ is a permutation operator on $\br{\mathbb{C}^{2^n}}^{\otimes t}$ defined as\[S_\sigma:\ket{x_1,\dots,x_t}\mapsto\ket{x_{\sigma^{-1}(1)},\dots,x_{\sigma^{-1}(t)}}\]
and the normalizer is given by
	$		
		\gamma_R \coloneq t!\cdot \sum_{x,y\in\bit{n}}
		\br{ \sum_{i=1}^t \delta_{(x_i,y_i) = (x,y)} }!.
 	$
\begin{fact}
	$\st{\ket{R}\reg{XY}}_{R\in\RR}$ forms an orthogonal basis.
\end{fact}
We further introduce additional notation for later use.
Let \yDBRt be the set of all size-$t$ relations
$R=\set{(x_1,y_1),\dots,(x_t,y_t)}$ such that
$(y_1,\dots,y_t)\in\newDB_t$.
Let $\yDBR:=\cup_{t=0}^N\yDBRt$.
Let \DBRt be the set of all size-$t$ relations
$R=\{(x_1,y_1),\dots,(x_t,$
$y_t)\}$ such that
both $(x_1,\dots,x_t)$ and $(y_1,\dots,y_t)$
are in $\newDB_t$.
Let $\DBR:=\cup_{t=0}^N\DBRt$.
For any relation $R$, define 
\begin{align*}
\DomR&:=\set{x:\exists y\in\bit{n}\text{ s.t. }(x,y)\in R}\,,\\
\ImR&:=\set{y:\exists x\in\bit{n}\text{ s.t. }(x,y)\in R}\,,\\
\BDomR&:=\set{x:\exists y\in\bit{n}\text{ s.t. }(x,y)\in R\text{ or }(\bar{x},y)\in R}\,,\\
\BImR&:=\set{y:\exists x\in\bit{n}\text{ s.t. }(x,y)\in R\text{ or }(x,\bar{y})\in R}\,.
\end{align*}

\subsection{Cryptography}
In this section, we will review various definitions and results in
cryptography. Throughout this work, $\secpar$ denotes a security
parameter.

\begin{definition}[Quantum-Secure Pseudorandom Function]
  Let $\K,\X$ and $\Y$ be the key space, the domain and range, all
  implicitly depending on the security parameter $\lambda$. A keyed
  family of functions $\set{\prf_k:\X\to\Y}_{k\in\K}$ is a
  quantum-secure pseudorandom function (\qprf) if the following two
  conditions hold:
\begin{enumerate}
\item \textbf{Efficient generation}. $\prf_k$ is polynomial-time computable on a classical computer.

\item \textbf{Pseudorandomness}. For any polynomial-time quantum oracle algorithm $\A$, $\prf_k$ with a random $k\gets\K$ is indistinguishable from a truly random function $f\gets\Y^\X$ in the sense that:
\[\abs{\prob{k\gets\K}{\A^{\prf_k}\!\br{1^\lambda}=1}-\prob{f\gets\Y^\X}{\A^{f}\!\br{1^\lambda}=1}}=\negl{\lambda} \enspace .\]
\end{enumerate}
\end{definition}

\begin{definition}[Quantum-Secure Pseudorandom Permutation]
  Let $\K$ be the key space, and $\X$ be both the domain and range,
  implicitly depending on the security parameter $\lambda$. A keyed
  family of permutations $\set{\prp_k\in \S_\X}_{k\in\K}$ is a
  quantum-secure pseudorandom permutation (\qprp) if the following two
  conditions hold:
\begin{enumerate}
\item (\textbf{Efficient generation}). $\prp_k$ and $\prp_k^{-1}$ are polynomial-time
  computable on a classical computer.
\item (\textbf{Pseudorandomness}). For any polynomial-time quantum
  oracle algorithm $\A$, $\prp_k$ with a random $k\gets\K$ is
  indistinguishable from a truly random permutation $\sigma\gets \S_\X$
  in the sense that:
\[\abs{\prob{k\gets\K}{\A^{\prp_k,\prp_k^{-1}}\!\br{1^\lambda}=1}-\prob{\sigma\gets S_\X}{\A^{\sigma,\sigma^{-1}}\!\br{1^\lambda}=1}}=\negl{\lambda} \enspace .\]
\end{enumerate}
\end{definition}

Under the assumption that post-quantum one-way functions exist,
Zhandry proved the existence of $\qprf$s~\cite{Zhandry21_qprf}. $\qprp$s can be
constructed from $\qprf$s efficiently \cite{Songblog17,Zhandry16_qprp}.

\begin{definition}[Random State Scrambler~\cite{LQSY+24}]
	For $n\in \N$, let $\D$ be a distribution over $\ugroup{N}$.
	We say $\D$ is a \emph{random state scrambler distribution} with error $\epsilon$ if for any $n$-qubit pure state $\ket{\phi}$ and $\ncopy\in \poly{n}$, 
	\[ \td\br{\expect{K\gets \D}{K^{\otimes\ncopy}\ketbra{\phi}^{\otimes \ncopy} K^{\otimes\ncopy, \dagger}},
        \expect{\ket{\psi}\in \haar}{ \ketbra{\psi}^{\otimes \ncopy}}}
      \leq \epsilon  \enspace. \]
We call a random unitary $K$ sampled from such distribution an $\epsilon$-random state scrambler ($\epsilon$-$\rss$).
\end{definition}

\begin{definition}[Pseudorandom Unitary Operator]
	For $n\in \N$, let $\D$ be a distribution over $\ugroup{N}$.
	We say $\D$ is a \emph{pseudorandom unitary distribution} ($\pru$) if 
	\begin{itemize}
		\item $\D$ can be sampled in $\poly{n}$ time;
		\item $\D$ is computationally indistinguishable from Haar random over $\ugroup{N}$.
	\end{itemize} 
\end{definition}

\section{Projecting into Distinct Block Subspace}
\label{sec:DBproj}
This section shows that applying an $\rss$ operator 
followed by a random permutation to any state results in a state
that is mostly within the \emph{distinct block subspace}. 
We define the projector onto the distinct block subspace as follows.
\begin{definition}[Distinct Block subspaces]
Let $0\leq t\leq N$, and let
\[\DBp{}:=\sum_{(x_1,\dots,x_t)\in\DB}\ketbra{x_1,\dots,x_t}\enspace.\]
\end{definition}

\begin{lemma}\label{lem:DBp}
    Let $0\leq t\leq N$,
    and let $\rssd$ be a distribution over $\ugroup{N}$ such that
    for $K\gets\rssd$, $K^{\dagger}$ is an
    $\epsilon$-$\rss$ with $\epsilon=O\br{\frac{1}{N^2}}$.
    Define $\D$ to be the distribution that samples $G=PK$ where
	$P\gets \S_N$ is a random permutation unitary and $K\gets\rssd$.
	For any state $\rho$ on an $nt$-qubit register $\mathsf{X} = \br{\mathsf{X}_1,...,\mathsf{X}_t}$ together with another arbitrary register,
    we have
	\[\expect{G\gets\D}{\Tr\br{\DBp{X}\cdot G^{\otimes t}\reg{X} \rho G^{\otimes t,\dagger}\reg{X}}}\geq 1-O\!\br{\frac{t^2}{N}} \enspace.\]
\end{lemma}

\begin{proof}
Using the cyclic property of trace and then the definition of $\DBp{}$,
\begin{align}
	\expect{G\gets\D}{\Tr\br{\DBp{X}\cdot G^{\otimes t}\reg{X} \rho G^{\otimes t,\dagger}\reg{X}}}
	= & ~ \Tr\br{ \expect{G\gets\D}{
		G^{\otimes t,\dagger}\reg{X} \cdot 
		\DBproj \cdot 
		G^{\otimes t}\reg{X} \cdot
		\rho }}\nonumber\\
	= & ~ \Tr\br{ \expect{G\gets\D}{
		G^{\otimes t,\dagger}\reg{X} \cdot 
		\sum_{x\in\DB}\ketbra{x}\reg{X} \cdot 
		G^{\otimes t}\reg{X} \cdot
		\rho }} \enspace.\label{DBproj_eq1}
\end{align}
Now define $\rho\reg{X_1, \dots, X_t} := \Tr\reg{-(X_1, \dots, X_t)} (\rho)$,
where $\Tr\reg{-(X_1, \dots, X_t)} $ represents tracing out all registers except ${\mathsf{X}_1, \dots, \mathsf{X}_t}$. 
Then \cref{DBproj_eq1} can be rewritten as
\begin{equation*}
	 \Tr\br{ \expect{G\gets\D}{
		G^{\otimes t,\dagger} \cdot 
		\sum_{x\in\DB}\ketbra{x} \cdot 
		G^{\otimes t} \cdot
		\rho\reg{X_1, \dots, X_t} }}\enspace.
\end{equation*}
This implies:
\begin{align}
	1- &\expect{G\gets\D}{\Tr\br{\DBp{X}\cdot G^{\otimes t}\reg{X} \rho G^{\otimes t,\dagger}\reg{X}}}\nonumber\\
	&\qquad=~
	\Tr\br{ \expect{G\gets\D}{
			G^{\tp t, \dag} \cdot
			\sum_{x\in \bit{nt} \setminus \DB}\ketbra{x} \cdot
			G^{\tp t} \cdot
			\rho\reg{X_1, \dots, X_t}}}\enspace.
	\label{DBproj_eq2}
\end{align}
We can see that the set  $\bit{nt} \setminus \DB$ includes all $t$-tuples where (a) $\exists i\neq j$ s.t. $x_i=x_j$
or (b) $\exists \i\neq j$ s.t. $x_i=\overline{x}_j$ where $\overline{x}_j$ flips the first bit of $x_j$. We further define two 
projectors to capture these two situations:
\[\eqProj = \sum_{x\in\bit{n}}\ketbra{x} \tp \ketbra{x}\enspace, \qquad\ffbProj = \sum_{x\in\bit{n}}\ketbra{x} \tp \ketbra{\bar{x}}\enspace.\]
And we have
\begin{equation}
	\sum_{x\in \bit{nt}\setminus \DB} \ketbra{x}\reg{X_1, \dots, X_t}
	\preceq 
	\sum_{1\leq i<j\leq t} \eqProj\reg{X_i, X_j}+\ffbProj\reg{X_i, X_j}\enspace,
	\label{DBproj_eq3}
\end{equation}
where $\preceq$ represents the positive semidefinite order; $\eqProj\reg{X_i, X_j}$ is the equality projector on registers $X_i$ and $X_j$;
and $\ffbProj\reg{X_i, X_j}$ is the flip-first-bit projector on registers $X_i$ and $X_j$.
Combining \cref{DBproj_eq2} and the inequality \cref{DBproj_eq3}:
\begin{align}
&1- \expect{G\gets\D}{\Tr\br{\DBp{X}\cdot G^{\otimes t}\reg{X} \rho G^{\otimes t,\dagger}\reg{X}}}\nonumber\\
	\leq &  \sum_{1\leq i<j\leq t} \expect{G\gets \D}{\Tr\br{G^{\tp t, \dagger} \cdot 
		\br{\eqProj\reg{X_i, X_j}+\ffbProj\reg{X_i, X_j}} \cdot 
		G^{\tp t}\cdot 
		\rho\reg{X_1, \dots, X_t}}}
	\nonumber\\
	= & \sum_{1\leq i<j\leq t} \expect{G\gets \D}{\Tr\br{\br{G^\dagger\reg{X_i}\tp G^\dagger\reg{X_j} }\cdot 
		\br{\eqProj\reg{X_i, X_j}+\ffbProj\reg{X_i, X_j}} \cdot  
		\br{G\reg{X_i}\tp G\reg{X_j}} \cdot
		\rho\reg{X_i, X_j}}}\nonumber
			\nonumber\\
	= & \sum_{1\leq i<j\leq t}
	\expect{P\gets \S_N \\ K\gets \rssd}{
	\Tr\br{
	\begin{aligned}
	&\br{(PK)^\dagger\reg{X_i}\tp (PK)^\dagger\reg{X_j}}
	\cdot \br{\eqProj\reg{X_i,X_j}+\ffbProj\reg{X_i,X_j}}\\
	&\qquad {}\cdot
	\br{PK\reg{X_i}\tp PK\reg{X_j}}
	\cdot \rho\reg{X_i,X_j}
	\end{aligned}
	}}\enspace,
	\label{DBproj_eq4}
	\end{align}
where the part inside of the summation of \cref{DBproj_eq4}
can be rewritten as the sum of the following two terms:
\begin{align}
	\label{DBproj_eq5}
	\expect{P\gets \S_N \\ K\gets \rssd}{\Tr\br{\br{(PK)^\dagger\reg{X_i}\tp (PK)^\dagger\reg{X_j} }\cdot 
			\eqProj\reg{X_i, X_j} \cdot  
			\br{PK\reg{X_i}\tp PK\reg{X_j}} \cdot
			\rho\reg{X_i, X_j}}}\enspace,
					\\
	\label{DBproj_eq6}
	\expect{P\gets \S_N \\ K\gets \rssd}{\Tr\br{\br{(PK)^\dagger\reg{X_i}\tp (PK)^\dagger\reg{X_j} }\cdot 
			\ffbProj\reg{X_i, X_j} \cdot  
			\br{PK\reg{X_i}\tp PK\reg{X_j}} \cdot
			\rho\reg{X_i, X_j}}}\enspace.
\end{align}
We will bound these two terms one by one. First, note that
for any permutation matrix $P$, $(P^\dag\otimes P^\dag)\cdot \eqProj\cdot (P\otimes P) = \eqProj$.
Therefore, we have
{\small\begin{align}
	(\ref{DBproj_eq5})
	    = & \expect{P\gets \S_N \\ K\gets \rssd}{\Tr\br{
	    		\br{K^\dagger\reg{X_i}\tp K^\dagger\reg{X_j} }
	    		\br{P^\dagger\reg{X_i}\tp P^\dagger\reg{X_j} } \cdot 
	    		\eqProj\reg{X_i, X_j} \cdot  
	    		\br{P\reg{X_i}\tp P\reg{X_j} }
	    		\br{K\reg{X_i}\tp K\reg{X_j} } \cdot
	    		\rho\reg{X_i, X_j}}} \nonumber
	    \\
		= & \expect{K\gets \rssd}{\Tr\br{
				\br{K^\dagger\reg{X_i}\tp K^\dagger\reg{X_j} } \cdot 
				\eqProj\reg{X_i, X_j} \cdot  
				\br{K\reg{X_i}\tp K\reg{X_j} } \cdot
				\rho\reg{X_1, \dots, X_t}}}\enspace. \nonumber
\end{align}}
Then, since $\sigma\reg{X_1, \dots, X_t}$ is a density operator, we have
\begin{align}
	(\ref{DBproj_eq5})
		\leq & \norm{
			\expect{K\gets \rssd}{
					\br{K^\dagger\reg{X_i}\tp K^\dagger\reg{X_j} } \cdot 
					\eqProj\reg{X_i, X_j} \cdot  
					\br{K\reg{X_i}\tp K\reg{X_j}}}
		}_\infty\enspace. \nonumber 
\end{align}
By the triangle inequality, we have
\begin{align*}
	(\ref{DBproj_eq5})
		\leq & \norm{
			\expect{U\gets \mu}{
				\br{U^\dagger\reg{X_i}\tp U^\dagger\reg{X_j} } \cdot 
				\eqProj\reg{X_i, X_j} \cdot  
				\br{U\reg{X_i}\tp U\reg{X_j}}}
			}_\infty  
			\\
			& +
			\Bigg\lVert 
			 \expect{U\gets \mu}{
					\br{U^\dagger\reg{X_i}\tp U^\dagger\reg{X_j} } \cdot 
					\eqProj\reg{X_i, X_j} \cdot  
					\br{U\reg{X_i}\tp U\reg{X_j}}} 
			\\
			& \qquad \qquad\qquad\qquad {}-
			 \expect{K\gets \rssd}{
			 	\br{K^\dagger\reg{X_i}\tp K^\dagger\reg{X_j} } \cdot 
			 	\eqProj\reg{X_i, X_j} \cdot  
			 	\br{K\reg{X_i}\tp K\reg{X_j}}}				
			\Bigg\rVert_\infty  
			\\
		\leq & \sum_{ x\in \bit{n}}
		\biggl(
		 \norm{
			\expect{U\gets \mu}{
				\br{U^\dagger\tp U^\dagger } \cdot 
				\ketbra{x,x}\cdot  
				\br{U\tp U}}
		}_\infty 
		\\
		& +
		\Bigg\lVert
			\expect{U\gets \mu}{
				\br{U^\dagger\tp U^\dagger } \cdot 
				\ketbra{x,x}\cdot  
				\br{U\tp U}} 
		\\
		& \qquad \qquad\qquad\qquad {}-
			\expect{K\gets \rssd}{
				\br{K^\dagger\tp K^\dagger} \cdot 
			    \ketbra{x,x}\cdot  
				\br{K\tp K}}				
		\Bigg\rVert_\infty\biggr) 
		\\
		\leq & \sum_{x\in \bit{n}}
		\br{
		\norm{
			\expect{U\gets \mu}{
				\br{U^\dagger\tp U^\dagger} \cdot 
				\ketbra{x,x}\cdot  
				\br{U\tp U}}
		}_\infty + O\!\br{\frac{1}{N^2}}} \enspace,
\end{align*}
where we use the property of $O\br{\frac{1}{N^2}}$-$\rss$ to derive
the last inequality.
Since $U\ket{x}$ is a Haar random state,
we have $ \expect{U\gets \mu}{
				\br{U^\dagger\tp U^\dagger} \cdot 
				\ketbra{x,x}\cdot  
				\br{U\tp U}} = \expect{\ket{\psi}\gets \mu}{
				\ketbra{\psi, \psi}}$
and the operator norm is $\frac{2}{N(N+1)}$ \cite[Proposition 6]{Har13}.
Thus, we have
\begin{align}
		(\ref{DBproj_eq5})\leq &  \sum_{x\in \bit{n}}
		\br{
		\norm{
			\expect{\ket{\psi}\gets \mu}{
				\ketbra{\psi, \psi}}
		}_\infty + O\!\br{\frac{1}{N^2}}} \nonumber
		\\
		\leq &~ N \cdot \br{\frac{2}{N(N+1)}+O\!\br{\frac{1}{N^2}}}= O\!\br{N^{-1}}\enspace.\nonumber
\end{align}
Next, we attempt to bound (\ref{DBproj_eq6}):
\begin{align}
(\ref{DBproj_eq6})
& =\expect{P\gets \S_N \\ K\gets \rssd}{\Tr\br{
		\br{K^\dagger\reg{X_i}\tp K^\dagger\reg{X_j} }
		\br{P^\dagger\reg{X_i}\tp P^\dagger\reg{X_j} } 
		\ffbProj\reg{X_i, X_j}    
		\br{P\reg{X_i}\tp P\reg{X_j} }
		\br{K\reg{X_i}\tp K\reg{X_j} } 
		\rho\reg{X_i,X_j}}} \nonumber
\\
\label{DBproj_eq7}
& \leq \norm{\expect{P\gets \S_N }{\br{P^\dagger\reg{X_i}\tp P^\dagger\reg{X_j} } \cdot 
		\ffbProj\reg{X_i, X_j} \cdot  
		\br{P\reg{X_i}\tp P\reg{X_j} }}}_{\infty} \nonumber\\
&\leq \sum_{x\in\bit{n}}\norm{\expect{P}{(P\otimes P)^\dag \cdot\ketbra{x,\bar{x}}\cdot (P\otimes P)}}_{\infty}  \nonumber\\
&= \sum_{x\in\bit{n}}
		\norm{\frac{1}{N(N-1)} \sum_{z\neq y} \ketbra{z,y}}_\infty \nonumber\\
&=\frac{1}{N-1} \enspace.\nonumber
\end{align}

Since both \cref{DBproj_eq5,DBproj_eq6} are upper bounded by $O\!\br{N^{-1}}$, by substituting into \cref{DBproj_eq4} and using a union bound on all $i$ and $j$, we have
\begin{equation*}
	1- \expect{G\gets\D}{\Tr\br{\DBp{X}\cdot G^{\otimes t}\reg{X} \rho G^{\otimes t,\dagger}\reg{X}}} \leq O\!\br{\frac{t^2}{N}}\enspace. \qedhere
\end{equation*}
\end{proof}

\section{\texorpdfstring{\pru}{PRU} from Parallel Kac's Walk}
\label{sec:construction}
In this section, we introduce our construction for \pru~which is inspired
by parallel Kac's walk.
Our construction is simply to repeat the parallel Kac's walk.

\subsection{The \texorpdfstring{$\HPC_{n,T}$}{HP(n,T)} distribution}
Our construction is based on \emph{parallel Kac's walk},
a random walk on unit vectors within Hilbert spaces.
A single step of parallel Kac's walk can be simulated
by firstly sampling a random function $f:\set{0,1}^{n-1}\to\set{0,1}^{3d}$
and a random permutation $\sigma:\bit{n}\to\bit{n}$,
and then applying two unitary operators $P_\sigma$ and $H_f$ in sequence.
The unitary $P_\sigma$ is the permutation matrix defined by
\[
	P_\sigma=\sum_{x\in\bit{n}} \ketbratwo{\sigma(x)}{x} \enspace.
\]
The unitary $P_\sigma$ will pair the $2^n$ computational basis up
according to their images after the permutation $\sigma$.
More specifically, the basis $\ket{x}$ and $\ket{z}$ are paired up
iff $\sigma(x)$ and $\sigma(z)$ share the same suffix of length $n-1$,
and each pair can be identified by its unique suffix.
The unitary $H_f$ then applies independent $2\times 2$ Haar random unitaries
on all pairs in the following way:
\begin{enumerate}
	\item for every $y\in\bit{n-1}$,
			we first parse $f(y)=f_\alpha(y)\Vert f_\beta(y)\Vert f_\theta(y)$
			such that $f_\alpha(y), $ $f_\beta(y),$ $f_\theta(y)\in\bit{d}$,
	\item calculate three angles
	\[
		\theta_y=\arcsin\br{\sqrt{\val{f_\theta(y)}}}\enspace, \enspace
		\alpha_y=2\pi\cdot\val{f_\alpha(y)}\enspace , \enspace
		\beta_y=2\pi\cdot\val{f_\beta(y)}\enspace ,
	\]
	\item apply $U(\alpha_y,\beta_y,\theta_y)=
\begin{bmatrix}
	e^{\i\alpha_y }\cos(\theta_y ) & -e^{\i\beta_y }\sin(\theta_y)\\
	e^{-\i\beta_y }\sin(\theta_y ) & e^{-\i\alpha_y }\cos(\theta_y)
\end{bmatrix}$ on the pair with suffix $y$.
\end{enumerate}
The expression for $H_f$ is
	{\small\begin{align}\label{eq:widehatq}
	\sum_{y\in\bit{n-1}}\!\!
	\left(
		\begin{matrix}
		\e^{\i\br{\frac{\alpha_y+\beta_y}{2}} } & 0\\
		0 & \e^{-\i\br{\frac{\alpha_y+\beta_y}{2}}}
		\end{matrix}
	\right)\!\!
	\begin{pmatrix}
	\cos\theta_{y}&-\sin\theta_{y}\\
	\sin\theta_{y}&\cos\theta_{y}
	\end{pmatrix}\!\!
	\left(
		\begin{matrix}
		\e^{\i\br{\frac{\alpha_y-\beta_y}{2}} } & 0\\
		0 & \e^{-\i\br{\frac{\alpha_y-\beta_y}{2}}}
		\end{matrix}
		\right)\!\!
		\otimes\ketbra{y} \,,
	\end{align}}

\noindent where $U(\alpha_y,\beta_y,\theta_y)$ is decomposed into a product of three matrices.
The unitary $H_f$ can be approximated
by a polynomial-time implementable unitary $\widehat{H}_f$
satisfying that $\Vert H_f - \widehat{H}_f\Vert_\infty$ is sufficiently small \cite{LQSY+24}.

Our construction for \pru~is simply to repeat the parallel Kac's walk.
We define two distributions over $\ugroup{N}$, denoted by $\HPC_{n,T}$ and $\hHPC_{n,T}$:

\begin{definition}\label{def:HPnT}
	$\HPC_{n,T}$ is a distribution over $\ugroup{N}$ which can be sampled via the following procedure:
	\begin{itemize}
	\item sample $T$ uniformly random functions
		  $f_1,\dots,f_T:\set{0,1}^{n-1}\to\set{0,1}^{3d}$,
          and $T$ uniformly random permutations $\sigma_1,\dots,\sigma_T:\bit{n}\to\bit{n}$,
    \item output the unitary $\kac= \prod_{i=1}^{T} \br{H_{f_i}\cdot P_{\sigma_i}}$.
\end{itemize}

Similarly, we define the distribution $\hHPC_{n,T}$ by substituting $H_f$ with $\widehat{H}_f$, and we denote the unitary sampled according to $\hHPC_{n,T}$ as $\hkac$.
\end{definition}
These two distributions are indistinguishable by any polynomial-time quantum adversary.
\begin{lemma}\label{lem:close}
	For $T = \poly{n}$ and $d = 5n$, $\hHPC_{n,T}$~is computationally indistinguishable from $\HPC_{n,T}$.
\end{lemma}

\begin{proof}
	Consider the views $\ket{\A_t^{\kac}}\reg{AB}$ and $\ket{\A_t^{\hkac}}\reg{AB}$ of a $t$-query adversary $\A$ with oracle access to $\kac$ and $\hkac$, respectively.
	It is sufficient to show that the trace distance between these two states is negligible.
	To this end, we define the following hybrids: for $0\leq j\leq t$,
	\[
		\ket{\varphi_j} = \prod_{k=1}^j
		\br{\hkac\reg{A}\cdot A^{(i)}\reg{AB}}
		\prod_{j=k+1}^t \cdot
		\br{\kac\reg{A}\cdot A^{(i)}\reg{AB}}
		\ket{0}\reg{AB} \enspace.
	\]
	It is evident that
	$\ket{\varphi_0} = \ket{\A_t^{\kac}}\reg{AB}$ and
	$\ket{\varphi_t} = \ket{\A_t^{\hkac}}\reg{AB}$, and
	the trace distance between two adjacent hybrids is bounded by
	\begin{align*}
		\norm{\ketbra{\varphi_j} - \ketbra{\varphi_{j+1}} }_1 &\leq
		\left\Vert \prod_{i=1}^{T} \br{H_{f_i}\cdot P_{\sigma_i}} - \prod_{i=1}^{T} \br{\widehat{H}_{f_i}\cdot P_{\sigma_i}}\right\Vert_\infty \\
		&\leq\sum_{i=1}^{T} \Vert H_{f_i} - \widehat{H}_{f_i}\Vert_\infty
		\leq 8\pi\cdot 2^{-d}\cdot T =\negl{n}\enspace,
	\end{align*}
	where the second inequality is from the triangle inequality and the last inequality is from Lemma 19 in \cite{LQSY+24}.
	Thus, we have that by the triangle inequality,
	\begin{align*}
		&\norm{
			\expect{\kac\leftarrow\HPC_{n,T}}{\ketbra{ \A_t^{\kac} }}
			- \expect{\hkac\leftarrow\hHPC_{n,T}}{\ketbra{\A_t^{\hkac}}}
		}_1 \\
		\leq~ & \sum_{j=0}^{t-1} \norm{
			\expect{ }{\ketbra{ \varphi_j }}
			- \expect{ }{\ketbra{\varphi_{j+1}}}
		}_1
		\leq t\cdot 8\pi\cdot 2^{-d} \cdot T = \negl{n}\enspace. \qedhere
	\end{align*}
\end{proof}

It is proved in \cite{LQSY+24} that with large enough $T$ and $d$, the
distribution $\HPC_{n,T}$ is an $\rss$ distribution. Formally, we have
the following theorem by adjusting the parameters used in \cite[Theorem 10]{LQSY+24}.

\begin{theorem} \label{thm:rssproperty}
	Let $T = 30n$ and $d = 5n$.
    Then the distribution $\HPC_{n,T}$~is an $\epsilon$-$\rss$ distribution on $\ugroup{N}$ with $\epsilon = O\!\br{\frac{1}{N^2}}$.
    Moreover, if $\kac\gets\HPC_{n,T}$, then $\kac^\dagger$ is also an $\epsilon$-$\rss$ with $\epsilon = O\!\br{\frac{1}{N^2}}$.
    \footnote{The second statement follows from the observation that $\kac^\dagger$ contains enough steps of the parallel Kac’s walk.}
\end{theorem}

In this work, we will prove that adding one more step of parallel Kac's walk results in a distribution that is close to Haar random.
Our main result is:
\begin{theorem}\label{thm:pru}
For $T = 30n$ and $d = 5n$, $\HPC_{n,T+1}$~is computationally indistinguishable from Haar distribution.
\end{theorem}

We view the procedure of sampling a unitary from $\HPC_{n,T+1}$ as three stages:
\begin{enumerate}
	\item sample a unitary operator $\kac$ from $\HPC_{n,T}$;
	\item sample a unitary operator $H_{f}\cdot P_{\sigma}$ from $\HPC_{n,1}$ corresponding to a random permutation $\sigma\in S_N$ and a random function $f:\bit{n-1}\to\bit{3d}$;
	\item output $H_f\cdot P_{\sigma}\cdot\kac$.
\end{enumerate}
The proof of this theorem then consists of three steps:
\begin{itemize}
	\item
	We begin by establishing a purification of the adversary's view state when making queries to the random unitary $H_f\cdot P_{\sigma}$.
	Specifically, we introduce two environment registers,
	$\mathsf{H}$ and $\mathsf{P}$,
	which record the function $f$ and the permutation $\sigma$ being utilized.
	From the adversary's perspective,
	making queries to $H_f\cdot P_{\sigma}$ is equivalent to
	querying a \emph{purified function-permutation oracle} $\HPO$
	that acts on both the adversary's registers and the environment registers.
	\item Next,
	we demonstrate that
	as long as $G$ is a random unitary
	sampled from an $\rss$ distribution,
	the adversary cannot distinguish between
	making queries to $\HPO\cdot G$
	and making queries to a path-recording oracle $\PR$,
	where $\PR$ is defined in \defi{simple_V}.
	The key insight in this step is that
	if the environment state resides in the distinct block subspace,
	we can identify an isometry acting on the environment
	that connects the behaviors of $\HPO$ and $\PR$.
	The random unitary $G$ ensures that
	the environment state is nearly within the distinct block subspace.
	\item
	Let $U$ be sampled from Haar distribution $\mu$.
	Since both $\HPC_{n,T}$ and $\mu$ are $\rss$ distributions,
	we can conclude that
	oracle $\HPO\cdot\kac$ and $\HPO\cdot U$ are both indistinguishable
	from the oracle $\PR$
	to the adversary.
	Thus, $\HPO\cdot\kac$ and $\HPO\cdot U$ are indistinguishable from each other.
	This implies that making queries to $H_f\cdot P_{\sigma}\cdot\kac$
	and $H_f\cdot P_{\sigma}\cdot U$ are indistinguishable as well.
	These two unitary operators correspond to $\HPC_{n,T+1}$ and Haar distribution,
	respectively.
\end{itemize}

In the following, we will introduce the purified function-permutation oracle $\HPO$ in \sct{hpo_intro} and the path-recording oracle in \sct{pro}.
We then explain how to connect the actions of $\HPO$ and $\PR$
by introducing an isometry $\CPS$ in \sct{isometry}.
Lastly, we prove the main result in \sct{mainproof}.

\subsection{The Purified Function-Permutation Oracle} \label{sec:hpo_intro}
To analyze the behavior of querying $H_f\cdot P_{\sigma}$, we employ a  purified oracle.
\begin{definition}[Purified Function-Permutation Oracle]\label{def:HPO}
The \emph{purified function-permutation oracle} $\HPO$ is a unitary
on registers $\mathsf{A}$, $\mathsf{H}$ and $\mathsf{P}$, where
\begin{itemize}
	\item $\mathsf{H}$ is a register with Hilbert space $\H\reg{H}$ spanned by the orthogonal states $\ket{f}$ for all $f\in\bit{n-1}\to\bit{3d}$, 
	\item $\mathsf{P}$ is a register with Hilbert space $\H\reg{P}$ spanned by the orthogonal states $\ket{\sigma}$ for all $\sigma\in\S_N$. 
\end{itemize}
The unitary operator $\HPO$ acts as
\[
	\HPO\reg{AHP} \ket{x}\reg{A}\ket{f}\reg{H}\ket{\sigma}\reg{P} \coloneq
	\br{H_f P_{\sigma}}\reg{A}\ket{x}\reg{A}\ket{f}\reg{H}\ket{\sigma}\reg{P} =
	{H_f}\reg{A} \ket{\sigma(x)}\reg{A}\ket{f}\reg{H}\ket{\sigma}\reg{P} \enspace .
\]

\end{definition}

In the adversary's view, querying $H_f\cdot P_{\sigma}$ and querying $\HPO$ are identical in the following sense:

\begin{fact} \label{fact:puri}
	For any adversary $\A$ holding the register $\mathsf{A}$, the following two oracles are perfectly indistinguishable:
	\begin{itemize}
		\item Sample uniformly random $\sigma\in S_N$ and $f:\bit{n-1}\to\bit{3d}$.
			  On each query, apply $H_fP_{\sigma}$ on register $\mathsf{A}$.
		\item Initialize registers $\mathsf{H}$ and $\mathsf{P}$ in the state
			  \[
			  	\ket{\phi_{\{\}}}\coloneq
			  	\frac{1}{\sqrt{2^{3d(n-1)}}}\sum_{f:\bit{n-1}\to\bit{3d}} \ket{f}\reg{H}
			  	\otimes \frac{1}{\sqrt{N!}} \sum_{\sigma\in S_N} \ket{\sigma}\reg{P} \enspace.
			  \]
			  On each query, apply $\HPO$ on registers $\mathsf{A}$, $\mathsf{H}$ and $\mathsf{P}$.
	\end{itemize}
\end{fact}

Consider the view of an adversary $\A$ who makes $t$ queries to
the oracle $H_f\cdot P_{\sigma}\cdot G$ where $G$ is an arbitrary unitary operator:
\[
	\rho_0\coloneq \expect{H_f P_{\sigma} }{\ketbra{\A_t^{H_f P_{\sigma} G}}\reg{AB}}
\]
where the view state is
\[
	\ket{\A_t^{H_f P_{\sigma} G}}\reg{AB} =
	\prod_{i=1}^t \br{ \br{H_f P_{\sigma} G} \reg{A} \cdot A\reg{AB}^{(i)} }
	\ket{0}\reg{AB} \enspace,
\]
and the view of the same adversary $\A$ who makes $t$ queries to
the oracle $\HPO$ and the unitary $G$:
\[
	\rho_1\coloneq \Tr\reg{HP}\br{ \ketbra{\A_t^{\HPO\cdot G } }\reg{ABHP} }
\]
where the view state is
\[
	\ket{\A_t^{\HPO\cdot G}}\reg{ABHP} =
	\prod_{i=1}^t \br{ \HPO \reg{AHP} \cdot G\reg{A} \cdot A\reg{AB}^{(i)} }
	\ket{0}\reg{AB} \ket{\phi_{\{\}}}\reg{HP} \enspace,
\]
\fct{puri} states that $\rho_0=\rho_1$. 
This enables us to analyze the purified state instead of the original mixed state.
The action of the purified function-permutation oracle on the purified state can be better understood by introducing $\mathsf{HP}$-relation states on registers $\mathsf{H}$ and $\mathsf{P}$.

\subsubsection{\texorpdfstring{$\mathsf{HP}$}{HP}-Relation States}
We define the following relation states on $\mathsf{H}$ and $\mathsf{P}$.
\begin{definition}[$\mathsf{HP}$-Relation States]\label{def:actofHP}
	For $0\leq t\leq N$ and a size-$t$ relation
	$R=\st{(x_1,y_1),\dots,(x_t,y_t)}\in\RR_t$, we define
	\[
		\ket{\phi_R}\reg{HP}\coloneq
		\frac{1}{\sqrt{2^{3d(n-1)}(N-t)!}}\sum_{f,\sigma}\sum_{b\in\bit{t}}\prod_{i=1}^t\bra{y_i}H_f\ket{y_i^{\oplus b_i}}\delta_{y_i^{\oplus b_i}=\sigma(x_i)}\ket{f}\reg{H}\ket{\sigma}\reg{P}
	\]
	where $x^{\oplus 0}=x$ and $x^{\oplus 1}=\overline{x}$ for any $x\in\bit{n}$, and $\delta_{y=x}$ is an indicator that equals $1$ iff strings $x$ and $y$ are identical in every coordinate.
\end{definition}

For the set $\DBR$ of restricted relations,
the corresponding
relation states form an orthogonal basis,
whose proof is given in \app{orthophiR}.

\begin{lemma}\label{lem:orthophiR}
$\set{\ket{\phi_R}}_{R\in\DBR}$ forms a set of orthogonal vectors.
\end{lemma}

\subsubsection{Action of \texorpdfstring{$\HPO$}{HPO}}
Using the $\mathsf{HP}$-relation states, the action of the $\HPO$ oracle is described by the following lemma.
\begin{lemma} \label{lem:action_hpo}
	For $0\leq t\leq N$, $x\in\bit{n}$ and a size-$t$ relation
	$R=\{(x_1,y_1),\dots,$ $(x_t,y_t)\}\in\RR_t$, we have
	\[
		\HPO\reg{AHP}\ket{x}\reg{A}\ket{\phi_R}\reg{HP} =  
		\frac{1}{\sqrt{N-t}}\sum_{y\in\bit{n}}
		\ket{y}\reg{A} \otimes \ket{\phi_{R\cup\st{x,y}}}\reg{HP}\enspace .
	\]
\end{lemma}

\begin{proof}
	Expanding the definitions, we have
	{\small\begin{align}
		\HPO&\reg{AHP}\ket{x}\reg{A}\ket{\phi_R}\reg{HP} \nonumber\\
		=&\frac{1}{\sqrt{2^{3d(n-1)}(N-t)!}}\sum_{f,\sigma}\sum_{b\in\bit{t}}\prod_{i=1}^t\bra{y_i}H_f\ket{y_i^{\oplus b_i}}\delta_{y_i^{\oplus b_i}=\sigma(x_i)}
		H_f\ket{\sigma(x)}\reg{A}\ket{f}\reg{H}\ket{\sigma}\reg{P} \,. \label{eq:expand}
	\end{align}}

	\noindent Note that for a fixed $f:\bit{n-1}\to\bit{3d}$, the matrix $H_f$ can be expressed as
	\[
		H_f = \sum_{y,y'\in\bit{n}} \bra{y}H_f\ket{y'} \ketbratwo{y}{y'}
	\]
	and $\bra{y}H_f\ket{y'} = 0$ if $y$ and $y'$ are not in the same block.
	So we can write $H_f$ as
	\[
		H_f = \sum_{y\in\bit{n},b\in\bit{}} \bra{y}H_f\ket{y^{\oplus b}} \ketbratwo{y}{y^{\oplus b}} \enspace.
	\]
	And we can write $\ket{\sigma(x)}$ as $\sum_{y\in\bit{n}} \delta_{y=\sigma(x)}\ket{y}$.
	Therefore, we have
	\begin{align*}
		H_f\ket{\sigma(x)} &= \br{\sum_{y\in\bit{n},b\in\bit{}} \bra{y}H_f\ket{y^{\oplus b}} \ketbratwo{y}{y^{\oplus b}}}\cdot\br{\sum_{y\in\bit{n}} \delta_{y=\sigma(x)}\ket{y}}\\
		&= \sum_{y\in\bit{n},b\in\bit{}} \bra{y}H_f\ket{y^{\oplus b}} \delta_{y^{\oplus b}=\sigma(x)} \ket{y} \enspace.
	\end{align*}
	Inserting this into \eq{expand}, we have
	\begin{align*}
		&\HPO\reg{AHP}\ket{x}\reg{A}\ket{\phi_R}\reg{HP} = ~\frac{1}{\sqrt{N-t}}\sum_{y_{t+1}\in\bit{n}} \ket{y_{t+1}}\reg{A} \nonumber\\
		&\quad \otimes\frac{1}{\sqrt{2^{3d(n-1)}(N-t-1)!}}\sum_{f,\sigma}\sum_{b\in\bit{t+1}}\prod_{i=1}^{t+1}\bra{y_i}H_f\lvert y_i^{\oplus b_i}\rangle\delta_{y_i^{\oplus b_i}=\sigma(x_i)}
		\ket{f}\reg{H}\ket{\sigma}\reg{P}\nonumber\\
		&=~ \frac{1}{\sqrt{N-t}}\sum_{y_{t+1}\in\bit{n}}
		\ket{y_{t+1}}\reg{A} \otimes \ket{\phi_{R\cup\st{x,y_{t+1}}}}\reg{HP} \enspace. \qedhere
	\end{align*}
\end{proof}

By expanding $\HPO$, we can then rewrite the view state of an adversary $\A$
with access to
the oracle $\HPO$ and the unitary $G$ in terms of $\mathsf{HP}$-relation states:
\begin{corollary}\label{clr:actofhpo}
	For an arbitrary $n$-qubit unitary operator $G$ and a $t$-query
	adversary $\A$ with query access to $\HPO\cdot G$, define the state after $\A$ has finished all the queries to be
	\[
		\ket{\A_t^{\HPO\cdot G}}\reg{ABHP} \coloneq
		\prod_{i=1}^t \br{ \HPO \reg{AHP} \cdot G\reg{A} \cdot A\reg{AB}^{(i)} }
		\ket{0}\reg{AB} \ket{\phi_{\{\}}}\reg{HP} \enspace.
	\]
	Then we have:
    {\small
        \begin{align*}
        &\ket{\A_t^{\HPO\cdot G}}\reg{ABHP} \\
        &\qquad =
		\sqrt{\frac{(N-t)!}{N!}}
		\sum_{
		\substack{(x_1,\dots,x_t)\in\br{\bit{n}}^t\\(y_1,\dots,y_t)\in\br{\bit{n}}^t}
		}
		\prod_{i=1}^t \br{ \ketbratwo{y_i}{x_i}\reg{A} \cdot G\reg{A} \cdot A\reg{AB}^{(i)} }
		\ket{0}\reg{AB} \ket{\phi_{\{ (x_i,y_i) \}_{i=1}^t}}\reg{HP}\enspace.
        \end{align*}
    }
\end{corollary}

\subsection{The Path-Recording Oracle
\texorpdfstring{$\PR$}{PR}}
\label{sec:pro}
We now present the definition of our new path-recording oracle $\PR$,
which is similar to that in \cite{MH25}.
The difference is that each time we query on $x$,
our oracle $\PR$ samples a $y\notin\BImR$
whose block differs from the blocks of all previous $y$'s in $R$.

\begin{definition}[Path-Recording Oracle]\label{def:simple_V}
The path-recording oracle $\PR$ is a linear map
\[\PR:\H_{\mathsf{A}}\otimes\H_{\mathsf{X}}\otimes\H_{\mathsf{Y}}\to\H_{\mathsf{A}}\otimes\H_{\mathsf{X}}\otimes\H_{\mathsf{Y}}\]
defined as follows. For all $x\in\bit{n}$ and $R\in\DBR$,
\[\PR:\ket{x}\reg{A}\ket{R}\reg{XY}\to\frac{1}{\sqrt{N-2\abs{R}}}\sum_{\substack{y\in\bit{n},\\ y\notin\BImR}}\ket{y}\reg{A}\ket{R\cup\set{(x,y)}}\reg{XY} \enspace.\]
\end{definition}

\begin{fact} \label{fact:actofpr}
	For an $n$-qubit unitary operator $G$ and a $t$-query
	adversary $\A$ with query access to $\PR\cdot G$, define the state after $\A$ has finished all the queries to be
	\[
		\ket{\A_t^{\PR\cdot G}}\reg{ABXY}\coloneq\prod_{i=1}^t\br{\PR\cdot G\reg{A}\cdot A^{(i)}\reg{AB}}\ket{0}\reg{AB}\ket{\{\}}\reg{XY} \enspace.
	\]
	Then we have:
    {\small
	\begin{align*}
		&\ket{\A_t^{\PR\cdot G}}\reg{ABXY} \\
        &\quad=
		\sqrt{\prod_{i=0}^{t-1} \frac{1}{(N-2i)} }
		\sum_{
			\substack{(x_1,\dots,x_t)\in\br{\bit{n}}^t\\
				(y_1,\dots,y_t)\in\DB}
		}
		\prod_{i=1}^t \br{ \ketbratwo{y_i}{x_i}\reg{A} \cdot G\reg{A} \cdot A\reg{AB}^{(i)} }
		\ket{0}\reg{AB} \ket{{\{ (x_i,y_i) \}_{i=1}^t}}\reg{XY} \enspace.
	\end{align*}
    }
\end{fact}

Similar to \cite[Lemma 4.3]{MH25}, our new path-recording oracle has the \emph{right invariance} property, whose proof is given in \app{right_inv}.
\begin{lemma} [Right invariance property]\label{lem:right_inv}
	For an arbitrary $n$-qubit unitary operator $G$, we have that
	\[\ket{\A_t^{\PR\cdot G}}\reg{ABXY} =
		\br{G\reg{X_1}\tp \dots \tp G\reg{X_t}}\cdot
		\ket{\A_t^{\PR}}\reg{ABXY} \enspace.\]
\end{lemma}
 
Moreover, if the oracle $\PR$ is preceded by a random scrambler and a random permutation, then the adversary’s record lies almost entirely in the desired \emph{distinct block subspace}.

\begin{lemma}\label{lem:DBproj}
        	    Let $\rssd$ be a distribution over $\ugroup{N}$ such that
    for $K\gets\rssd$, $K^{\dagger}$ is an
    $\epsilon$-$\rss$ with $\epsilon=O\br{\frac{1}{N^2}}$.
    Define $\D$ to be the distribution that samples $G=PK$ where
	$P\gets \S_N$ is a random permutation unitary and $K\gets\rssd$.
	Let
	\[\rho^{\D}\coloneq\expect{G\gets\D}{\ketbra{\A_t^{\PR\cdot G}}\reg{ABXY}}\enspace,\] then
	\[\Tr\br{\DBproj\cdot\rho^{\D}\reg{ABXY}}\geq 1-O\!\br{\frac{t^2}{N}}\enspace.\]
\end{lemma}

\begin{proof}
    By \lem{right_inv}, we can rewrite $\rho^{\D}$ as
\begin{align*}
	\rho^{\D}\reg{ABXY}
	&= \expect{G\gets\D}{\br{G\reg{X_1}\tp \dots \tp G\reg{X_t}}\cdot
		\ketbra{A_t^{\PR}}\reg{ABXY} \cdot 
		\br{G\reg{X_1}\tp \dots \tp G\reg{X_t}}^\dag}\enspace.
\end{align*}
Then by \lem{DBp}, we have
\begin{align*}
    &\Tr\br{\DBproj\cdot\rho^{\D}\reg{ABXY}} \\
    =~&
    \Tr\br{\DBproj\cdot  \expect{G\gets\D}{\br{G\reg{X_1}\tp \dots \tp G\reg{X_t}}\cdot
		\ketbra{A_t^{\PR}}\reg{ABXY} \cdot 
		\br{G\reg{X_1}\tp \dots \tp G\reg{X_t}}^\dag}} 
        \\
    \geq ~ & 1-O\!\br{\frac{t^2}{N}}\enspace. \qedhere
\end{align*}
\end{proof}

\subsection{Connecting \texorpdfstring{$\HPO$ and $\PR$ via $\CPS$}{HPO and PR via Compress} Isometry}\label{sec:isometry}

By defining the following isometry,
we are able to connect the behavior of
the purified function-permutation oracle $\HPO$ and
the path-recording oracle $\PR$.

\begin{definition}
	We define an isometry $\CPS:\H\reg{P}\otimes\H\reg{H}\to \H\reg{X}\otimes\H\reg{Y}$ as:
\[\CPS:=\sum_{R\in\DBR}\ketbratwo{R}{\phi_R} \enspace.\]
\end{definition}

\begin{lemma}\label{lem:compress}
Define the projector onto the \emph{distinct block subspace} of $\mathsf{HP}$-relation states by
\[\DBRproj:=\sum_{R\in\DBRt}\ketbra{\phi_R} \enspace.\]
We have that for every $n$-qubit unitary $G$,
\[
\CPS\cdot\DBRproj\cdot\ket{\A_t^{\HPO\cdot G}}\reg{ABHP}=\br{1+O\!\br{\frac{t^2}{N}}}\DBproj\cdot\ket{\A_t^{\PR\cdot G}}\reg{ABXY} \enspace.
\]
\end{lemma}

\begin{proof}
	By \clr{actofhpo}, it is easy to see that
	{\small
    \begin{align*}
		&\CPS\cdot\DBRproj\cdot\ket{\A_t^{\HPO\cdot G}}\reg{ABHP} \nonumber\\
		&\quad  =
		\sqrt{\frac{(N-t)!}{N!}}
		\sum_{
		\substack{(x_1,\dots,x_t)\in\DB\\(y_1,\dots,y_t)\in\DB}	
		}
		\prod_{i=1}^t \br{ \ketbratwo{y_i}{x_i}\reg{A} \cdot G\reg{A} \cdot A\reg{AB}^{(i)} }
		\ket{0}\reg{AB} \ket{{\{ (x_i,y_i) \}_{i=1}^t}}\reg{HP}\enspace .
	\end{align*} 
    }
	
	\noindent By \fct{actofpr}, we have
	{\small
    \begin{align*}
		&\DBproj\cdot\ket{\A_t^{\PR\cdot G}}\reg{ABXY} \nonumber\\
		&\quad =
		 \sqrt{\prod_{i=0}^{t-1} \frac{1}{(N-2i)} }
		\sum_{
		\substack{(x_1,\dots,x_t)\in\DB\\
		(y_1,\dots,y_t)\in\DB}
		}
		\prod_{i=1}^t \br{ \ketbratwo{y_i}{x_i}\reg{A} \cdot G\reg{A} \cdot A\reg{AB}^{(i)} }
		\ket{0}\reg{AB} \ket{{\{ (x_i,y_i) \}_{i=1}^t}}\reg{XY}\enspace .
	\end{align*}
    }
	
	\noindent Therefore, we observe that
	\[
		\CPS\cdot\DBRproj\cdot\ket{\A_t^{\HPO\cdot G}} = \rho\cdot \DBproj\cdot\ket{\A_t^{\PR\cdot G}}
	\]
	where
	\[
		\rho = \sqrt{\prod_{i=0}^{t-1} \frac{N-i}{N-2i}} =
		\sqrt{\prod_{i=0}^{t-1} \br{1+\frac{i}{N-2i}}} = 1+O\!\br{\frac{t^2}{N}} \enspace. \qedhere
	\]
\end{proof}

Using the above lemma, we show that
for a suitably sampled unitary $G$,
an adversary's views before and after applying the projector $\DBRproj$ are close.

\begin{lemma}\label{lem:rho12}
Let $\rssd$ be a distribution over $\ugroup{N}$ such that
    for $K\gets\rssd$, $K^{\dagger}$ is an
    $\epsilon$-$\rss$ with $\epsilon=O\br{\frac{1}{N^2}}$.
Define $\D$ to be a distribution that samples $G=SK$ where
	$S\gets \S_N$ and $K\gets\rssd$.
Let $\A$ be a $t$-query oracle adversary.
Write
\[
	\sigma \coloneq \expect{G\gets \D}{\ketbra{\A_t^{\HPO\cdot G}}\reg{ABHP}}\enspace .
\]
Then we have
\begin{align*}
	\norm{
		\Tr_\mathsf{HP}\!\br{\sigma}
		-\Tr_\mathsf{HP}\!\br{\DBRproj\cdot \sigma\cdot \DBRproj}
	}_1
	= O\!\br{\frac{t^2}{N}} \enspace.
\end{align*}

\end{lemma}

\begin{proof}
We first apply \lem{projdist}:
\begin{align*}
	\norm{\Tr_\mathsf{HP}\!\br{\sigma}
	- \Tr_\mathsf{HP}\!\br{\DBRproj\cdot \sigma\cdot \DBRproj}
	}_1 
	= 1 - \Tr\!\br{\DBRproj\cdot \sigma\cdot\DBRproj}\enspace.
\end{align*}
Notice that $\DBRproj = \CPS^\dagger\cdot \CPS\cdot\DBRproj$.
Therefore,
\begin{align*}
	&\norm{\Tr_\mathsf{HP}\br{\sigma}
	- \Tr_\mathsf{HP}\!\br{\DBRproj\cdot \sigma\cdot \DBRproj}
	}_1 \\
	=\ &1 - \Tr\!\br{ \CPS\cdot\DBRproj\cdot \expect{G\gets \D}{\ketbra{\A_t^{\HPO\cdot G } }\reg{ABHP}}\cdot\DBRproj \cdot\CPS^\dagger}\\
	=\ &1- \br{1+O\!\br{\frac{t^2}{N}}}\Tr\!\br{\DBproj\cdot\expect{G\gets \D}{\ketbra{\A_t^{\PR\cdot G } }\reg{ABXY}}} = O\!\br{\frac{t^2}{N}}\enspace,
\end{align*}
where the second equality uses \lem{compress} and the last uses \lem{DBproj}.
\end{proof}

\subsection{Computational Indistinguishability of \texorpdfstring{$\HPC_{n,T+1}$}{HP(n,T+1)}}\label{sec:mainproof}

We first show that if $G$ consists of an \rss~and a random permutation,
then the view of an adversary when making queries to $H_fP_{\sigma}G$
is nearly the view it will see when making queries to the path-recording oracle $\PR$.
\begin{lemma}\label{lem:inter}
Let $\rssd$ be a distribution over $\ugroup{N}$ such that
    for $K\gets\rssd$, $K^{\dagger}$ is an
    $\epsilon$-$\rss$ with $\epsilon=O\br{\frac{1}{N^2}}$.
Define $\D$ to be a distribution that samples $G=SK$ where
	$S\gets \S_N$ and $K\gets\rssd$.
Let $\A$ be a $t$-query oracle adversary.
Then 
{\small
\[
\TD
{\expect{H_fP_{\sigma}\gets\HPC_{n,1}, G\gets \D}{\ketbra{\A_t^{H_fP_{\sigma}G}}}}
{\Tr_{\mathsf{XY}}\br{\ketbra{\A_t^{\PR}}\reg{ABXY}}}
= O\!\br{\frac{t^2}{N}} \enspace.
\]
}
\end{lemma}

\begin{proof}
We start by defining the following states: 
\begin{align*}
	\rho_0 &= \expect{H_fP_{\sigma}\gets\HPC_{n,1}, G\gets \D}{\ketbra{\A_t^{H_fP_{\sigma}G}}} \\
	\rho_1 &= \Tr_\mathsf{HP}\!\br{ \expect{G\gets \D}{\ketbra{\A_t^{\HPO\cdot G } }\reg{ABHP}} } \\
	\rho_2 &= \Tr_\mathsf{HP}\!\br{ \DBRproj\cdot \expect{G\gets \D}{\ketbra{\A_t^{\HPO\cdot G } }\reg{ABHP}}\cdot \DBRproj }\\
	\rho_3 &= \Tr_\mathsf{XY} \!\br{ \DBproj\cdot \expect{G\gets \D}{\ketbra{\A_t^{\PR\cdot G}}\reg{ABXY}}\cdot \DBproj }\\
	\rho_4 &= \Tr_{\mathsf{XY}}\!\br{\expect{G\gets \D}{\ketbra{\A_t^{\PR\cdot G}}\reg{ABXY}}} \qquad
	\rho_5 = \Tr_{\mathsf{XY}}\!\br{\ketbra{\A_t^{\PR}}\reg{ABXY}}
\end{align*}

By \fct{puri}, $\rho_0 = \rho_1$.
By \lem{rho12}, $\norm{\rho_1-\rho_2}_1\leq O\!\br{\frac{t^2}{N}}$.
By \lem{compress} and the fact that $\CPS$ acts only on the environment register, we have that $\norm{\rho_2-\rho_3}_1 = O\br{\frac{t^2}{N}}$.
By \lem{DBproj} and \lem{projdist}, $\norm{\rho_3-\rho_4}_1\leq O\!\br{\frac{t^2}{N}}$.
Finally,
{\small
\begin{align*}
	\rho_4 = ~&\expect{G\gets \D}{ \Tr_{\mathsf{XY}}\br{{\ketbra{\A_t^{\PR\cdot G}}\reg{ABXY}}} } \\
	= ~& \expect{G\gets \D}{ \Tr_{\mathsf{XY}}\br{G^{\tp t}\reg{X}\cdot\ketbra{\A_t^{\PR}}\reg{ABXY}\cdot G^{\tp t, \dag}\reg{X}}}
	=~ \Tr_{\mathsf{XY}}\!\br{\ketbra{\A_t^{\PR}}\reg{ABXY}} =~\rho_5 \enspace,
\end{align*}
}

\noindent where the second equality follows from \lem{right_inv}.
Combining the above bounds with the triangle inequality proves the lemma.
\end{proof}

We now prove that the distribution $\HPC_{n,T+1}$~is computationally indistinguishable from Haar distribution.
\begin{proof}[Proof of \thm{pru}]
	Consider an adversary $\A$ who makes $t=\poly{n}$ queries.
	Note that the distribution of $H_f\cdot P_{\sigma}\cdot\kac$
	is the same as the distribution of $H_f\cdot P_\sigma\cdot S \cdot\kac$ where we add a random permutation matrix $S\gets\S_N$.
	Therefore, 
	\[
		\expect{H_fP_\sigma\kac}{ \ketbra{\A_t^{H_fP_\sigma\kac}} } = 
		\expect{H_fP_\sigma S\kac}{ \ketbra{\A_t^{H_fP_\sigma S\kac}} }\enspace .
	\]
	Since $\kac$ is sampled from $\HPC_{n,T}$ and $S\gets\S_N$, by \thm{rssproperty} and \lem{inter},
	\begin{align*}
		\TD{\expect{H_fP_\sigma\kac}{ \ketbra{\A_t^{H_fP_\sigma\kac}} }}{\Tr_{\mathsf{XY}}\br{\ketbra{\A_t^{\PR}}\reg{ABXY}}} = O\!\br{\frac{t^2}{N}} \enspace.
	\end{align*}
	Now substitute $\kac$ with a Haar random unitary $U$ and we have
	\begin{align*}
		\expect{U}{ \ketbra{\A_t^{U}} } = 
		\expect{H_fP_\sigma S U}{ \ketbra{\A_t^{H_fP_\sigma S U}} } \enspace .
	\end{align*}
	Since $U$ is a $0$-$\rss$ and $S\gets\S_N$, we have that by \lem{inter}
	\begin{align*}
		\TD{\expect{U}{ \ketbra{\A_t^{U}} }}{\Tr_{\mathsf{XY}}\br{\ketbra{\A_t^{\PR}}\reg{ABXY}}} = O\!\br{\frac{t^2}{N}} \enspace.
	\end{align*}
	Then by the triangle inequality, we have
	\begin{align*}
		\TD{\expect{U}{ \ketbra{\A_t^{U}} }}{\expect{H_fP_\sigma\kac}{ \ketbra{\A_t^{H_fP_\sigma\kac}} }} = O\!\br{\frac{t^2}{N}} =\negl{n}\enspace. 
	\end{align*}
	This proves the theorem.
\end{proof}

Our construction of $\pru$ based on parallel
Kac's walk is to use \qprf~and \qprp~when sampling from $\hHPC_{n,T+1}$.

\begin{theorem}
	By replacing random functions and random permutations with
	their post-quantum secure pseudorandom counterparts in the sampling procedure of $\hHPC_{n,T+1}$ where $T=30n$ and $d=5n$, we obtain a \pru.
	\end{theorem}
		
	\begin{proof}
	By the post-quantum security of \qprp~and \qprf,
	\lem{close} and \thm{pru}, the new distribution is computationally indistinguishable from Haar distribution.
	This new distribution can be sampled in polynomial time
	since \qprp~and \qprf~can be sampled efficiently.	
\end{proof}

A closer look at our analysis actually reveals a more general and abstract statement. Namely, given a distribution over unitaries, we identify a property which approximates Haar random unitary under a specific twirling operation. As soon as this property holds, composing it with a single step of the parallel Kac's walk would produce a \pru. This readily implies a number of instantiations and can help mechanize the design of \pru. 

\begin{theorem}\label{thm:complier_kac}
    Let $V$ be an $n$-qubit unitary sampled from the distribution $\HPC_{n,1}$ (i.e., a single step of parallel Kac's walk). Let $\eqProj := \sum_{x\in\st{0,1}^n} \ketbra{x,x}$. Suppose a unitary $K$ is sampled from a distribution $\rssd$ over $\ugroup{N}$
    such that 
        \[
            \norm{
		\expect{K}{(K\otimes K)^\dagger \cdot
		\eqProj\cdot (K\otimes {K})}
		}_{\infty} \le \delta \enspace.
        \]
    Then, for any adversary making at most $t$ queries,
    $VK$ is indistinguishable from a Haar random unitary except with advantage at most $O(t^2\delta+t^2/N)$. Some examples of instantiations of $K$ are below.
    \begin{itemize}
        \item $K$ is a $\delta$-approximate unitary $2$-design, or
        \item $K^\dagger$ is an $O(\delta/N)$-\rss. For instance, for $K \gets \HPC_{n,T=30n}$, $K^\dagger$ has the same distribution, and hence satisfies it. 
    \end{itemize}
\end{theorem}
 
\section{Showing the Strong Security of
\texorpdfstring{$\HPC_{n,2T+1}$}{HP(n,2T+1)}}
\label{sec:strong_security}
We further show that our construction, based on Kac's walk, also achieves strong security
when adversaries are granted query access to the inverse unitary.

\begin{theorem}[$\HP_{n, 2T+1}$ is a statistical strong-$\pru$] \label{thm:spru}
	Let $\A$ be a $t$-query oracle adversary
	performing both forward and inverse queries to oracle $\OO$, and let
	$\HP_{n,2T+1}$ be defined as in \defi{HPnT} with $T = 30n$ and $d = 5n$. Then
	\begin{equation*}
		\td\br{
			\expect{\OO\gets \HP_{n, 2T+1}}{\ketbra{\A^{\OO}_t}_{\reg{AB}}},
			\expect{\OO\gets \haar}{\ketbra{\A^{\OO}_t}_{\reg{AB}}}
		}
		\leq \frac{2t(11t+20)}{N^{1/8}} \enspace.
	\end{equation*}
\end{theorem}

This theorem establishes the statistical strong security.
Assuming post-quantum secure $\owf$s,
it then implies the existence of computationally
strong $\pru$s.
The proof of this theorem follows a similar routine to \thm{pru}:
\begin{itemize}
	\item
	We first use the purification to establish that
	making queries to $H_f\cdot P_{\sigma}$ is equivalent to
	querying $\HPO$.
	\item
	We then show that 
	the adversary cannot distinguish between
	making queries to $D\cdot\HPO\cdot C$
	and querying a path-recording oracle $V$ introduced in \sct{new_v}.
	Here $C$ and $D$ are sampled from either $\HP_{n, T}$
	or Haar distribution $\mu$
	to ensure that
	the environment state is mostly within the distinct block subspace.
	We require two unitary operators to achieve this purpose,
	as the adversary can make both forward and inverse queries.
	\item
    Finally, we apply the above indistinguishability result to two cases. 
    If $C,D\gets \HP_{n,T}$, then querying $D\cdot\HPO\cdot C$ is equivalent to querying $\OO\gets \HP_{n,2T+1}$. On the other hand, if $C,D\gets \mu$, then querying $D\cdot\HPO\cdot C$ is equivalent to querying $\OO\gets \mu$.
    Hence, both query cases can be reduced to the
    same path-recording oracle $V$, and \cref{thm:spru} follows from the
    triangle inequality.
\end{itemize}

Similarly to~\thm{complier_kac}, we can distill from our proof a stronger security condition for a distribution of unitaries, which approximates Haar random unitary under a larger family of twirling operations. Composing such a distribution of unitaries with an additional round of parallel Kac's walk yields a strong \pru.

\begin{theorem}
    Let $V$ be an $n$-qubit unitary sampled from the distribution $\HPC_{n,1}$. Define
    $\eqProj := \sum_{x\in\st{0,1}^n} \ketbra{x,x}$ and
    $\Pi^{\mathsf{EPR}} := \frac{1}{N}\sum_{x,y\in\st{0,1}^n} \ketbratwo{x,x}{y,y}$.
    Suppose $K$ is sampled from a distribution $\rssd$ over $\ugroup{N}$ such that
    \[
        \norm{\expect{K}{A\cdot\eqProj\cdot A^\dag}}_{\infty}
        \leq\delta
        \quad\text{for all } A\in\{K\otimes K,(K\otimes K)^\dag\} \enspace,
    \]
    and
    \[
        \norm{\expect{K}{B\cdot
        (\eqProj-\Pi^{\mathsf{EPR}})\cdot B^\dag}}_{\infty}
        \leq\delta
        \quad\text{for all } B\in\{K\otimes\widebar{K},
        (K\otimes\widebar{K})^\dag\} \enspace.
    \]
    Then, for any adversary making at most $t$ forward and inverse queries,
    $K_2VK_1$, with $K_1,K_2\gets \rssd$ independently, is indistinguishable from a Haar random unitary except with advantage at most $O(t^2\cdot (\delta + 1/\sqrt{N})^{1/4})$. Concretely, the following are examples of instantiations of $K$.
    \begin{itemize}
        \item $K$ is a $\delta$-approximate unitary $2$-design, or
        \item $K \gets \HPC_{n,T=30n}$ is an $O(n)$-round parallel Kac's walk.  
    \end{itemize}
\end{theorem}

Below, we first extend the $\HPO$ oracle and the $\mathsf{HP}$-relation states to handle inverse queries and describe the action of $\HPO$ using the $\mathsf{HP}$-relation states in \sct{new_hpo}.
Then, we introduce a partial path-recording oracle $W$ in \sct{new_w}
as an intermediate operator to connect the action of
$\HPO$ and the path-recording oracle $V$ defined in \sct{new_v}.
Finally, we prove \thm{spru} in \sct{new_proof}.

\subsection{Action of \texorpdfstring{$\HPO$}{HPO} Oracle and Its Inverse}
\label{sec:new_hpo}
We start by adding inverse query to the $\HPO$ oracle.

\begin{definition}[Purified Function-Permutation Oracle]\label{def:hpo}
The \emph{purified function permutation oracle} $\HPO$ is a unitary
on registers $\mathsf{A}$, $\mathsf{H}$ and $\mathsf{P}$, where
\begin{itemize}
	\item $\mathsf{H}$ is a register with Hilbert space $\H\reg{H}$ spanned by the orthogonal states $\ket{f}$ for all $f\in\bit{n-1}\to\bit{3d}$, 
	\item $\mathsf{P}$ is a register with Hilbert space $\H\reg{P}$ spanned by the orthogonal states $\ket{\sigma}$ for all $\sigma\in\S_N$. 
\end{itemize}
The unitary operator $\HPO$ acts as
\[
	\HPO\reg{AHP} \ket{x}\reg{A}\ket{f}\reg{H}\ket{\sigma}\reg{P} \coloneq
	\br{H_f P_{\sigma}}\reg{A}\ket{x}\reg{A}\ket{f}\reg{H}\ket{\sigma}\reg{P}  \enspace ,
\]
and $\HPO^\dag$ acts as
\[
	\HPO^\dag\reg{AHP} \ket{y}\reg{A}\ket{f}\reg{H}\ket{\sigma}\reg{P} =
	\br{H_f P_{\sigma}}^\dag\reg{A}\ket{y}\reg{A}\ket{f}\reg{H}\ket{\sigma}\reg{P}\enspace .
\]
\end{definition}

As in \fct{puri}, querying $H_f\cdot P_{\sigma}$ and querying $\HPO$ are identical in the adversary's view:

\begin{fact} \label{fact:puri_inv}
	For any adversary $\A$ holding the register $\mathsf{A}$, the following two oracles are perfectly indistinguishable:
	\begin{itemize}[leftmargin=*]
		\item Sample uniformly random $\sigma\in S_N$ and $f:\bit{n-1}\to\bit{3d}$.
			  On each forward query, apply $H_fP_{\sigma}$ $($$(H_fP_{\sigma})^\dag$, if it is an inverse query$)$ on register $\mathsf{A}$.
		\item Initialize registers $\mathsf{H}$ and $\mathsf{P}$ in the state
			  \[
			  	\ket{\phi_{\{\}}}\coloneq
			  	\frac{1}{\sqrt{2^{3d(n-1)}}}\sum_{f:\bit{n-1}\to\bit{3d}} \ket{f}\reg{H}
			  	\otimes \frac{1}{\sqrt{N!}} \sum_{\sigma\in S_N} \ket{\sigma}\reg{P} \enspace.
			  \]
			  On each forward query, apply $\HPO$ $($$\HPO^\dag$, if it is an inverse query$)$ on registers $\mathsf{A}$, $\mathsf{H}$ and $\mathsf{P}$.
	\end{itemize}
\end{fact}

The $\mathsf{HP}$-relation states are modified in the following way.

\begin{definition}[$\mathsf{HP}$-Relation States]\label{def:modhp}
	For two integers $l$ and $r$ such that $0\leq l+r\leq N$, and two relations
	$L=\st{(x_1,y_1),\dots,(x_l,y_l)}\in\RR_l$ and 
	$R=\st{(x'_1,y'_1),\dots,(x'_r,y'_r)}\in\RR_r$, we define
	{\small
    \begin{align*}
		\ket{\phi_{L,R}}\reg{HP}\coloneq
		&\frac{1}{\sqrt{2^{3d(n-1)}(N-l-r)!}}
		\sum_{f,\sigma}
		\sum_{\substack{b\in\bit{l}\\b'\in\bit{r}}}\\
		&\quad\prod_{i=1}^l\bra{y_i}H_f\ket{y_i^{\oplus b_i}}\delta_{y_i^{\oplus b_i}=\sigma(x_i)}
		\prod_{i=1}^r\bra{{y'}_i^{\oplus b'_i}}H_f^\dag\ket{y'_i}\delta_{{y'}_i^{\oplus b'_i}=\sigma(x'_i)}
		\ket{f}\reg{H}
		\ket{\sigma}\reg{P} \enspace .
	\end{align*}}

\noindent	where $x^{\oplus 0}=x$ and $x^{\oplus 1}=\overline{x}$ for any $x\in\bit{n}$, and $\delta_{y=x}$ is an indicator that
equals $1$ iff strings $x$ and $y$ are identical in every coordinate.
\end{definition}

The $\mathsf{HP}$-relation states are orthonormal if we require
$L\cup R\in\DBR$. 

\begin{lemma}\label{lem:orthophiR2}
$\set{\ket{\phi_{L,R}}}_{L,R:L\cup R\in\DBR}$ forms a set of orthonormal vectors.
\end{lemma}

Intuitively, the action of $\HPO$ is to add the query pair $(x,y)$ into the set $L$,
while the action of $\HPO^\dag$ is to add the query pair $(x,y)$ into the set $R$.

\begin{lemma} \label{lem:hpo_action2}
	For two integers $l$ and $r$ such that $0\leq l+r\leq N$, and two relations
	$L=\st{(x_1,y_1),\dots,(x_l,y_l)}\in\RR_l$ and 
	$R=\st{(x'_1,y'_1),\dots,(x'_r,y'_r)}\in\RR_r$, we have for $x\in\bit{n}$
	\[
		\HPO\reg{AHP}\ket{x}\reg{A}\ket{\phi_{L,R}}\reg{HP} =  
		\frac{1}{\sqrt{N-l-r}}\sum_{y\in\bit{n}}
		\ket{y}\reg{A} \otimes \ket{\phi_{L\cup\st{x,y},R}}\reg{HP}\enspace .
	\]
	Similarly, we have for $y\in\bit{n}$
	\[
		\HPO^\dag\reg{AHP}\ket{y}\reg{A}\ket{\phi_{L,R}}\reg{HP} =  
		\frac{1}{\sqrt{N-l-r}}\sum_{x\in\bit{n}}
		\ket{x}\reg{A} \otimes \ket{\phi_{L,R\cup\st{x,y}}}\reg{HP}\enspace .
	\]
\end{lemma}

We provide the proofs of the above two lemmas in \app{hpo_ivs}.

\subsection{Partial Path-Recording Oracle \texorpdfstring{$W$}{W}}
\label{sec:new_w}
We first introduce
some relevant notation from \cite{MH25}.
For $t\in\N$, the register $\mathsf{R}^{(t)} \coloneq \br{\mathsf{R}^{(t)}_\mathsf{X},\mathsf{R}^{(t)}_\mathsf{Y}}$
is defined with the Hilbert space
\[\H_{\textcolor{Black}{\mathsf{R}^{(t)}}} \coloneq 
\H_{\textcolor{Black}{\mathsf{R}^{(t)}_\mathsf{X}}} \otimes \H_{\textcolor{Black}{\mathsf{R}^{(t)}_\mathsf{Y}}} \coloneq
\br{\C^N}^{\otimes t}\otimes \br{\C^N}^{\otimes t}\enspace.\]
Note that $\mathsf{R}^{(t)}_\mathsf{X}$ and $\mathsf{R}^{(t)}_\mathsf{Y}$ both consist of $t$ registers with Hilbert space $\C^N$.
For $i\leq t$, let $\mathsf{R}^{(t)}_{\mathsf{X},i}$ denote the $i$-th register in $\mathsf{R}^{(t)}_\mathsf{X}$,
and $\mathsf{R}^{(t)}_{\mathsf{Y},i}$ denote the $i$-th register in $\mathsf{R}^{(t)}_\mathsf{Y}$.
The register $\mathsf{R}$ is defined with the infinite-dimensional Hilbert space $\H\reg{R} \coloneq  \bigoplus_{t\geq 0} \H_{\textcolor{Black}{\mathsf{R}^{(t)}}}.$
And the register $\mathsf{L}$ is defined in the same way.

For integers $l,r\geq 0$, $\Pi_{l,r,\textcolor{Black}{\mathsf{LR}}}$
is the projector onto the Hilbert space $\H_{\textcolor{Black}{\mathsf{L}^{(l)}}}\otimes \H_{\textcolor{Black}{\mathsf{R}^{(r)}}}$.
For $t\geq 0$, $\Pi_{\leq t,\textcolor{Black}{\mathsf{LR}}}$
is the projector onto the Hilbert space
$\bigoplus_{l,r\geq 0:l+r\leq t}\H_{\textcolor{Black}{\mathsf{L}^{(l)}}}\otimes \H_{\textcolor{Black}{\mathsf{R}^{(r)}}}$.
For an operator $M\reg{LR}$,
$M_{l,r,\textcolor{Black}{\mathsf{LR}}} \coloneq M\reg{LR}\cdot \Pi_{l,r,\textcolor{Black}{\mathsf{LR}}}$,
$M_{\leq t,\textcolor{Black}{\mathsf{LR}}} \coloneq M\reg{LR}\cdot \Pi_{\leq t,\textcolor{Black}{\mathsf{LR}}}$,
and $M_{\leq t,\textcolor{Black}{\mathsf{LR}}}^\dag$ is $\br{M_{\leq t,\textcolor{Black}{\mathsf{LR}}}}^\dag$. For any unitary $U$, define $U^{\otimes *} \coloneq \sum_{t=0}^{\infty} U^{\otimes t}$.

\begin{definition}
	Let $L, R$ be such that $L\cup R\in\DBR$.
	We define operators $W^L$ and $W^R$ to be the linear maps
	such that for $x\in\bit{n}$ and $x\notin \BDom{L\cup R}$,
	\[
	W^L\ket{x}\reg{A}\ket{L}\reg{L}\ket{R}\reg{R} = 
	\frac{1}{\sqrt{N-2\abs{L\cup R}}}
	\sum_{\substack{y\in\bit{n}\\ y\notin \BIm{L\cup R}}}
	\ket{y}\reg{A} \otimes \ket{L\cup\st{x,y}}\reg{L} \otimes\ket{R}\reg{R} \enspace,
	\]
	and for $y\in\bit{n}$ and $y\notin \BIm{L\cup R}$,
	\[
	W^R\ket{y}\reg{A}\ket{L}\reg{L}\ket{R}\reg{R} = 
	\frac{1}{\sqrt{N-2\abs{L\cup R}}}
	\sum_{\substack{x\in\bit{n}\\ x\notin \BDom{L\cup R}}}
		\ket{x}\reg{A} \otimes \ket{L}\reg{L} \otimes\ket{R\cup\st{x,y}}\reg{R} \enspace.
	\]
	The partial path-recording oracle $W$ is defined as $W = W^L+W^{R,\dag}$.
\end{definition}

By the above definition, it is easy to verify that $W^L$ and $W^R$ are partial isometries, and that
\begin{equation*}
W^L\cdot W^{R}=W^{R,\dag}\cdot W^{L,\dag}=0 \enspace.
\end{equation*}
Thus 
\begin{equation}\label{eqn:wpiso}
W^\dag \cdot W=W^{L,\dag} \cdot W^L+W^R \cdot W^{R,\dag} \enspace.
\end{equation}
Therefore $W\cdot W^\dagger \cdot W=W$,
which implies $W$ is a partial isometry.

\begin{notation}
For a partial isometry $G$, let $\domProj{G}=G^\dag\cdot G$ and $\imProj{G}=G\cdot G^\dag$ denote the orthogonal projectors onto $\dom{G}$ and $\im{G}$. 
\end{notation}

Up to the partial isometry $\CPS:\H\reg{P}\otimes\H\reg{F}\to \H\reg{L}\otimes\H\reg{R}$ defined by
\[
\CPS
:=
\sum_{L\cup R\in\DBR}
\br{\ket{L}\reg{L}\otimes\ket{R}\reg{R}}
\bra{\phi_{L,R}}\reg{PF}\enspace,
\]
the oracle $W$ is an approximate restriction of $\HPO$:

\begin{lemma}\label{lem:wclosetohpo}
For $0\leq t< N$,
\begin{equation}\label{eqn:whpo}
\norm{W_{\leq t} - \CPS \cdot \HPO \cdot \CPS^\dag \cdot \Pi^{\dom{W}} \cdot \Pi_{\leq t} }_{\infty}\leq\frac{2t}{N-t} \enspace,
\end{equation}
\begin{equation}\label{eqn:whpodag}
	\norm{(W^\dag)_{\leq t} - \CPS \cdot \HPO^\dag \cdot \CPS^\dag \cdot \Pi^{\im{W}} \cdot \Pi_{\leq t}  }_{\infty}\leq\frac{2t}{N-t} \enspace.
\end{equation}
\end{lemma}

\begin{proof}
	We will prove \cref{eqn:whpo}, and \cref{eqn:whpodag} follows from a symmetric argument. Denote
	\[X=\CPS \cdot \HPO \cdot \CPS^\dag \cdot \Pi^{\dom{W}} \cdot \Pi_{\leq t} \enspace,\]
	\[X^L=\CPS \cdot \HPO \cdot \CPS^\dag \cdot \Pi^{\dom{W^L}} \cdot \Pi_{\leq t} \] and
	\[X^R=\Pi_{\leq t}\cdot\Pi^{\im{W^R}}\cdot\CPS \cdot \HPO^\dag \cdot \CPS^\dag  \enspace.\]
	Then by \cref{eqn:wpiso}, $X=X^L+X^{R,\dag}$. To prove \cref{eqn:whpo}, it suffices to prove \begin{equation}\label{eqn:wx}
\norm{W^L_{\leq t}-X^L}_{\infty}\leq\frac{t}{N-t}\quad\mbox{and}\quad
\norm{(W^{R,\dag})_{\leq t}-X^{R,\dag}}_{\infty}\leq\frac{t}{N-t}\enspace.
\end{equation}

Next, we prove the first inequality in \cref{eqn:wx}. The other inequality follows from a symmetric argument.
For $L, R$ such that $L\cup R\in\DBR$ and $\abs{L\cup R} \leq t$,
and $x\in\bit{n}$ such that $x\notin \BDom{L\cup R}$, we have
	\[X^L\ket{x}\reg{A}\ket{L}\reg{L}\ket{R}\reg{R} = 
	\frac{1}{\sqrt{N-\abs{L\cup R}}}
	\sum_{\substack{y\in\bit{n}\\ y\notin \BIm{L\cup R}}}
	\ket{y}\reg{A} \otimes \ket{L\cup\st{x,y}}\reg{L} \otimes\ket{R}\reg{R} \enspace,
	\]
and
\[
	W^L_{\leq t}\ket{x}\reg{A}\ket{L}\reg{L}\ket{R}\reg{R} = 
	\frac{1}{\sqrt{N-2\abs{L\cup R}}}
	\sum_{\substack{y\in\bit{n}\\ y\notin \BIm{L\cup R}}}
	\ket{y}\reg{A} \otimes \ket{L\cup\st{x,y}}\reg{L} \otimes\ket{R}\reg{R} \enspace.
	\]
Then for such $x,L,R$, we have
	{\small\begin{align*}
&\norm{\br{W^L_{\leq t}-X^L}\ket{x}\reg{A}\ket{L}\reg{L}\ket{R}\reg{R}}_2 \\
=~&\left\lVert\br{\frac{1}{\sqrt{N-\abs{L\cup R}}}-\frac{1}{\sqrt{N-2\abs{L\cup R}}}}\sum_{\substack{y\in\bit{n}\\ y\notin \BIm{L\cup R}}}
	\ket{y}\reg{A} \otimes \ket{L\cup\st{x,y}}\reg{L} \otimes\ket{R}\reg{R}\right\rVert_2\\
	=~&\sqrt{\br{\frac{1}{\sqrt{N-\abs{L\cup R}}}-\frac{1}{\sqrt{N-2\abs{L\cup R}}}}^2\br{N-2\abs{L\cup R}}} \\
	=~& 1 - \sqrt{1 - \frac{\abs{L\cup R}}{N-\abs{L\cup R}}}\\
	\leq~&\frac{\abs{L\cup R}}{N-\abs{L\cup R}} \enspace.
\end{align*}}

\noindent Note that for other $x,L,R$,
$X^L\ket{x}\ket{L}\ket{R}=W^L\ket{x}\ket{L}\ket{R}=0$.
Therefore,
\begin{align*}
	\norm{W^L_{\leq t}-X^L}_{\infty}=\sup_{\substack{x,L,R:\\L\cup R\in\DBR, \abs{L\cup R} \leq t\\
	x\notin \BDom{L\cup R}}}\norm{\br{W^L_{\leq t}-X^L}\ket{x}\reg{A}\ket{L}\reg{L}\ket{R}\reg{R}}_2\leq\frac{t}{N-t}\enspace.
\end{align*}
This completes the proof.
\end{proof}

\subsection{Path-Recording Oracle \texorpdfstring{$V$}{V}}
\label{sec:new_v}
\begin{definition}\label{def:proV}
	Let $L, R \in\RR$.
	We define operators $V^L$ and $V^R$ to be the linear maps
	such that for $x\in\bit{n}$,
	\[
	V^L\ket{x}\reg{A}\ket{L}\reg{L}\ket{R}\reg{R} = 
	\frac{1}{\sqrt{N-\abs{\BIm{L\cup R}}}}
	\sum_{\substack{y\in\bit{n}\\ y\notin \BIm{L\cup R}}}
	\ket{y}\reg{A} \otimes \ket{L\cup\st{x,y}}\reg{L} \otimes\ket{R}\reg{R} \,,
	\]
	and for $y\in\bit{n}$
	\[
	V^R\ket{y}\reg{A}\ket{L}\reg{L}\ket{R}\reg{R} = 
	\frac{1}{\sqrt{N-\abs{\BDom{L\cup R}}}}\!\!\!
	\sum_{\substack{x\in\bit{n}\\ x\notin \BDom{L\cup R}}} \!\!\!\!\!\!
		\ket{x}\reg{A} \otimes \ket{L}\reg{L} \otimes\ket{R\cup\st{x,y}}\reg{R} \,.
	\]
	The path-recording oracle $V$ is defined to be
	\begin{equation}\label{eqn:Vdef}
	V =  V^L\cdot (\id - V^R\cdot V^{R,\dag}) + (\id - V^L\cdot V^{L,\dag}) \cdot V^{R,\dag}\enspace .
	\end{equation}
\end{definition}

By the above definition, it is easy to verify that $V^L$ and $V^R$ are partial isometries, and that $V^L(V^R)$ differs from $W^L(W^R)$ by only a projection, i.e., 
\begin{equation}\label{eqn:vwlr}
V^L\cdot\domProj{W^L}=W^L\quad\mbox{and}\quad \imProj{W^R}\cdot V^R=W^R \enspace.
\end{equation}
Also, we have 
\begin{equation}\label{eqn:wvlreq0}
W^L\cdot V^R=W^R\cdot V^L=0\enspace.
\end{equation}

\begin{claim}\label{clm:V_partial_iso}
$V$ is a partial isometry.
\end{claim}

\begin{proof}
We first show that $V^L\cdot (\id - V^R\cdot V^{R,\dag})$ is a partial isometry, which is true if and only if $(\id - V^R\cdot V^{R,\dag})\cdot V^{L,\dag}\cdot V^L\cdot (\id - V^R\cdot V^{R,\dag})$ is a projector. It suffices to show that $V^{L,\dag}\cdot V^L$ and $V^R\cdot V^{R,\dag}$ commute. By the definition of $V^L$, $V^{L,\dag}\cdot V^L=\id\reg{A}\otimes\Pi_{\leq N-1,\textcolor{Black}{\mathsf{LR}}}$. Since $V^R\cdot V^{R,\dag}$ takes states in $\Pi_{\leq i+1,\textcolor{Black}{\mathsf{LR}}}$ to $\Pi_{\leq i+1,\textcolor{Black}{\mathsf{LR}}}$ (for $0\leq i\leq N-1$), it commutes with 
$\id\reg{A}\otimes\Pi_{\leq N-1,\textcolor{Black}{\mathsf{LR}}}$.
Moreover, we know that $(\id - V^L\cdot V^{L,\dag}) \cdot V^{R,\dag}$ is also a partial isometry by a similar argument.

To show that $V$ is a partial isometry, we then show that the domains of $V^L\cdot (\id - V^R\cdot V^{R,\dag})$ and $(\id - V^L\cdot V^{L,\dag}) \cdot V^{R,\dag}$ are orthogonal.
Indeed, these domains lie within the subspaces \(\id - \imProj{V^R}\) and \(\imProj{V^R}\), respectively, and are therefore orthogonal.
\end{proof}

We formalize the fact that $W$ is a restriction of $V$ in the following lemma.
\begin{lemma}\label{lem:WandV}
\begin{equation}\label{eqn:wvres}
W = V \cdot \domProj{W}\enspace, 
\end{equation}
\begin{equation}\label{eqn:wvdagres}
W^\dag = V^\dag \cdot \imProj{W}\enspace.
\end{equation}
\end{lemma} 

\begin{proof}
	To prove \cref{eqn:wvres}, it suffices to show that 
    \begin{align}
        W^L &= V \cdot \domProj{W^L}\enspace, \label{eqn:wlvres}\\
W^{R,\dag} &= V \cdot \imProj{W^R} \label{eqn:wrvres}\enspace.
    \end{align}
By \cref{eqn:wpiso}, \cref{eqn:wvres} can be obtained by summing these two equations.

We now prove \cref{eqn:wlvres}.	By \cref{eqn:Vdef}, we have
	\[V \cdot \domProj{W^L}=\br{V^L\cdot (\id - V^R\cdot V^{R,\dag}) + (\id - V^L\cdot V^{L,\dag}) \cdot V^{R,\dag}}\cdot \domProj{W^L}\,.\]
	By \cref{eqn:wvlreq0}, we have
	\[V^{R,\dag}\cdot \domProj{W^L}=V^{R,\dag} \cdot W^{L,\dag}\cdot W^L=0 \enspace.\]
	Thus
	\[V \cdot \domProj{W^L}=V^L \cdot \domProj{W^L}=W^L\enspace,\]
	where the second equality follows from \cref{eqn:vwlr}.
	
	It remains to prove \cref{eqn:wrvres}. By \cref{eqn:Vdef}, we have
	\[V \cdot \imProj{W^R}=\br{V^L\cdot (\id - V^R\cdot V^{R,\dag}) + (\id - V^L\cdot V^{L,\dag}) \cdot V^{R,\dag}}\cdot \imProj{W^R}\, .\]
	By the definition of $V^R$ and $W^R$, we have that $\mathrm{Im}(W^R)\subseteq\mathrm{Im}(V^R)$. Thus
	\[(\id - V^R\cdot V^{R,\dag})\imProj{W^R}=\imProj{W^R}-\imProj{V^R}\imProj{W^R}=0\enspace.\]
	At last, by \cref{eqn:vwlr} and \cref{eqn:wvlreq0}, we have
	\[(\id - V^L\cdot V^{L,\dag}) \cdot V^{R,\dag}\cdot \imProj{W^R}=(\id - V^L\cdot V^{L,\dag}) W^{R,\dag}=W^{R,\dag}\enspace.\]
    
	The proof of \cref{eqn:wvdagres} follows from a similar argument.
\end{proof}

\subsubsection{Two-sided Unitary Invariance}
The modified path-recording oracle $V$ also admits an approximate two-sided unitary invariance as in \cite[Section 8.3]{MH25}. We state it below and provide the proof in \app{invV}.
\begin{definition}
	For any two $n$-qubit unitaries $C$ and $D$, define
	\[
		Q[C,D] \coloneq (C\otimes D^T)\reg{L}^{\otimes *} \otimes (\widebar{C}\otimes D^\dag)\reg{R}^{\otimes *} \enspace .
	\]
\end{definition}

\begin{lemma}\label{lem:inv}
	For any two $n$-qubit unitaries $C$ and $D$, and for an integer $0\leq t\leq N-1$,
	\[
	\norm{D\reg{A} \cdot V_{\leq t} \cdot C\reg{A} \cdot Q[C,D]\reg{LR} - Q[C,D]\reg{LR} \cdot V_{\leq t} }_{\infty}\leq 32\sqrt{\frac{t(t+1)}{N}} \enspace,
	\]
	\[
	\norm{C^\dag\reg{A} \cdot (V^\dag)_{\leq t} \cdot D^\dag\reg{A} \cdot Q[C,D]\reg{LR} - Q[C,D]\reg{LR} \cdot (V^\dag)_{\leq t} }_{\infty}\leq 32\sqrt{\frac{t(t+1)}{N}} \enspace.
	\]
\end{lemma}

\subsection{Main Proof}
\label{sec:new_proof}

We are now prepared to prove \thm{spru},
which demonstrates the strong security
of our construction $\HP_{n,2T+1}$.
We first show that
our construction is indistinguishable from
the path-recording oracle $V$.

\begin{lemma}
	[$\HP_{n, 2T+1}$ is indistinguishable from $V$]\label{lem:spru_V}
	Let $\A$ be a $t$-query oracle adversary capable of
	performing both forward and inverse queries. Let
	$\HP_{n,2T+1}$ be defined as in \defi{HPnT} with $T = 30n$ and $d = 5n$, and $V$ be
	defined as in \defi{proV}. Then
	\begin{align*}
		\td\br{
			\expect{\OO\gets \HP_{n, 2T+1}}{\ketbra{\A^{\OO}_t}_{\reg{AB}}},
			\Tr\reg{LR}\br{\ketbra{\A_t^V}_{\reg{ABLR}}}
		}
		\leq \frac{t(11t+20)}{N^{1/8}}\enspace .
	\end{align*}
\end{lemma}

Then, we show the Haar random unitary is indistinguishable from $V$ as well:

\begin{lemma}
	[$V$ is indistinguishable from Haar random unitaries]\label{lem:V_haar}
	Let $\A$ be a $t$-query oracle adversary capable of
	performing both forward and inverse queries. Then
	\begin{equation*}
		\td\br{
			\expect{\OO\gets \haar}{\ketbra{\A^{\OO}_t}_{\reg{AB}}},
			\Tr\reg{LR}\br{\ketbra{\A_t^V}_{\reg{ABLR}}}
		}
		\leq \frac{t(11t+20)}{N^{1/8}}\enspace .
	\end{equation*}
\end{lemma}

In the following sections, we provide the details of the main lemma \lem{spru_V}. The proof is
structured as follows: first, we demonstrate that $V$
is indistinguishable from the twirled partial path-recording
oracle $W$; then, we show that the twirled $W$ is indistinguishable
from the twirled purified $\HPO$ oracle (\defi{thpop}).
\lem{V_haar} follows from a similar argument.
Finally, by applying the triangle inequality to
\lem{spru_V} and \lem{V_haar}, we conclude the proof of
\thm{spru} at the end of this section.

\subsubsection{\texorpdfstring{$V$ Is Indistinguishable from $W$}{V Is Indistinguishable from W}}
Our goal here is to establish that
even if the adversary can make queries from both directions,
the path-recording oracle $V$ is 
indistinguishable from $D\cdot W \cdot C$ for $C,D\gets \D$,
as long as $\D$ is one of the following two distributions:
\begin{definition}\label{def:d1d2}
We define two distributions of unitaries:
	\begin{itemize}
		\item $\D_1$: sample two independent unitary
		operators $C$ and $D$ from Haar measure on
		$\ugroup{N}$, and output $C$ and $D$,
				\item $\D_2$: sample independent $C'$ and $D$ from
		$\HP_{n,T}$ with $T=30n$ and $d=5n$, and a random
		permutation matrix $P$. Output $C=P\cdot C'$ and $D$.
	\end{itemize} 
\end{definition}
We first introduce some notation.
Let
\[
\Pi^{\mathrm{DB}}\reg{LR} \coloneq \sum_{L,R:L\cup R\in\DBR} 
\ketbra{L}\reg{L}\otimes\ketbra{R}\reg{R}\enspace,\enspace
\Pi^{\RR^2} \reg{LR} \coloneq \sum_{L,R\in\RR} 
\ketbra{L}\reg{L}\otimes\ketbra{R}\reg{R}\enspace.
\]
\begin{definition}[Controlled $C$, $D$ and $Q$]\label{def:cq}
	Define the following operators:
    \begin{align*}
        \cC &\coloneq \int_C C\reg{A} \tp \ketbra{C}\reg{C} \quad, \qquad
        \cD \coloneq \int_D D\reg{A} \tp \ketbra{D}\reg{D}\quad, \\
        \cQ &\coloneq \int_{C,D} Q[C,D]\reg{LR} \tp  \ketbra{C}\reg{C} \tp \ketbra{D}\reg{D}\quad,
    \end{align*}
	where $Q[C,D] \coloneq (C\tp D^T)^{\tp *}\reg{L} \tp (\overline{C} \tp D^\dag)^{\tp *}\reg{R}$.
\end{definition}
We then define the quantum states after querying the twirled-$W$ and $V$.
\begin{definition}[Purification of Twirled-$W$]\label{def:twp}
	We define the adversary state $\ket{\A_i^{W, \D}}\reg{ABLRCD}$ after 
	the $i$-th query to twirled-$W$ as follows:
	\begin{itemize}
				\item For $i=0$,
		\begin{equation*}
			\ket{\A_0^{W, \D}}\reg{ABLRCD} \coloneq 
			\ket{0^n0^m}\reg{AB}
			\ket{\set{ }}\reg{L} \ket{\set{ }}\reg{R}
			\ket{\init(\D)}\reg{CD}
		\end{equation*}
		where 
		\begin{equation*}
			\ket{\init(\D)} \coloneq \int_{C,D} \sqrt{d\mu_\D (C) d\mu_\D (D)}
			\ket{C}\reg{C} \ket{D}\reg{D}
		\end{equation*}
		is the initial purification on registers $\mathsf{C, D}$ set up for 
		$C, D\gets \D$; and $\mu_\D(\cdot)$ denotes the probability measure 
		of unitaries sampled from the distribution $\D$.
				\item For $1\leq i\leq t$,
		\begin{equation*}
			\ket{\A_i^{W, \D}}\reg{ABLRCD} \coloneq 
			\br{
				\br{1-b_i} \cdot \br{\cD\cdot W \cdot \cC} +
				b_i \cdot \br{\cD \cdot W \cdot \cC}^\dag
			}
			\cdot A_i
			\cdot \ket{\A_{i-1}^{W, \D}}
		\end{equation*}
		where $b_i\in \set{0,1}$ indicates that the adversary makes a 
		forward/inverse query in the $i$-th step.
	\end{itemize}
\end{definition}

\begin{definition}[Purification of $V$]\label{def:vp}
	Define the adversary state $\ket{\A_i^{V}}\reg{ABLR}$ after the $i$-th
	query to oracle $V$ as follows:
	\begin{itemize}
				\item For $i=0$,
		\begin{equation*}
			\ket{\A_0^{V}}\reg{ABLR} \coloneq 
			\ket{0^n0^m}\reg{AB}
			\ket{\set{ }}\reg{L} \ket{\set{ }}\reg{R}
		\end{equation*}
				\item For $1\leq i\leq t$,
		\begin{equation*}
			\ket{\A_i^{V}}\reg{ABLR} \coloneq 
			\br{
				\br{1-b_i} \cdot V +
				b_i \cdot V^\dagger
			}
			\cdot A_i
			\cdot \ket{\A_{i-1}^{V}}
		\end{equation*}
		where $b_i\in \set{0,1}$ indicates that the adversary makes a 
		forward/inverse query in the $i$-th step.
	\end{itemize}
\end{definition}

\begin{fact}[Norm of the purified states]\label{fact:purinorm}
	For any $t\geq 0$, $\ket{\A_t^{W, \D}}$ and $\ket{\A_t^V}$ both have 
	a norm at most $1$, since $W$ and $V$ are partial isometries, 
	meaning that applying $W, W^\dag, V$ or $V^\dag$ is equivalent to a 
	projection followed by a unitary operation.
\end{fact}
\begin{fact}[Spaces containing purified states]\label{fact:purispace}
	By definition, for any $t\geq 0$, $\ket{\A_t^{W,\D}}$ lies in the image of $\Pi_{\leq t}^{\mathrm{DB}}$, and $\ket{\A_t^V}$ lies in the image of $\Pi_{\leq t}^{\RR^2}$.
\end{fact}
We now state the following claim, which shows that the state produced by twirled-$W$ is almost the same as that from $V$, up to a local unitary transformation acting only on the registers $\mathsf{LRCD}$.
The proof of this claim follows a similar argument to that of \cite[Claim 18]{MH25}, with the full details provided in \app{claim}.
\begin{claim}\label{clm:clm18}
	For any integer $t\geq 0$, and $\D\in\st{\D_1,\D_2}$ as defined in \defi{d1d2}
	\begin{equation*}
		\re\Br{
			\bra{\A_t^{W, \D}}\reg{ABLRCD} \cdot \cQ\reg{LRCD} \cdot
			\br{\ket{\A_t^V}\reg{ABLR} \ket{\init(\D)}\reg{CD}}
		}\geq 1-\frac{38t^2}{N^{1/4}} \enspace .
	\end{equation*}
\end{claim}

Now, we are ready to prove the indistinguishability between oracles $V$ and twirled $W$.

\begin{lemma}\label{lem:V_W}
	For any $0\leq t < N$ and $\D\in\st{\D_1,\D_2}$ as defined in \defi{d1d2},
	\begin{align*}
		\td{\br{
				\Tr{\reg{-AB}} \br{\ketbra{\A_t^{W,\D}}\reg{ABLRCD}}, 
				\Tr{\reg{-AB}} \br{\ketbra{\A_t^{V}}\reg{ABLR}}
			}\leq \frac{9t}{N^{1/8}} \enspace.
		}
	\end{align*} 
\end{lemma}

\begin{proof}
Using the fact that $\norm{\ketbra{u}-\ketbra{v}}_1\leq 2\norm{\ket{u}-\ket{v}}_2$, we have 
{\small
	\begin{align*}
		&~\td \br{ 
			\ketbra{\A_t^{W,\D}}\reg{ABLRCD},
			\cQ\reg{LRCD} \cdot \br{
				\ketbra{\A_t^V}\reg{ABLR} \tp \ketbra{\init(\D)}\reg{CD}
			} \cdot  \cQ\reg{LRCD}^\dag
		}^2\\
		\leq & ~\norm{
			\ket{\A_t^{W,\D}}\reg{ABLR} - \cQ\reg{LRCD} \cdot 
			\br{\ket{\A_t^V}\reg{ABLR} \tp \ket{\init(\D)\reg{CD}}}
		}_2^2\\
		= & ~\braket{\A_t^{W,\D}\mid \A_t^{W,\D}} + \braket{\A_t^V\mid \A_t^V} - 
		2\re\Br{
			\bra{\A_t^{W, \D}}\reg{ABLRCD} \cdot \cQ\reg{LRCD} \cdot
			\br{\ket{\A_t^V}\reg{ABLR} \ket{\init(\D)}\reg{CD}}
		}\\
		\leq & ~2-2\cdot \br{ 1-\frac{38t^2}{N^{1/4}} } = \frac{76t^2}{N^{1/4}} \enspace.
	\end{align*}
}
	
\noindent Since unitary $\cQ$ only acts on registers $\mathsf{L,R,C,D}$,
    {\small
	\begin{align*}
		&\td{\br{
				\Tr{\reg{-AB}} \br{\ketbra{\A_t^{W,\D}}\reg{ABLRCD}}, 
				\Tr{\reg{-AB}} \br{\ketbra{\A_t^{V}}\reg{ABLR}}
		}}\\
		=~&\td \biggl(
				\Tr{\reg{-AB}} \br{\ketbra{\A_t^{W,\D}}\reg{ABLRCD}},\\
		&\hspace{4em} 		\Tr{\reg{-AB}} \br{\cQ\reg{LRCD} \cdot \br{
				\ketbra{\A_t^V}\reg{ABLR} \tp \ketbra{\init(\D)}\reg{CD}
			} \cdot  \cQ\reg{LRCD}^\dag}
		\biggr)\\
		\leq ~& \td \br{ 
			\ketbra{\A_t^{W,\D}}\reg{ABLRCD},
			\cQ\reg{LRCD} \cdot \br{
				\ketbra{\A_t^V}\reg{ABLR} \tp \ketbra{\init(\D)}\reg{CD}
			} \cdot  \cQ\reg{LRCD}^\dag
		}\\
		\leq ~& \frac{9t}{N^{1/8}}\enspace. \qedhere
	\end{align*}
    }
\end{proof}

\clm{clm18} also gives a bound
on the norm of $\ket{\A_t^{W,\D}}\reg{ABLRCD}$.
\begin{lemma}\label{lem:normbound_w}
	For any $0\leq t < N$ and $\D\in\st{\D_1,\D_2}$ as defined in \defi{d1d2}, we have
	\[
		\norm{\ket{\A_t^{W,\D}}\reg{ABLRCD}}_2\geq 1-\frac{38t^2}{N^{1/4}}\enspace.
	\]
\end{lemma}

\begin{proof}
	We have
	\begin{align*}
		\norm{\ket{\A_t^{W,\D}}\reg{ABLRCD}}^2_2
		= & ~
		\langle \A_t^{W,\D} \vert \A_t^{W,\D} \rangle\\
		\geq & ~
		\langle \A_t^{W,\D} \vert \A_t^{W,\D} \rangle
		\cdot
		\langle \A_t^{V} \vert \A_t^{V} \rangle \\
		= & ~
		\langle \A_t^{W,\D} \vert \A_t^{W,\D} \rangle
		\cdot
		\br{\bra{\A_t^V}\bra{\init(\D)}}
		\cQ^\dag \cdot \cQ \cdot
		\br{\ket{\A_t^V} \ket{\init(\D)}} \\
		\geq & ~
		\abs{
		\bra{\A_t^{W, \D}}\reg{ABLRCD} \cdot \cQ\reg{LRCD} \cdot
			\br{\ket{\A_t^V}\reg{ABLR} \ket{\init(\D)}\reg{CD}}
		}^2\\
		\geq & ~
		\re\Br{
			\bra{\A_t^{W, \D}}\reg{ABLRCD} \cdot \cQ\reg{LRCD} \cdot
			\br{\ket{\A_t^V}\reg{ABLR} \ket{\init(\D)}\reg{CD}}
		}^2\\
		\geq & ~
		\br{1-\frac{38t^2}{N^{1/4}}}^2\enspace,
	\end{align*}
	where the first inequality uses the fact that
	$\ket{\A_t^{V}}$ has norm at most $1$,
	the second one is from the Cauchy-Schwarz inequality,
	and the last one is from \clm{clm18}.
\end{proof}

\subsubsection{\texorpdfstring{$W$ Is Indistinguishable from $\HPO$}{W Is Indistinguishable from HPO}}
In this section, we show that after twirling,
adversaries cannot distinguish the $\HPO$ oracle from $W$
(\lem{thpo_tw}).
We first define the purification
of twirled-$\HPO$ oracle and then connect it with
twirled-$W$ (\defi{twp}) via some projection.
\begin{definition}[Purification of twirled-$\HPO$]\label{def:thpop}
	Define the adversary state $\ket{\A_i^{\HPO, \D}}\reg{ABHPCD}$ after 
	the $i$-th query to twirled-$\HPO$ as follows:
	\begin{itemize}
				\item For $i=0$,
		\begin{align*}
			\ket{\A_0^{\HPO, \D}}\reg{ABHPCD} \coloneq 
			\ket{0^n0^m}\reg{AB}
			\ket{+_f}\reg{H}\ket{+_\sigma}\reg{P}
			\ket{\init(\D)}\reg{CD}
		\end{align*}
		where $\ket{+_f}\reg{H}$ and $\ket{+_\sigma}\reg{P}$  are uniform superposition over
		all functions and permutations, respectively, and
		\begin{equation*}
			\ket{\init(\D)} \coloneq \int_{C,D} \sqrt{d\mu_\D (C) d\mu_\D (D)}
			\ket{C}\reg{C} \ket{D}\reg{D}
		\end{equation*}
		is the initial purification on registers $\mathsf{C, D}$ set up for 
		$C, D\gets \D$; and $\mu_\D(\cdot)$ denotes the probability measure 
		of unitaries sampled from the distribution $\D$.
				\item For $1\leq i\leq t$,
        {\small
		\begin{equation*}
			\ket{\A_i^{\HPO, \D}}\reg{ABHPCD} \coloneq 
			\br{
				\br{1-b_i} \cdot \br{\cD\cdot \HPO \cdot \cC} +
				b_i \cdot \br{\cD \cdot \HPO \cdot \cC}^\dag
			}
			\cdot A_i
			\cdot \ket{\A_{i-1}^{\HPO, \D}}
		\end{equation*}
        }
		
		\noindent where $b_i\in \set{0,1}$ indicates that the adversary makes a 
		forward/inverse query in the $i$-th step.
	\end{itemize}
\end{definition}

\begin{definition}	Define the projectors
	\begin{align*}
		\widetilde{\Pi}^{\dom{W}} &\coloneq \CPS^\dag \cdot \domProj{W} \cdot \CPS \enspace,
		\\
		\widetilde{\Pi}^{\im{W}} &\coloneq \CPS^\dag \cdot \imProj{W} \cdot \CPS \enspace.
	\end{align*}
\end{definition}

\begin{definition}[Purification of twirled-projected-$\HPO$]\label{def:thpoproj}
	Define the adversary state $\ket{\A_i^{\HPOp, \D}}\reg{ABHPCD}$ after 
	the $i$-th query to twirled-$\HPO$ with projection:
		\begin{itemize}
				\item For $i=0$, $\ket{\A_0^{\HPOp, \D}} \coloneq \ket{\A_0^{\HPO, \D}}$
				\item For $1\leq i\leq t$,
		\begin{align*}
			\begin{split}
			\ket{\A_i^{\HPOp, \D}} \coloneq &
			\bigg(
				\br{1-b_i} \cdot \br{\cD\cdot \HPO \cdot \widetilde{\Pi}^{\dom{W}} \cdot \cC} +
				\\
				& \enspace b_i \cdot \br{\cD\cdot \widetilde{\Pi}^{\im{W}}\cdot  \HPO \cdot \cC}^\dag
			\bigg)
			\cdot A_i
			\cdot \ket{\A_{i-1}^{\HPOp, \D}}
			\end{split}
		\end{align*}
		where $b_i\in \set{0,1}$ indicates that the adversary makes a 
		forward/backward query in the $i$-th step.
	\end{itemize}
\end{definition}

\begin{claim}\label{clm:w_hpo}
	For all integers $0\leq t \leq N/2$,
\[
		\norm{
		\ket{\A_t^{W, \D}}\reg{ABLRCD} - \CPS\reg{HP} \cdot
		\ket{\A_t^{\HPOp, \D}} 
		}_2\leq \frac{2t(t+1)}{N} \enspace .
	\]
\end{claim}
\begin{proof}[Proof by induction.]
	First, we check that the claim is true for the base case $t=0$:
	\begin{align*}
		\CPS\reg{HP} \cdot  \ket{\A_0^{\HPOp, \D}}
		&=
		\CPS\reg{HP} \cdot  \ket{\A_0^{\HPO, \D}}\\
		&= 
		\CPS\reg{HP} \cdot \br{ \ket{0^n0^m}\reg{AB}
		\ket{+_f}\reg{H}\ket{+_\sigma}\reg{P} 
		\ket{\init(\D)}\reg{CD}}\\
		&= 
		\ket{0^n0^m}\reg{AB}
		\ket{\set{}}\reg{H} \ket{\set{}}\reg{P}
		\ket{\init(\D)}\reg{CD}
		= \ket{\A_t^{W, \D}}
	\end{align*}
		Next, assuming the claim for case $(i-1)$, we will show the case $i$. 
	W.l.o.g, we assume $b_i=0$, and then
{\small\begin{align*}
&\norm{
			\ket{\A_i^{W, \D}}\reg{ABLRCD} - \CPS\reg{HP} \cdot
			\ket{\A_i^{\HPOp, \D}} 
			}_2 \nonumber\\
	&\,=\norm{\cD\cdot W\cdot\cC\cdot A_{i} \ket{\A_{i-1}^{W, \D}}\reg{ABLRCD}-\CPS\cdot\cD\cdot \HPO\cdot \widetilde{\Pi}^{\dom{W}}\cdot\cC\cdot A_{i}\ket{\A_{i-1}^{\HPOp, \D}}}_2 \,.
\end{align*}}

\noindent Note that $\ket{\A_{i-1}^{W, \D}}\reg{ABLRCD}$ lies in the image of $\Pi_{\leq i-1}$.
Thus, $\Pi_{\leq i-1}\cdot\ket{\A_{i-1}^{W, \D}}\reg{ABLRCD} = \ket{\A_{i-1}^{W, \D}}\reg{ABLRCD} $.
Therefore, we have
\begin{align*}
	&\norm{
			\ket{\A_i^{W, \D}}\reg{ABLRCD} - \CPS\reg{HP} \cdot
			\ket{\A_i^{\HPOp, \D}} 
			}_2 \\
	=~&\left\lVert
		\cD\cdot W\cdot\cC\cdot A_{i}\cdot\Pi_{\leq i-1}\cdot \ket{\A_{i-1}^{W, \D}}\reg{ABLRCD}
		 \right. \\ &\quad\quad\quad\left.
		-~\CPS\cdot\cD\cdot \HPO\cdot \widetilde{\Pi}^{\dom{W}}\cdot\cC\cdot A_{i}\ket{\A_{i-1}^{\HPOp, \D}}
		\right\rVert_2 \\
	=~&\left\lVert
		\cD\cdot W_{\leq i-1}\cdot\cC\cdot A_{i} \ket{\A_{i-1}^{W, \D}}\reg{ABLRCD}
		\right. \\ &\quad\quad\quad\left.
		-~\CPS\cdot\cD\cdot \HPO\cdot \widetilde{\Pi}^{\dom{W}}\cdot\cC\cdot A_{i}\ket{\A_{i-1}^{\HPOp, \D}}
		\right\rVert_2\\
	\leq ~&\left\lVert\cD\cdot \CPS \cdot \HPO \cdot \CPS^\dag \cdot \Pi^{\dom{W}}\cdot\Pi_{\leq i-1}\cdot\cC\cdot A_{i} \ket{\A_{i-1}^{W, \D}}\reg{ABLRCD}\right.\\
	&\left.\quad\quad\quad-\CPS\cdot\cD\cdot \HPO\cdot \widetilde{\Pi}^{\dom{W}}\cdot\cC\cdot A_{i}\ket{\A_{i-1}^{\HPOp, \D}}\right\rVert_2+\frac{2i-2}{N-i+1}\\
	=~&\left\lVert\CPS\cdot\cD\cdot\HPO\cdot\widetilde{\Pi}^{\dom{W}}\cdot\cC\cdot A_i\cdot\CPS^\dag\cdot\ket{\A_{i-1}^{W, \D}}\reg{ABLRCD}\right.\\
	&\left.\quad\quad\quad-~\CPS\cdot\cD\cdot \HPO\cdot \widetilde{\Pi}^{\dom{W}}\cdot\cC\cdot A_{i}\ket{\A_{i-1}^{\HPOp, \D}}\right\rVert_2+\frac{2i-2}{N-i+1}\\
	\leq~&\norm{\ket{\A_{i-1}^{W, \D}}\reg{ABLRCD}-\CPS\cdot\ket{\A_{i-1}^{\HPOp, \D}}}_2+\frac{4i}{N}\\
	{\leq}~&\frac{2i(i-1)}{N}+\frac{4i}{N}\quad =\frac{2i(i+1)}{N} \enspace,
	\end{align*}
where the first inequality is by \cref{lem:wclosetohpo} and the last one is by induction.
\end{proof}

We are ready to prove the indistinguishability between
oracles $W$ and $\HPO$:
\begin{lemma}\label{lem:thpo_tw}
	For all integers $0\leq t\leq N$,
	and $\D\in\st{\D_1,\D_2}$ as defined in \defi{d1d2}
    {\small
	\begin{align*}
	\td \br{
		\Tr{\reg{-AB}} \br{\ketbra{\A_t^{\HPO, \D}}\reg{ABLRCD}},
		\Tr{\reg{-AB}} \br{\ketbra{\A_t^{W,\D}}\reg{ABLRCD}}
	}\leq \frac{11t(t+1)}{N^{1/8}} \enspace.
	\end{align*}
    }
\end{lemma}

\begin{proof}
	Recall the state $\ket{\A_t^{\HPOp, \D}}\reg{ABHPCD}$ defined in \defi{thpoproj}.
	By the triangle inequality, we have
	{\small\begin{align*}
		&\td \br{
		\Tr{\reg{-AB}} \br{\ketbra{\A_t^{\HPO, \D}}\reg{ABLRCD}},
		\Tr{\reg{-AB}} \br{\ketbra{\A_t^{W,\D}}\reg{ABLRCD}}
	}\\
	\leq~
	&\underbrace{\td\br{
		\Tr{\reg{-AB}} \br{\ketbra{\A_t^{\HPO, \D}}\reg{ABLRCD}},
		\Tr{\reg{-AB}} \br{\ketbra{\A_t^{\HPOp, \D}}\reg{ABHPCD}}
	}}_{(\star)}\\
	&\quad +~
	\underbrace{\td \br{
		\Tr{\reg{-AB}} \br{\ketbra{\A_t^{\HPOp, \D}}\reg{ABHPCD}},
		\Tr{\reg{-AB}} \br{\ketbra{\A_t^{W,\D}}\reg{ABLRCD}}
	}}_{(*)} \enspace.
	\end{align*}}
	
	\noindent It suffices to show that
	\[
		(\star)\leq ~\frac{9t(t+1)}{N^{1/8}}\quad\text{and}\quad
		(*)\leq ~\frac{2t(t+1)}{N} \enspace.
	\]
	
	We first bound term $(\star)$.
	Since $\HPO$ and $\CPS$ are isometries,
	we have $\norm{\ket{\A_t^{\HPO, \D}}\reg{ABHPCD}}_2 = 1$
	and $\norm{\ket{\A_t^{\HPOp, \D}}\reg{ABHPCD}}_2\leq 1$.
	Since
	$\norm{\ketbra{u}-\ketbra{v}}_1\leq 2\norm{\ket{u}-\ket{v}}_2$
	for $\norm{u}_2\leq 1$ and $\norm{v}_2\leq 1$,
	we have
	\begin{align}
		(\star)
		\leq &~
		\td\br{
		\ketbra{\A_t^{\HPO, \D}}\reg{ABLRCD},
		\ketbra{\A_t^{\HPOp, \D}}\reg{ABHPCD}
		}\nonumber\\
		\leq &~
		\norm{
			\ket{\A_t^{\HPO, \D}}\reg{ABLRCD} -
			\ket{\A_t^{\HPOp, \D}}\reg{ABHPCD}
		}_2\nonumber\\
		\leq &~
		t\cdot\sqrt{ 1 -  \norm{\ket{\A_t^{\HPOp, \D}}\reg{ABHPCD}}_2^2}
		\enspace,
		\label{eq:bound_star}
	\end{align}
	where the last inequality is from the gentle measurement lemma in
	\lem{gentle_measurement}.
	Notice that, since $\CPS$ is an isometry, we have
	\begin{align}
		\norm{\ket{\A_t^{\HPOp, \D}}\reg{ABHPCD}}_2
		&= ~
		\norm{\CPS\cdot\ket{\A_t^{\HPOp, \D}}\reg{ABHPCD}}_2 \nonumber\\
		&\geq ~
		\norm{\ket{\A_t^{W,\D}}\reg{ABLRCD}}_2 - \frac{2t(t+1)}{N}\nonumber\\
		&\geq ~
		1-\frac{38t^2}{N^{1/4}} - \frac{2t(t+1)}{N}\nonumber\\
		&\geq ~
		1-\frac{40t(t+1)}{N^{1/4}}\label{eq:normbound_hpod}\enspace,
	\end{align}
	where the first inequality is from the triangle inequality
	and \clm{w_hpo},
	and the second one is from \lem{normbound_w}.
	Combining \cref{eq:bound_star,eq:normbound_hpod},
	we have
	\[
		(\star) \leq t\cdot\sqrt{1 - \br{1-\frac{40t(t+1)}{N^{1/4}}}^2} \leq t\cdot \sqrt{\frac{80t(t+1)}{N^{1/4}}} \leq \frac{9t(t+1)}{N^{1/8}}\enspace.
	\]
	As for term $(*)$, note that $\CPS$ acts on environment registers.
	We have
	\begin{align*}
		&(*) \nonumber\\
		&= \td\! \left(
		\Tr\reg{-AB}\!\Br{\CPS\cdot\ketbra{\A_t^{\HPOp, \D}}\reg{ABHPCD}\cdot\CPS^\dag}, \right. \nonumber\\
		& \qquad\qquad\qquad\qquad\qquad\qquad\qquad\qquad\left.
		\Tr\reg{-AB}\!\Br{\ketbra{\A_t^{W,\D}}\reg{ABLRCD}}
		\right) \nonumber \\
		&\leq \td\! \br{
		\CPS\cdot\ketbra{\A_t^{\HPOp, \D}}\reg{ABHPCD}\!\!\!\cdot\CPS^\dag,
		\ketbra{\A_t^{W,\D}}\reg{ABLRCD}
	} \nonumber\\
		&\leq \norm{\CPS\cdot\ket{\A_t^{\HPOp, \D}}\reg{ABHPCD}-
		\ket{\A_t^{W,\D}}\reg{ABLRCD}}_2 \nonumber\\
		&\leq \frac{2t(t+1)}{N} \enspace,
	\end{align*}
	where the second inequality is from the fact that
	$\norm{\ketbra{u}-\ketbra{v}}_1\leq 2\norm{\ket{u}-\ket{v}}_2$
	for $\norm{u}_2\leq 1$ and $\norm{v}_2\leq 1$,
	and the last one is from \clm{w_hpo}.
\end{proof}

\subsubsection{The Strong Security of \texorpdfstring{$\HPO$}{HPO}}
\label{subsec:spru-fpf}
Now, we complete the main proof of \thm{spru} by mainly
establishing \lem{spru_V} and \lem{V_haar}
that both our construction and Haar distribution are indistinguishable from  $V$. 

\begin{proof}[Proof of \lem{spru_V}]
	Consider $\D=\D_2$ defined in \defi{d1d2}.
	First, due to the perfect indistinguishability between the standard oracle
	and its purified version (\fct{puri}), 
	\begin{align*}
		\Tr\reg{-AB} \br{\ketbra{\A_t^{\HPO, \D}}\reg{ABHPCD}}
		&=
		\expect{\substack{H_fP_\sigma\\C,D\leftarrow\D}}{\ketbra{\A_t^{D H_fP_\sigma C}}\reg{AB}} \\
		&=
		\expect{\oracle \gets \HPC_{n,2T+1}}{\ketbra{\A_t^\oracle}\reg{AB}} \enspace.
	\end{align*}
	And, from \lem{V_W}, we get
	\begin{align*}
		\td{\br{
				\Tr{\reg{-AB}} \br{\ketbra{\A_t^{W,\D}}\reg{ABLRCD}}, 
				\Tr{\reg{-AB}} \br{\ketbra{\A_t^{V}}\reg{ABLR}}
			}\leq \frac{9t}{N^{1/8}} \enspace.
		}
	\end{align*}
	Then, according to \lem{thpo_tw}, we have
		{\small\begin{align*}
		\td \br{
			\Tr{\reg{-AB}} \br{\ketbra{\A_t^{\HPO, \D}}\reg{ABLRCD}},
			\Tr{\reg{-AB}} \br{\ketbra{\A_t^{W,\D}}\reg{ABLRCD}}
		}\leq \frac{11t(t+1)}{N^{1/8}} \enspace.
	\end{align*}}
	
	\noindent Thus, via the triangle inequality, we show:
		\begin{align*}
		&\td\br{
			\expect{\OO\gets \HP_{n, 2T+1}}{\ketbra{\A^{\OO}_t}_{\reg{AB}}},
			\Tr\reg{-AB}\br{\ketbra{\A_t^V}_{\reg{ABLR}}}
		} \\
		&= \td\br{
			\Tr\reg{-AB} \br{\ketbra{\A_t^{\HPO, \D}}\reg{ABHPCD}},
			\Tr\reg{-AB} \br{\ketbra{\A_t^V}_{\reg{ABLR}}}
		}\\
		&\leq\td{\br{
			\Tr{\reg{-AB}} \br{\ketbra{\A_t^{W,\D}}\reg{ABLRCD}}, 
			\Tr{\reg{-AB}} \br{\ketbra{\A_t^{V}}\reg{ABLR}}}} \\
		& \enspace \enspace + \td \br{
			\Tr{\reg{-AB}} \br{\ketbra{\A_t^{\HPO, \D}}\reg{ABLRCD}},
			\Tr{\reg{-AB}} \br{\ketbra{\A_t^{W,\D}}\reg{ABLRCD}}}
		\\
		& \leq \frac{11t^2+20t}{N^{1/8}}\enspace . \qedhere
	\end{align*}
\end{proof}

The argument above also proves \lem{V_haar} by considering 
$\D=\D_1$ as defined in \defi{d1d2}.
Therefore, combining \lem{spru_V} and \lem{V_haar} yields \thm{spru}.
We can finally conclude that our construction, based on parallel Kac's walk, satisfies strong security against adversaries making inverse queries.
\begin{theorem}
	By replacing random functions and random permutations with
	their post-quantum secure pseudorandom counterparts in the sampling procedure of $\hHPC_{n,2T+1}$ where $T=30n$ and $d=5n$, we obtain a strong \pru.
	\end{theorem}
  
\addtocontents{toc}{\protect\setcounter{tocdepth}{-10}}

\newpage
\bibliographystyle{alpha}
\bibliography{ref}

@misc{CSBH25,
  title={Unitary designs in nearly optimal depth},
  author={Cui, Laura and Schuster, Thomas and Brandao, Fernando and Huang, Hsin-Yuan},
  year={2025},
  howpublished={arXiv:2507.06216}
}

@InProceedings{BHHP24,
  author =	{Bostanci, John and Haferkamp, Jonas and Hangleiter, Dominik and Poremba, Alexander},
  title =	{Efficient Quantum Pseudorandomness from Hamiltonian Phase States},
  booktitle =	{20th Conference on the Theory of Quantum Computation, Communication and Cryptography (TQC 2025)},
  pages =	{9:1--9:18},
  series =	{Leibniz International Proceedings in Informatics (LIPIcs)},
  ISBN =	{978-3-95977-392-8},
  ISSN =	{1868-8969},
  year =	{2025},
  volume =	{350},
  editor =	{Fefferman, Bill},
  publisher =	{Schloss Dagstuhl -- Leibniz-Zentrum f{\"u}r Informatik},
  address =	{Dagstuhl, Germany},
  URN =		{urn:nbn:de:0030-drops-240586},
  doi =		{10.4230/LIPIcs.TQC.2025.9}
}

@article{LQSY+24,
author = {Lu, Chuhan and Qin, Minglong and Song, Fang and Yao, Penghui and Zhao, Mingnan},
title = {Quantum Pseudorandom Scramblers},
journal = {SIAM Journal on Computing},
volume = {55},
number = {4},
pages = {787-850},
year = {2026},
doi = {10.1137/24M170990X},

URL = { 
    
        https://doi.org/10.1137/24M170990X
    
    

},
eprint = { 
    
        https://doi.org/10.1137/24M170990X
    
    

}
}

@article{
schuster2024randomunitariesextremelylow,
author = {Thomas Schuster  and Jonas Haferkamp  and Hsin-Yuan Huang },
title = {Random unitaries in extremely low depth},
journal = {Science},
volume = {389},
number = {6755},
pages = {92-96},
year = {2025},
doi = {10.1126/science.adv8590},
URL = {https://www.science.org/doi/abs/10.1126/science.adv8590},
eprint = {https://www.science.org/doi/pdf/10.1126/science.adv8590}}

@INPROCEEDINGS{MPSY24,
  author={Metger, Tony and Poremba, Alexander and Sinha, Makrand and Yuen, Henry},
  booktitle={IEEE 65th Annual Symposium on Foundations of Computer Science -- FOCS 2024}, 
  title={Simple Constructions of Linear-Depth $t$-Designs and Pseudorandom Unitaries}, 
  year={2024},
  pages={485-492},
  doi={10.1109/FOCS61266.2024.00038}
}

@InProceedings{brakerski2024realvalued,
author="Brakerski, Zvika
and Magrafta, Nir",
title="Real-Valued Somewhat-Pseudorandom Unitaries",
booktitle="Theory of Cryptography Conference -- TCC 2024",
year="2024",
publisher="Springer",
pages="36--59",
}

@InProceedings{AGKL23,
author="Ananth, Prabhanjan
and Gulati, Aditya
and Kaleoglu, Fatih
and Lin, Yao-Ting",
title="Pseudorandom Isometries",
booktitle="Advances in Cryptology -- EUROCRYPT 2024",
year="2024",
publisher="Springer",
pages="226--254"
}

@article{PS17,
  title={Kac's walk on $ n $-sphere mixes in $ n\log n $ steps},
  author={Pillai, Natesh S. and Smith, Aaron},
  journal={The Annals of Applied Probability},
  volume={27},
  number={1},
  pages={631--650},
  year={2017},
  publisher={Institute of Mathematical Statistics},
  doi={10.1214/16-AAP1214}
}

@article{Zhandry16_qprp,
  doi = {10.22331/q-2025-04-08-1696},
  title = {A note on quantum-secure {PRP}s},
  author = {Zhandry, Mark},
  journal = {{Quantum}},
  issn = {2521-327X},
  publisher = {{Verein zur F{\"{o}}rderung des Open Access Publizierens in den Quantenwissenschaften}},
  volume = {9},
  pages = {1696},
  month = apr,
  year = {2025}
}

@article{JPSSS22,
author = {Vishesh Jain and Natesh S. Pillai and Ashwin Sah and Mehtaab Sawhney and Aaron Smith},
title = {{Fast and memory-optimal dimension reduction using {Kac}’s walk}},
volume = {32},
journal = {The Annals of Applied Probability},
number = {5},
publisher = {Institute of Mathematical Statistics},
pages = {4038 -- 4064},
year = {2022},
doi = {10.1214/22-AAP1784}
}

@misc{jeronimo2024pseudorandompseudoentangledstatessubset,
      title={Pseudorandom and Pseudoentangled States from Subset States}, 
      author={Fernando Granha Jeronimo and Nir Magrafta and Pei Wu},
      year={2023},
      howpublished={arXiv:2312.15285}
}

@misc{giurgicatiron2023pseudorandomnesssubsetstates,
      title={Pseudorandomness from Subset States}, 
      author={Tudor Giurgica-Tiron and Adam Bouland},
      year={2023},
      howpublished={arXiv:2312.09206},
}

@article{BHH16,
  title={Local random quantum circuits are approximate polynomial-designs},
  author={Brand{\~a}o, Fernando G.S.L. and Harrow, Aram W. and
                  Horodecki, Micha{\l}},
  journal={Communications in Mathematical Physics},
  volume={346},
  pages={397--434},
  year={2016},
  publisher={Springer},
  doi={10.1007/s00220-016-2706-8}
}

@article{Oliveira09,
  title={On The Convergence to Equilibrium of {K}ac's Random Walk on Matrices},
  author={Oliveira, Roberto Imbuzeiro},
  journal={The Annals of Applied Probability},
  volume = {19},
  number={3},
  pages={1200--1231},
  year={2009},
  publisher={JSTOR},
  url={https://www.jstor.org/stable/30243617}
}

@InProceedings{JLS18,
  author={Ji, Zhengfeng and Liu, Yi-Kai and Song, Fang},
  title={Pseudorandom Quantum States},
  booktitle={Advances in Cryptology -- CRYPTO 2018},
  year={2018},
  publisher={Springer},
  pages={126--152},
  doi={10.1007/978-3-319-96878-0_5}
}

@misc{Songblog17,
  author = {Fang Song},
  title  = {Quantum-secure pseudorandom permutations}, 
  year = {2017},
  note={Blog post: \url{https://qcc.fangsong.info/2017-06-quantumprp}}
}

@inproceedings{BS19,
  author="Brakerski, Zvika
and Shmueli, Omri",
title="({P}seudo) Random Quantum States with Binary Phase",
booktitle="Theory of Cryptography Conference -- TCC 2019",
year="2019",
publisher="Springer",
pages="229--250",
}

@InProceedings{BFV20,
  author =	{Bouland, Adam and Fefferman, Bill and Vazirani, Umesh},
  title =	{Computational Pseudorandomness, the Wormhole Growth Paradox, and Constraints on the {AdS/CFT} Duality},
  booktitle =	{11th Innovations in Theoretical Computer Science Conference -- ITCS 2020},
  pages =	{63:1--63:2},
  series =	{Leibniz International Proceedings in Informatics (LIPIcs)},
  year =	{2020},
  publisher =	{Schloss Dagstuhl -- Leibniz-Zentrum f{\"u}r Informatik},
  doi =		{10.4230/LIPIcs.ITCS.2020.63}
}

@article{YE23,
  title={The complexity of learning (pseudo) random dynamics of black holes and other chaotic systems},
  author={Yang, Lisa and Engelhardt, Netta},
  journal={Journal of High Energy Physics},
  volume={2025},
  number={3},
  pages={1--65},
  year={2025},
  publisher={Springer}
}

@inproceedings{BS20,
  title={Scalable pseudorandom quantum states},
  author={Brakerski, Zvika and Shmueli, Omri},
  booktitle={Advances in Cryptology -- CRYPTO 2020},
  pages={417--440},
  year={2020},
  organization={Springer}, 
  doi = {10.1007/978-3-030-56880-1_15}, 
}

@inproceedings{AQY22,
  title={Cryptography from pseudorandom quantum states},
  author={Ananth, Prabhanjan and Qian, Luowen and Yuen, Henry},
  booktitle={Advances in Cryptology -- CRYPTO 2022},
  pages={208--236},
  year={2022},
  organization={Springer},
  doi={10.1007/978-3-031-15802-5_8},
}

@inproceedings{AGQY22,
  title={Pseudorandom (function-like) quantum state generators: New definitions and applications},
  author={Ananth, Prabhanjan and Gulati, Aditya and Qian, Luowen and Yuen, Henry},
  booktitle={Theory of Cryptography Conference -- TCC 2022},
  pages={237--265},
  year={2022},
  organization={Springer},
  doi = {10.1007/978-3-031-22318-1_9}
}

@article{doi:10.1126/science.abn7293,
author = {Hsin-Yuan Huang  and Michael Broughton  and Jordan Cotler  and Sitan Chen  and Jerry Li  and Masoud Mohseni  and Hartmut Neven  and Ryan Babbush  and Richard Kueng  and John Preskill  and Jarrod R. McClean },
title = {Quantum advantage in learning from experiments},
journal = {Science},
volume = {376},
number = {6598},
pages = {1182-1186},
year = {2022},
doi = {10.1126/science.abn7293},
URL = {https://www.science.org/doi/abs/10.1126/science.abn7293},
eprint = {https://www.science.org/doi/pdf/10.1126/science.abn7293}}

@inproceedings{Kac56,
  title={Foundations of kinetic theory},
  author={Kac, Mark},
  booktitle={Third Berkeley symposium on mathematical statistics and
                  probability},
  volume={3},
  pages={171--197},
  year={1956},
}

@article{PS18,
  title={On the mixing time of {Kac’s} walk and other high-dimensional
                  {Gibbs} samplers with constraints},
  author={Pillai, Natesh S. and Smith, Aaron},
  journal={The Annals of Probability},
  volume={46},
  number={4},
  pages={2345--2399},
  year={2018},
  publisher={JSTOR},
  url={https://www.jstor.org/stable/26506605}
}

@article{Zhandry21_qprf,
  author = {Zhandry, Mark}, title = {How to Construct Quantum Random Functions}, year = {2021}, issue_date = {October 2021}, publisher = {Association for Computing Machinery}, address = {New York, NY, USA}, volume = {68}, number = {5}, issn = {0004-5411},  doi = {10.1145/3450745}, journal = {J. ACM},  articleno = {33}, numpages = {43}
}

@inproceedings{MH25,
  title={How to construct random unitaries},
  author={Ma, Fermi and Huang, Hsin-Yuan},
  booktitle={Proceedings of the 57th Annual ACM Symposium on Theory of Computing -- STOC 2025},
  pages={806--809},
  year={2025},
  url={https://doi.org/10.1145/3717823.3718254},
  doi={10.1145/3717823.3718254},
  note = {Full paper at arXiv preprint arXiv:2410.10116},
}

@misc{Har13,
      title={The Church of the Symmetric Subspace}, 
      author={Aram W. Harrow},
      year={2013},
      howpublished={arXiv:1308.6595},
}

@article{Jan2001,
 ISSN = {00911798, 2168894X},
 URL = {http://www.jstor.org/stable/2652922},
 author = {Elise Janvresse},
 journal = {The Annals of Probability},
 number = {1},
 pages = {288--304},
 publisher = {Institute of Mathematical Statistics},
 title = {Spectral Gap for {Kac's} Model of {Boltzmann} Equation},
 urldate = {2026-03-31},
 volume = {29},
 year = {2001}
}

@Article{Maslen2003,
author={Maslen, David K.},
title={The eigenvalues of {Kac's} master equation},
journal={Mathematische Zeitschrift},
year={2003},
month={Feb},
day={01},
volume={243},
number={2},
pages={291-331},
issn={1432-1823},
doi={10.1007/s00209-002-0466-y},
url={https://doi.org/10.1007/s00209-002-0466-y}
}

@article{Jan03,
author = {Janvresse, Elise},
title = {Bounds on Semigroups of Random Rotations on {SO$(n)$}},
journal = {Theory of Probability \& Its Applications},
volume = {47},
number = {3},
pages = {526-532},
year = {2003},
doi = {10.1137/S0040585X97979950},

URL = { 
    
        https://doi.org/10.1137/S0040585X97979950
    
    

},
eprint = { 
    
        https://doi.org/10.1137/S0040585X97979950
    
    

}
}

@Article{Carlen2003,
author={Eric A. Carlen
and Mario C. Carvalho
and Michael Loss},
title={Determination of the spectral gap for {Kac's} master equation and related stochastic evolution},
journal={Acta Mathematica},
year={2003},
month={Mar},
day={01},
volume={191},
number={1},
pages={1-54},
doi={10.1007/BF02392695},
url={https://doi.org/10.1007/BF02392695}
}

@misc{PS26,
      title={Kac's walk on rotation matrices mixes in $n^2 \log n$ steps}, 
      author={Natesh S. Pillai and Aaron Smith},
      year={2026},
      howpublished={arXiv:2604.23828},
}

\newpage
\appendix
\renewcommand{\theHsection}{appendix.\Alph{section}}
\renewcommand{\theHsubsection}{\theHsection.\arabic{subsection}}
\phantomsection
\pdfbookmark[0]{Appendix}{appendix}

\section{Proof of \texorpdfstring{\lem{orthophiR}}{Lemma 20}}\label{sec:orthophiR}
\begin{replemma}{lem:orthophiR}
$\set{\ket{\phi_R}}_{R\in\DBR}$ forms a set of orthogonal vectors.
\end{replemma}

\setcounter{lemma}{61}
Before proving this lemma, we show some properties of $H_f$:
\begin{lemma}\label{lem:Hfproperty}
For any $x\in\bit{n}$, 
\begin{enumerate}
\item $\expect{f}{\bra{x}H_f\ket{x}}=\expect{f}{\bra{x}H_f\ket{\bar{x}}}=0$;
\item $\expect{f}{\overline{\bra{x}H_f\ket{x}}\bra{x}H_f\ket{\bar{x}}}=0$;
\item $\expect{f}{\overline{\bra{x}H_f\ket{x}}\bra{\bar{x}}H_f\ket{\bar{x}}}=\expect{f}{\overline{\bra{x}H_f\ket{\bar{x}}}\bra{\bar{x}}H_f\ket{x}}=0.$
\end{enumerate}

\end{lemma}
\begin{proof}
By \cref{eq:widehatq}, if $x=0y$ for some $y\in\bit{n-1}$, then
\begin{enumerate}
\item  Note that $\bra{x}H_f\ket{x}=\e^{\i\alpha_y}\cos\theta_y$ and $\bra{x}H_f\ket{\bar{x}}=-\e^{\i\beta_y}\sin\theta_y$. Item 1 holds, since $\expect{f}{\e^{\i\alpha_y}}=\expect{f}{\e^{\i\beta_y}}=0$

\item Note that $\overline{\bra{x}H_f\ket{x}}\bra{x}H_f\ket{\bar{x}}=-\e^{\i(\beta_y-\alpha_y)}\sin\theta_y\cos\theta_y.$ Item 2 holds, since $\expect{f}{\e^{\i(\beta_y-\alpha_y)}}=0$. 
\item Note that $\overline{\bra{x}H_f\ket{x}}\bra{\bar{x}}H_f\ket{\bar{x}}=\e^{-2\i\alpha_y}\cos^2\theta_y$ and $\overline{\bra{x}H_f\ket{\bar{x}}}\bra{\bar{x}}H_f\ket{x}=-\e^{-2\i\beta_y}\sin^2\theta_y.$
Since $\expect{f}{\e^{-2\i\alpha_y}}=\expect{f}{\e^{-2\i\beta_y}}=0$, item 3 holds. 
\end{enumerate}
The case when $x=1y$ can be argued similarly.
\end{proof}

\begin{proof}[Proof of \lem{orthophiR}]
Consider two relations $R,S\in\DBR$, where $R=\{(x_1,y_1),\dots,$
$(x_{\abs{R}},y_{\abs{R}})\}$ and $S=\set{(x'_1,y'_1),\dots,(x'_{\abs{S}},y'_{\abs{S}})}$. By \cref{lem:Hfproperty} item 1, if $\abs{R}\ne\abs{S}$, then $\braket{\phi_R|\phi_S}=0$. So we may assume $\abs{R}=\abs{S}=t.$ Then
\begin{align*}
	&\braket{\phi_R|\phi_S}=\frac{1}{2^{3d(n-1)}(N-t)!}\\
	&\qquad\cdot\sum_{f,\sigma}\sum_{b,b'\in\bit{t}}\prod_{i=1}^{t}\overline{\bra{y_i}H_f\ket{y_i^{\oplus b_i}}}\delta_{y_i^{\oplus b_i}=\sigma(x_i)}\bra{y_i'}H_f\ket{\br{y_i'}^{\oplus b_i'}}\delta_{\br{y_i'}^{\oplus b_i'}=\sigma(x_i')} \enspace.
\end{align*}

 There are two cases to consider:
\begin{itemize}
\item $R=S$. Then by \cref{lem:Hfproperty} item 2 and the fact that
\[
	\abs{\bra{x}H_f\ket{x}}^2+\abs{\bra{x}H_f\ket{\bar{x}}}^2=1
\]
for all $x\in\bit{n}$, it is not hard to check that $\braket{\phi_R|\phi_S}=1.$
\item $R\ne S.$ Now we consider three sub-cases:
\begin{itemize}
\item $\ImR=\ImS$, then $\DomR\ne\DomS$. Without loss of generality, we can assume that $y_i=y_i'$ for all $i\in[t]$. 
	Fix $i$ such that $x_i\neq x_i'$.
	For all $b_i,b_i'\in\bit{}$, if $b_i=b_i'$, then for all $\sigma$, $\delta_{y_i^{\oplus b_i}=\sigma(x_i)}\delta_{\br{y_i'}^{\oplus b_i'}=\sigma(x_i')}=0$; if $b_i\ne b_i'$, then by \cref{lem:Hfproperty} item 2, 
\[\sum_f\overline{\bra{y_i}H_f\ket{y_i^{\oplus b_i}}}\bra{y_i'}H_f\ket{\br{y_i'}^{\oplus b_i'}}=0 \enspace.\] Both cases imply $\braket{\phi_R|\phi_S}=0$.
\item $\ImR\ne\ImS$ and $\BImR=\BImS.$ Without loss of generality, we can assume that there exists $i\in[t]$ such that $y_i'=\overline{y_i}$. If $b_i=b_i'$, then by \cref{lem:Hfproperty} item 3, \[\sum_f\overline{\bra{y_i}H_f\ket{y_i^{\oplus b_i}}}\bra{y_i'}H_f\ket{\br{y_i'}^{\oplus b_i'}}=0 \enspace;\]
if $b_i\ne b_i'$, then by \cref{lem:Hfproperty} item 2, \[\sum_f\overline{\bra{y_i}H_f\ket{y_i^{\oplus b_i}}}\bra{y_i'}H_f\ket{\br{y_i'}^{\oplus b_i'}}=0\enspace.\]
Both cases imply $\braket{\phi_R|\phi_S}=0$.
\item $\BImR\ne\BImS.$ By \cref{lem:Hfproperty} item 1, $\braket{\phi_R|\phi_S}=0$. \qedhere
\end{itemize}
\end{itemize}
\end{proof} 
\section{Right Invariance Property}
\label{sec:right_inv}

\begin{replemma}{lem:right_inv}
    For an arbitrary $n$-qubit unitary operator $G$, we have that
	\[\ket{\A_t^{\PR\cdot G}}\reg{ABXY} =
		\br{G\reg{X_1}\tp \dots \tp G\reg{X_t}}\cdot
		\ket{\A_t^{\PR}}\reg{ABXY} \enspace.\]
\end{replemma}

\begin{proof}
    We have trivial identities for registers $\regi{A}$ and $\regi{X} = \regi{X_1}, \dots, \regi{X_t}$.
    \begin{align}
        \sum_{z\in\bit{n}} \ketbra{z}\reg{A} &= \idy\reg{A} \label{eq:idA}\\
        \sum_{z\in\bit{n}} \ketbra{z}\reg{X_i} &= \idy\reg{X_i}\,, \enspace \forall i\in[t] \enspace.\label{eq:idX}
    \end{align}
    For an arbitrary $n$-qubit unitary $G$ and any $x,z \in \bit{n}$, we trivially have the following scalar equality:
    \begin{equation}
        \bra{x}\reg{A} G\reg{A} \ket{z}\reg{A} = \bra{x}\reg{X_i} G\reg{X_i} \ket{z}\reg{X_i} \enspace.\label{eq:scalareq}
    \end{equation}
    Thus, we have
    \begin{align}
        \sum_{x\in\bit{n}} \ket{x}\reg{X_i} \tp \bra{x}\reg{A} G\reg{A} 
        &= \sum_{x\in\bit{n}}\br{\ket{x}\reg{X_i} \tp \bra{x}\reg{A} G\reg{A} \br{\sum_{z\in\bit{n}}\ketbra{z}\reg{A}}} \nonumber\\
        &= \sum_{x,z} \ket{x}\reg{X_i} \tp \br{\bra{x}\reg{A} G\reg{A} \ket{z}\reg{A}}\bra{z}\reg{A}\nonumber\\
        &= \sum_{x,z} \ket{x}\reg{X_i} \tp \br{\bra{x}\reg{X_i} G\reg{X_i} \ket{z}\reg{X_i}}\bra{z}\reg{A}\nonumber\\
        &= \sum_{x,z} \br{\ketbra{x}\reg{X_i} G\reg{X_i} \ket{z}\reg{X_i}} \tp \bra{z}\reg{A} \nonumber\\
        &= \sum_{z} G\reg{X_i} \ket{z}\reg{X_i} \tp \bra{z}\reg{A}\nonumber\\
        &= \sum_{x\in\bit{n}} G\reg{X_i} \ket{x}\reg{X_i} \tp \bra{x}\reg{A} \enspace,
        \label{eq:ri}
    \end{align}
    where we apply \eq{idA} in the first equality, \eq{scalareq} in the third equality, \eq{idX} in the fifth equality, and simply relabel the variable in the last equality. 
    Using this identity on each register $\regi{X_i}$, we can rewrite the expression of \fct{actofpr} as follows. For simplicity, let $\alpha_{N,t} = \sqrt{\prod_{i=0}^{t-1} \frac{1}{(N-2i)} }$. Then we have
    \begin{align}
        &\ \ \ket{\A_t^{\PR\cdot G}}\reg{ABXY} \nonumber\\
        &= \alpha_{N,t}\cdot
		\sum_{
			\substack{(x_1,\dots,x_t)\in\br{\bit{n}}^t\\
				(y_1,\dots,y_t)\in\DB}
		}
            \Br{
		      \prod_{i=1}^t \br{ \ketbratwo{y_i}{x_i}\reg{A} \cdot G\reg{A} \cdot A\reg{AB}^{(i)} }
		      \ket{0}\reg{AB}
            } \tp
	        \ket{{\{ (x_i,y_i) \}_{i=1}^t}}\reg{XY}\nonumber\\
                &
        \begin{aligned}
        &= \alpha_{N,t}\cdot
		\sum_{
			\substack{(x_1,\dots,x_t)\in\br{\bit{n}}^t\\
				(y_1,\dots,y_t)\in\DB}
		}
            \Br{
		      \prod_{i=1}^t \br{ \ketbratwo{y_i}{x_i}\reg{A} \cdot G\reg{A} \cdot A\reg{AB}^{(i)} }
		      \ket{0}\reg{AB}
            } \tp \\
        & \enspace \frac{1}{\sqrt{t!}}\sum_{\sigma\in S_t}\br{S_\sigma \ket{x_1, \dots, x_t}\reg{X}}\tp\br{S_\sigma \ket{y_1, \dots, y_t}\reg{Y}}
        \label{eq:simplifygamma}
        \end{aligned}\\
                 &
        \begin{aligned}
        &=\alpha_{N,t}\cdot
		\sum_{
			\substack{(x_1,\dots,x_t)\in\br{\bit{n}}^t\\
				(y_1,\dots,y_t)\in\DB}
		}
            \Br{
		      \prod_{i=1}^t \br{ \ketbratwo{y_i}{x_i}\reg{A} \cdot A\reg{AB}^{(i)} }
		      \ket{0}\reg{AB}
            } \tp \\
        & \enspace \frac{1}{\sqrt{t!}}\sum_{\sigma\in S_t}\br{S_\sigma \cdot \br{G\reg{X_1} \ket{x_1}\reg{X_1} \tp \dots \tp G\reg{X_t}\ket{x_t}\reg{X_t}}}\tp\br{S_\sigma \ket{y_1, \dots, y_t}\reg{Y}} 
        \label{eq:applyri}
        \end{aligned} \\
                &= \br{G\reg{X_1}\tp \dots \tp G\reg{X_t}}\cdot
		\ket{\A_t^{\PR}}\reg{ABXY} \enspace,
        \label{eq:commuteG}
    \end{align}
    where in \eq{simplifygamma}, we rewrite the relation state on registers $\regi{XY}$, using the fact that $\gamma_R$ in \eq{rstatedef} simplifies to $t!$ when all $y_i$ in the set are distinct. We then apply \eq{ri} to obtain \eq{applyri}, and \eq{commuteG} follows from the fact that $\br{G\reg{X_1}\tp \cdots \tp G\reg{X_t}}$ performs an identical operation on all $t$ registers.
\end{proof}
 
\section{Proofs of \texorpdfstring{\lem{orthophiR2} and \lem{hpo_action2}}{Lemma 38 and Lemma 39}}
\label{sec:hpo_ivs}
\subsection{Proof of \texorpdfstring{\lem{orthophiR2}}{Lemma 38}}
\begin{replemma}{lem:orthophiR2}
$\set{\ket{\phi_{L,R}}}_{L,R:L\cup R\in\DBR}$ forms a set of orthonormal vectors.
\end{replemma}

\begin{proof}
	Let $L=\st{(x_1,y_1),\dots,(x_l,y_l)}\in\RR_l$,
	$R=\st{(x'_1,y'_1),\dots,(x'_r,y'_r)}\in\RR_r$,
	$S=\st{(x^*_1,y^*_1),\dots,(x^*_s,y^*_s)}\in\RR_l$, and
	$T=\st{(x^\triangle_1,y^\triangle_1),\dots,(x^\triangle_t,y^\triangle_t)}\in\RR_t$ such that
	$L\cup R\in\DBR$ and $S\cup T\in\DBR$.
	We need to prove:
	\begin{itemize}
		\item $\langle \phi_{L,R} \vert \phi_{S,T} \rangle = 1$, if $L=S$ and $R=T$; 
		\item $\langle \phi_{L,R} \vert \phi_{S,T} \rangle = 0$, otherwise.
	\end{itemize}
	Here,
	\begin{align}\label{eq:ip}
        \langle \phi_{L,R} \vert \phi_{S,T} \rangle 
		= &~\frac{1}{2^{3d(n-1)}(N-l-r)!} \sum_{f,\sigma}
\sum_{\substack{b,b^*\in\bit{l}\\b',b^{\triangle}\in\bit{r}}}
			\prod_{i=1}^l\langle y_i^{\oplus b_i}\rvert H_f^\dag\lvert y_i\rangle
			\delta_{y_i^{\oplus b_i}=\sigma(x_i)} \nonumber\\
		&\quad\cdot \prod_{i=1}^r\langle {y'_i}\rvert H_f\lvert {y_i'}^{\oplus b'_i}\rangle
			\delta_{{y'_i}^{\oplus b'_i}=\sigma(x'_i)} \cdot \prod_{i=1}^l\langle y^*_i\rvert H_f\lvert {y^*_i}^{\oplus b^*_i}\rangle
			\delta_{{y^*}_i^{\oplus b^*_i}=\sigma(x^*_i)}  \nonumber\\
		&\quad\cdot
		\prod_{i=1}^r\bra{{y^\triangle_i}^{\oplus b^\triangle_i}}H_f^\dag\ket{{y^\triangle}_i}\delta_{{y^\triangle_i}^{\oplus b^\triangle_i}=\sigma(x^\triangle_i)}  \enspace.
	\end{align}

	\paragraph{Case 1: $L=S=\st{(x_1,y_1),\dots,(x_l,y_l)}$ and $R=T=\st{(x'_1,y'_1),\dots,(x'_r,y'_r)}$.} Note that
	\begin{align*}
        \langle \phi_{L,R} \vert \phi_{S,T} \rangle 
		= &~\frac{1}{2^{3d(n-1)}(N-l-r)!} \sum_{f,\sigma}
\sum_{\substack{b,b^*\in\bit{l}\\b',b^{\triangle}\in\bit{r}}}
			\prod_{i=1}^l\langle y_i^{\oplus b_i}\rvert H_f^\dag\lvert y_i\rangle
			\delta_{y_i^{\oplus b_i}=\sigma(x_i)} \nonumber\\
		&\quad\cdot \prod_{i=1}^r\langle {y'_i}\rvert H_f\lvert {y'_i}^{\oplus b'_i}\rangle
			\delta_{{y'_i}^{\oplus b'_i}=\sigma(x'_i)} \cdot \prod_{i=1}^l\langle y_i\rvert H_f\lvert y_i^{\oplus b^*_i}\rangle
			\delta_{y_i^{\oplus b^*_i}=\sigma(x_i)} \nonumber\\
		&\quad\cdot \prod_{i=1}^r\langle {y'_i}^{\oplus b^{\triangle}_i}\rvert H_f^\dag\lvert y'_i\rangle
			\delta_{{y'_i}^{\oplus b^{\triangle}_i}=\sigma(x'_i)} \enspace. 
															\end{align*}
	If $b_i\neq b^*_i$, then $\delta_{y_i^{\oplus b_i}=\sigma(x_i)}\cdot \delta_{y_i^{\oplus b^*_i}=\sigma(x_i)}=0$. Similarly, if $b'_i\neq b^\triangle_i$, then $\delta_{y_i^{\oplus b'_i}=\sigma(x_i)}\cdot \delta_{y_i^{\oplus b^\triangle_i}=\sigma(x_i)}=0$.
	Therefore, the term in the summation is zero whenever $b\neq b^*$ or $b'\neq b^\triangle$. Thus,
	\begin{align*}
	\langle &\phi_{L,R} \vert \phi_{S,T} \rangle
		= ~ \frac{1}{2^{3d(n-1)}(N-l-r)!} \sum_{f,\sigma}
		\sum_{\substack{b\in\bit{l}\\b'\in\bit{r}}}
		\prod_{i=1}^l\langle y_i^{\oplus b_i}\rvert H_f^\dag\lvert y_i\rangle \delta_{y_i^{\oplus b_i}=\sigma(x_i)} \nonumber\\
		&\quad\cdot\prod_{i=1}^r\langle {y'_i}\rvert H_f\lvert {y'_i}^{\oplus b'_i}\rangle \delta_{{y'_i}^{\oplus b'_i}=\sigma(x'_i)}
		\cdot\prod_{i=1}^l\langle y_i\rvert H_f\lvert y_i^{\oplus b_i}\rangle \delta_{y_i^{\oplus b_i}=\sigma(x_i)} \nonumber\\
		&\quad\cdot\prod_{i=1}^r\langle {y'_i}^{\oplus b'_i}\rvert H_f^\dag\lvert y'_i\rangle \delta_{{y'_i}^{\oplus b'_i}=\sigma(x'_i)}\nonumber\\
		= &~ \frac{1}{2^{3d(n-1)}(N-l-r)!} \sum_{f,\sigma}
		\sum_{\substack{b\in\bit{l}\\b'\in\bit{r}}} \prod_{i=1}^l\langle y_i^{\oplus b_i}\rvert H_f^\dag\lvert y_i\rangle \langle y_i\rvert H_f\lvert y_i^{\oplus b_i}\rangle\nonumber\\
		&\quad\cdot \prod_{i=1}^r\langle {y'_i}\rvert H_f\lvert {y'_i}^{\oplus b'_i}\rangle \langle {y'_i}^{\oplus b'_i}\rvert H_f^\dag\lvert y'_i\rangle
		\cdot\prod_{i=1}^l\delta_{y_i^{\oplus b_i}=\sigma(x_i)} \cdot
		\prod_{i=1}^r\delta_{{y'_i}^{\oplus b'_i}=\sigma(x'_i)}\nonumber\\
		= &~\E_{f}\biggl[
			\frac{1}{(N-l-r)!} 
			\sum_{\substack{b\in\bit{l}\\b'\in\bit{r}}}
			\prod_{i=1}^l\langle y_i^{\oplus b_i}\rvert H_f^\dag\lvert y_i\rangle
			\langle y_i\rvert H_f\lvert y_i^{\oplus b_i}\rangle \nonumber\\
		&\cdot \prod_{i=1}^r\langle {y'_i}\rvert H_f\lvert {y'_i}^{\oplus b'_i}\rangle
			\langle {y'_i}^{\oplus b'_i}\rvert H_f^\dag\lvert y'_i\rangle  \cdot \sum_{\sigma} \prod_{i=1}^l\delta_{y_i^{\oplus b_i}=\sigma(x_i)} \cdot \prod_{i=1}^r\delta_{{y'_i}^{\oplus b'_i}=\sigma(x'_i)}
		\biggr] \enspace.
	\end{align*}
	For fixed $b,b'$, $y_i^{\oplus b'_i}$ and ${y'}_i^{\oplus b_i}$ are distinct since $y_i$ and $y'_i$ are from different blocks. Moreover,
	$x_i$ and $x'_i$ are distinct. Thus,
	\[
		\sum_{\sigma} \prod_{i=1}^l\delta_{y_i^{\oplus b_i}=\sigma(x_i)} \cdot
		\prod_{i=1}^r\delta_{{y'_i}^{\oplus b'_i}=\sigma(x'_i)} = (N-l-r)! \enspace.
	\]
	Then, we have
	\begin{align*}
	&\langle \phi_{L,R} \vert \phi_{S,T} \rangle\\
		= &\E_{f}\biggl[
			\sum_{\substack{b\in\bit{l}\\b'\in\bit{r}}}
			\prod_{i=1}^l
			\bra{y_i^{\oplus b_i}}H_f^\dag\ket{y_i}\bra{y_i}H_f\ket{y_i^{\oplus b_i}}
		\prod_{i=1}^r\bra{{y'_i}}H_f\ket{{y'_i}^{\oplus b'_i}}\bra{{y'_i}^{\oplus b'_i}}H_f^\dag\ket{y'_i}
		\biggr] \\
		=& \prod_{i=1}^l \E_{f}\biggl[ \sum_{b_i\in\st{0,1}} \bra{y_i^{\oplus b_i}}H_f^\dag\ket{y_i}\bra{y_i}H_f\ket{y_i^{\oplus b_i}} \biggr] \\
        &\qquad \qquad\qquad \qquad\qquad\cdot \prod_{i=1}^r \E_{f}\biggl[ \sum_{b'_i\in\st{0,1}} \bra{{y'_i}}H_f\ket{{y'_i}^{\oplus b'_i}}\bra{{y'_i}^{\oplus b'_i}}H_f^\dag\ket{y'_i} \biggr]
	\end{align*}
	Note that for $y\in\bit{n}$,
	\begin{align*}
		\E_{f}\biggl[
		\bra{y}H_f\ket{y}\bra{y}H_f^\dag\ket{y} + \bra{y}H_f\ket{\widebar{y}}\bra{\widebar{y}}H_f^\dag\ket{y}
	\biggr] = 1 \enspace.
	\end{align*}
	Therefore, we have $\langle \phi_{L,R} \vert \phi_{S,T} \rangle=1$. 
	
	\paragraph{Case 2: $L\neq S$ or $R\neq T$.}
	\begin{itemize}
		\item {\bf Case 2.1:} $\BIm{L\cup R} \neq \BIm{S\cup T}$. Without loss of generality, we assume that there exists a $y$ such that $y \in \st{y_1,\dots,y_l,y'_1,\dots,y'_r}$ and $y\notin \BIm{S\cup T}$.
			It is easy to check
			$\langle \phi_{L,R} \vert \phi_{S,T} \rangle = 0$, since
			\[
				\expect{f}{\bra{y}H_f^\dag\ket{y}} = \expect{f}{\bra{\widebar{y}}H_f^\dag\ket{y}} = \expect{f}{\bra{y}H_f\ket{y}} = \expect{f}{\bra{y}H_f\ket{\widebar{y}}} = 0\enspace.
			\]
		\item {\bf Case 2.2:} $\BIm{L\cup R} = \BIm{S\cup T}$.
		\begin{itemize}
			\item {\bf Case 2.2.1:} $\BIm{L} \neq \BIm{S}$.
			Without loss of generality,
			we assume that $y_1$ and $\widebar{y_1}$ are not in $S$. This means
			$y_1\in\st{y^\triangle_1,\dots,y^\triangle_t}$ or
			$\widebar{y_1}\in\st{y^\triangle_1,\dots,y^\triangle_t}$.
			\begin{itemize}
				\item Suppose that $y_1\in\st{y^\triangle_1,\dots,y^\triangle_t}$.
			Without loss of generality, we can assume $y_1 = y^\triangle_1$.
			We have two cases:
			\begin{itemize}
				\item $x_1 = x^\triangle_1$. In this case, the term in the summation
				in Eq. \eq{ip} is zero when $b_1 \neq b^\triangle_1$.
				It is not hard to check
				$\langle \phi_{L,R} \vert \phi_{S,T} \rangle = 0$,
				since
				\[
				\expect{f}{ \bra{y_1}H_f^\dag\ket{y_1}\bra{y^\triangle_1}H_f^\dag\ket{y^\triangle_1} } = \expect{f}{ \bra{\widebar{y_1}}H_f^\dag\ket{y_1}\bra{\widebar{y^\triangle_1}}H_f^\dag\ket{y^\triangle_1}} = 0 \enspace.
				\]
				\item $x_1 \neq x^\triangle_1$. In this case, the term in the summation
				in Eq. \eq{ip} is zero when $b_1 = b^\triangle_1$.
				It is not hard to check
				$\langle \phi_{L,R} \vert \phi_{S,T} \rangle = 0$,
				since
				\[
				\expect{f}{ \bra{y_1}H_f^\dag\ket{y_1}
				\bra{\widebar{y^\triangle_1}}H_f^\dag\ket{y^\triangle_1} } = \expect{f}{ \bra{\widebar{y_1}}H_f^\dag\ket{y_1}\bra{y^\triangle_1}H_f^\dag\ket{y^\triangle_1}} = 0 \enspace.
				\]
			\end{itemize}
			\item Suppose that $\widebar{y_1}\in\st{y^\triangle_1,\dots,y^\triangle_t}$.
			Without loss of generality, we can assume $\widebar{y_1} = y^\triangle_1$.
			We have two cases:
			\begin{itemize}
				\item $x_1 = x^\triangle_1$. In this case, the term in the summation in Eq. \eq{ip} is zero when $b_1 = b^\triangle_1$. It is not hard to check
				$\langle \phi_{L,R} \vert \phi_{S,T} \rangle = 0$,
				since
				\[
				\expect{f}{ \bra{y_1}H_f^\dag\ket{y_1}
				\bra{\widebar{y^\triangle_1}}H_f^\dag\ket{y^\triangle_1} } = \expect{f}{ \bra{\widebar{y_1}}H_f^\dag\ket{y_1}\bra{y^\triangle_1}H_f^\dag\ket{y^\triangle_1}} = 0 \enspace.
				\]
				\item $x_1 \neq x^\triangle_1$. In this case, the term in the summation in Eq. \eq{ip} is zero when $b_1 \neq b^\triangle_1$. It is not hard to check
				$\langle \phi_{L,R} \vert \phi_{S,T} \rangle = 0$,
				since
				\[
				\expect{f}{ \bra{y_1}H_f^\dag\ket{y_1}\bra{y^\triangle_1}H_f^\dag\ket{y^\triangle_1} + \bra{\widebar{y_1}}H_f^\dag\ket{y_1}\bra{\widebar{y^\triangle_1}}H_f^\dag\ket{y^\triangle_1}} = 0 \enspace.
				\]
			\end{itemize}
			\end{itemize}
		\item {\bf Case 2.2.2:} $\BIm{L} = \BIm{S}$. This means $\BIm{R} = \BIm{T}$ as well.
		\begin{itemize}
		\item {\bf Case 2.2.2.1:} $\st{y_1,\dots,y_l} \neq \st{y^*_1,\dots,y^*_s}$. 
			Without loss of generality, we can assume $y_1 = \widebar{y^*_1}$.
			We have two cases:
			\begin{itemize}
				\item $x_1 = x^*_1$. In this case, the term in the summation in Eq. \eq{ip} is zero when $b_1 = b^*_1$. It is not hard to check
				$\langle \phi_{L,R} \vert \phi_{S,T} \rangle = 0$,
				since
				\[
				\expect{f}{ \bra{y_1}H_f^\dag\ket{y_1}
				\bra{y^*_1}H_f\ket{\widebar{y^*_1}} } = \expect{f}{ \bra{\widebar{y_1}}H_f^\dag\ket{y_1}\bra{y^*_1}H_f\ket{y^*_1}} = 0 \enspace.
				\]
				\item $x_1 \neq x^*_1$. In this case, the term in the summation in Eq. \eq{ip} is zero when $b_1 \neq b^*_1$. It is not hard to check
				$\langle \phi_{L,R} \vert \phi_{S,T} \rangle = 0$,
				since
				\[
				\expect{f}{ \bra{y_1}H_f^\dag\ket{y_1}
				\bra{y^*_1}H_f\ket{y^*_1}} =
				\expect{f}{ \bra{\widebar{y_1}}H_f^\dag\ket{y_1} \bra{y^*_1}H_f\ket{\widebar{y^*_1}}} = 0 \enspace.
				\]
			\end{itemize}
		\item {\bf Case 2.2.2.2:} $\st{y'_1,\dots,y'_r} \neq \st{y^\triangle_1,\dots,y^\triangle_t}$. Similar to Case 2.2.2.1, we have $\langle \phi_{L,R} \vert \phi_{S,T} \rangle = 0$.
		\item {\bf Case 2.2.2.3:} $\st{y_1,\dots,y_l} = \st{y^*_1,\dots,y^*_s}$ and $\st{y'_1,\dots,y'_r}=\{y^\triangle_1,$
        $\dots,y^\triangle_t\}$.
			\begin{itemize}
				\item Suppose $L\neq S$. Then there exist $(x,y)\in L$ and $(x^*,y)\in S$ such that $x\neq x^*$.
				Without loss of generality, we can assume $y_1 = y^*_1$ and $x_1 \neq x^*_1$. In this case, the term in the summation in Eq. \eq{ip} is zero when $b_1 = b^*_1$. It is not hard to check
				$\langle \phi_{L,R} \vert \phi_{S,T} \rangle = 0$,
				since
				\[
				\expect{f}{ \bra{y_1}H_f^\dag\ket{y_1}
				\bra{y^*_1}H_f\ket{\widebar{y^*_1}} } = \expect{f}{ \bra{\widebar{y_1}}H_f^\dag\ket{y_1}\bra{y^*_1}H_f\ket{y^*_1}} = 0 \enspace.
				\]
				\item Suppose $R\neq T$. Similarly, we have $\langle \phi_{L,R} \vert \phi_{S,T} \rangle = 0$. \qedhere
			\end{itemize}
		\end{itemize}
		\end{itemize}
	\end{itemize}
\end{proof}

\subsection{Proof of \texorpdfstring{\lem{hpo_action2}}{Lemma 39}}
\begin{replemma}{lem:hpo_action2}
For two integers $l$ and $r$ such that $0\leq l+r\leq N$, and two relations
	$L=\st{(x_1,y_1),\dots,(x_l,y_l)}\in\RR_l$ and 
	$R=\st{(x'_1,y'_1),\dots,(x'_r,y'_r)}\in\RR_r$, we have for $x\in\bit{n}$
	\[
		\HPO\reg{AHP}\ket{x}\reg{A}\ket{\phi_{L,R}}\reg{HP} =  
		\frac{1}{\sqrt{N-l-r}}\sum_{y\in\bit{n}}
		\ket{y}\reg{A} \otimes \ket{\phi_{L\cup\st{x,y},R}}\reg{HP}\enspace .
	\]
	Similarly, we have for $y\in\bit{n}$
	\[
		\HPO^\dag\reg{AHP}\ket{y}\reg{A}\ket{\phi_{L,R}}\reg{HP} =  
		\frac{1}{\sqrt{N-l-r}}\sum_{x\in\bit{n}}
		\ket{x}\reg{A} \otimes \ket{\phi_{L,R\cup\st{x,y}}}\reg{HP}\enspace .
	\]
\end{replemma}

\begin{proof}
	In the proof of \lem{action_hpo}, we know
	\begin{align*}
		\HPO\reg{AHP} \ket{x}\reg{A}\ket{f}\reg{H}\ket{\sigma}\reg{P} &=
		{H_f}\reg{A} \ket{\sigma(x)}\reg{A}\ket{f}\reg{H}\ket{\sigma}\reg{P}\nonumber\\
		&= \sum_{y\in\bit{n},b\in\bit{}} 
		\langle y\rvert H_f\lvert y^{\oplus b}\rangle
		\delta_{y^{\oplus b}=\sigma(x)}
		\lvert y\rangle \reg{A}\ket{f}\reg{H}\ket{\sigma}\reg{P}\enspace.
	\end{align*}
	Similarly, for $\HPO^\dag$ we have
	\begin{align*}
		\HPO^\dag\reg{AHP} \ket{y}\reg{A}\ket{f}\reg{H}\ket{\sigma}\reg{P}= \sum_{x\in\bit{n},b\in\bit{}} \bra{y^{\oplus b}}H_f^\dag\ket{y} \delta_{y^{\oplus b}=\sigma(x)} \ket{x}\reg{A}\ket{f}\reg{H}\ket{\sigma}\reg{P} \enspace.
	\end{align*}
	Then, we have
	\begin{align*}
		&\HPO\reg{AHP}\ket{x}\reg{A}\ket{\phi_{L,R}}\reg{HP} \nonumber\\
		&\, =\frac{1}{\sqrt{2^{3d(n-1)}(N-l-r)!}}
		\sum_{f,\sigma}
		\sum_{\substack{b\in\bit{l}\\b'\in\bit{r}}}
		\prod_{i=1}^l\langle y_i\rvert H_f\lvert y_i^{\oplus b_i}\rangle
		\delta_{y_i^{\oplus b_i}=\sigma(x_i)} \nonumber\\
		&\qquad\qquad\qquad\qquad\qquad\prod_{i=1}^r\langle {y'}_i^{\oplus b'_i}\rvert H_f^\dag\lvert y'_i\rangle
		\delta_{{y'}_i^{\oplus b'_i}=\sigma(x'_i)} 
		\cdot\HPO\reg{AHP}\ket{x}\reg{A}\ket{f}\reg{H}\ket{\sigma}\reg{P} \nonumber\\
		&\,=\frac{1}{\sqrt{2^{3d(n-1)}(N-l-r)!}}
		\sum_{f,\sigma}
		\sum_{\substack{b\in\bit{l}\\b'\in\bit{r}}}
		\prod_{i=1}^l\langle y_i\rvert H_f\lvert y_i^{\oplus b_i}\rangle
		\delta_{y_i^{\oplus b_i}=\sigma(x_i)} \nonumber\\
		&\quad \prod_{i=1}^r\langle {y'}_i^{\oplus b'_i}\rvert H_f^\dag\lvert y'_i\rangle
		\delta_{{y'}_i^{\oplus b'_i}=\sigma(x'_i)}
		\cdot\sum_{y\in\bit{n},b\in\bit{}}
		\langle y\rvert H_f\lvert y^{\oplus b}\rangle
		\delta_{y^{\oplus b}=\sigma(x)}
		\lvert y\rangle\reg{A}\ket{f}\reg{H}\ket{\sigma}\reg{P} \nonumber\\
		&\,=\frac{1}{\sqrt{N-l-r}}\sum_{y\in\bit{n}}
		\lvert y\rangle\reg{A} \otimes \ket{\phi_{L\cup\st{x,y},R}}\reg{HP} \enspace.
	\end{align*}
	Similarly, we have
	\begin{align*}
		&\HPO^\dag\reg{AHP}\ket{y}\reg{A}\ket{\phi_{L,R}}\reg{HP} \nonumber\\
		&\, =\frac{1}{\sqrt{2^{3d(n-1)}(N-l-r)!}}
		\sum_{f,\sigma}
		\sum_{\substack{b\in\bit{l}\\b'\in\bit{r}}}
		\prod_{i=1}^l\langle y_i\rvert H_f\lvert y_i^{\oplus b_i}\rangle
		\delta_{y_i^{\oplus b_i}=\sigma(x_i)} \nonumber\\
		&\qquad\qquad\qquad\qquad\qquad
		\prod_{i=1}^r\langle {y'_i}^{\oplus b'_i}\rvert H_f^\dag\lvert y'_i\rangle
		\delta_{{y'_i}^{\oplus b'_i}=\sigma(x'_i)}
		\cdot\HPO^\dag\reg{AHP}\ket{y}\reg{A}\ket{f}\reg{H}\ket{\sigma}\reg{P} \nonumber\\
		&\, =\frac{1}{\sqrt{2^{3d(n-1)}(N-l-r)!}}
		\sum_{f,\sigma}
		\sum_{\substack{b\in\bit{l}\\b'\in\bit{r}}}
		\prod_{i=1}^l\langle y_i\rvert H_f\lvert y_i^{\oplus b_i}\rangle
		\delta_{y_i^{\oplus b_i}=\sigma(x_i)} \nonumber\\
		&\quad \prod_{i=1}^r\langle {y'_i}^{\oplus b'_i}\rvert H_f^\dag\lvert y'_i\rangle
		\delta_{{y'_i}^{\oplus b'_i}=\sigma(x'_i)}
		\cdot\sum_{x\in\bit{n},b\in\bit{}}
		\langle y^{\oplus b}\rvert H_f^\dag\lvert y\rangle
		\delta_{y^{\oplus b}=\sigma(x)}
		\ket{x}\reg{A}\ket{f}\reg{H}\ket{\sigma}\reg{P} \nonumber\\
		&\,=\frac{1}{\sqrt{N-l-r}}\sum_{x\in\bit{n}}
		\ket{x}\reg{A} \otimes \ket{\phi_{L,R\cup\st{x,y}}}\reg{HP} \enspace. \qedhere
	\end{align*}
\end{proof}
 
\section{Approximate Two-Side Unitary Invariance}
\label{sec:invV}
We will show that the path-recording oracle $V$ defined in \defi{proV}
satisfies an approximate unitary invariance property in this section. 
Formally, we have

\begin{replemma}{lem:inv}
	For any two $n$-qubit unitaries $C$ and $D$, and any integer $0\leq t\leq N-1$,
	\[
	\norm{D\reg{A} \cdot V_{\leq t} \cdot C\reg{A} \cdot Q[C,D]\reg{LR} - Q[C,D]\reg{LR} \cdot V_{\leq t} }_{\infty}\leq 32\sqrt{\frac{t(t+1)}{N}} \enspace,
	\]
	\[
	\norm{C^\dag\reg{A} \cdot (V^\dag)_{\leq t} \cdot D^\dag\reg{A} \cdot Q[C,D]\reg{LR} - Q[C,D]\reg{LR} \cdot (V^\dag)_{\leq t} }_{\infty}\leq 32\sqrt{\frac{t(t+1)}{N}} \enspace.
	\]
\end{replemma}

This lemma is proved by showing the closeness between $V^L$ and $E^L$
as well as $V^R$ and $E^R$
where $E^L$ and $E^R$ are operators introduced
in \cite[Section 10.1]{MH25}
that have exact two-side unitary invariance property.
We first give the definition of $E^L$ and $E^R$.

\begin{definition} \label{def:opE}
	$E^L$ and $E^R$ are linear maps on register $\mathsf{A}$, $\mathsf{L}$ and $\mathsf{R}$ such that
	{\small
    \begin{align*}
		E^L \coloneq \frac{1}{\sqrt{N}}
		\sum_{x,y\in\bit{n}} \ketbratwo{y}{x}\reg{A} \otimes
		\sum_{L\in\RR} \sqrt{\num{L,(x,y)}+1} \cdot \ketbratwo{L \cup \st{(x,y)}}{L}\reg{L} \otimes
		\sum_{R\in\RR} \ketbra{R} \reg{R} \enspace,\\
		E^R \coloneq \frac{1}{\sqrt{N}}
		\sum_{x,y\in\bit{n}} \ketbratwo{y}{x}\reg{A} \otimes
		\sum_{L\in\RR} \ketbra{L}\reg{L} \otimes
		\sum_{R\in\RR} \sqrt{\num{R,(x,y)}+1} \cdot \ketbratwo{R \cup \st{(x,y)}}{R} \reg{R} \enspace.
	\end{align*}}
	
	\noindent Here, $\num{L,(x,y)}$ is the number of times that $(x,y)$ appears in $L$.
\end{definition}

$E^L$ and $E^R$ satisfy the exact unitary invariance:
\begin{lemma}[Claim 20 in \cite{MH25}]\label{lem:invE}
	For any two $n$-qubit unitaries $C$ and $D$,
	\begin{align*}
		D\reg{A} \cdot E^L \cdot C\reg{A} = Q[C,D]\reg{LR} \cdot E^L \cdot Q[C,D]\reg{LR}^\dag \enspace, \\
		D\reg{A} \cdot E^R \cdot C\reg{A} = Q[C,D]\reg{LR} \cdot E^R \cdot Q[C,D]\reg{LR}^\dag \enspace.
	\end{align*}
\end{lemma}

The approximate unitary invariance of $V$ arises from the property that
$V^L$ and $V^R$ are very close to $E^L$ and $E^R$ in operator norm respectively.
That is,

\begin{lemma} \label{lem:V_and_E}
	For any integer $0\leq t\leq N-1$,
	\begin{align*}
		\norm{V^L_{\leq t} - E^L_{\leq t}}_{\infty} \leq \sqrt{\frac{4t(t+1)}{N}} \quad \text{and} \quad
		\norm{V^R_{\leq t} - E^R_{\leq t}}_{\infty} \leq \sqrt{\frac{4t(t+1)}{N}} \enspace.
	\end{align*}
\end{lemma}

\begin{proof}
	We demonstrate that \( V^L \) and \( E^L \) are close in operator norm, and a similar argument shows that the proximity between \( V^R \) and \( E^R \) holds as well.
	It is sufficient to show that for any state 
	\[\ket{\psi} = \sum_{\substack{x\in\bit{n}\\ L,R\in\RR \text{ s.t. } \abs{L\cup R}\leq t }} \alpha_{x,L,R} \ket{x}\reg{A}\ket{L}\reg{L}\ket{R}\reg{R} \enspace,\]
	we have
	\[
		\norm{V^L_{\leq t}\ket{\psi} - E^L_{\leq t}\ket{\psi}}_{2} \leq \sqrt{\frac{4t(t+1)}{N}} \enspace.
	\]
	Note that
	{\small\begin{align*}
		&V^L_{\leq t}\ket{\psi} - E^L_{\leq t}\ket{\psi}
		=
		\sum_{\substack{x\in\bit{n}\\ L,R\in\RR \text{ s.t. } \abs{L\cup R}\leq t }} \alpha_{x,L,R} \nonumber\\
		&\qquad\cdot\sum_{y\in\bit{n}}
		\br{ \frac{\delta_{y\notin \BIm{L\cup R}}}{\sqrt{N-\abs{\BIm{L\cup R}}}} - \frac{\sqrt{\num{L,(x,y)}+1}}{\sqrt{N}} }
		\ket{y} \ket{L \cup \st{(x,y)}} \ket{R} \nonumber\\
		&=
		\underbrace{ \sum_{\substack{x\in\bit{n}\\ L,R\in\RR \text{ s.t. } \abs{L\cup R}\leq t }} \!\!\! \alpha_{x,L,R}
		\sum_{\substack{y\in\bit{n} \\ y \notin \BIm{L\cup R} }}
		\br{ \frac{1}{\sqrt{N-\abs{\BIm{L\cup R}}}} - \frac{1}{\sqrt{N}} }
		\ket{y} \ket{L \cup \st{(x,y)}} \ket{R} }_{\ket{u}}\nonumber\\
		& \quad + \underbrace{ \sum_{\substack{x\in\bit{n}\\ L,R\in\RR \text{ s.t. } \abs{L\cup R}\leq t }} \!\!\! \alpha_{x,L,R}
		\sum_{\substack{y\in\bit{n} \\ y \in \BIm{L\cup R} }}
		\br{ - \frac{\sqrt{\num{L,(x,y)}+1}}{\sqrt{N}} }
		\ket{y} \ket{L \cup \st{(x,y)}} \ket{R} }_{\ket{v}} \,.
	\end{align*}}
	
	\noindent Since $\ket{u}$ and $\ket{v}$ are orthogonal, we have
	\[
		\norm{V^L_{\leq t}\ket{\psi} - E^L_{\leq t}\ket{\psi}}_{2}^2
		= \langle u \vert u \rangle + \langle v \vert v \rangle \enspace .
	\]
	Therefore, it is left to show $\langle u \vert u \rangle\leq \frac{2t(t+1)}{N}$ and $\langle v \vert v \rangle\leq \frac{2t(t+1)}{N}$.
	
	Notice that we can rewrite $\ket{u}$ as
    {\small
	\[
		\sum_{ \substack{y\in\bit{n}\\ L',R\in\RR : \abs{L'\cup R}\leq t+1} }
		\br{
		\sum_{ \substack{x\in\bit{n},L\in\RR:\\ L'=L \cup \st{(x,y)},\\ y\notin \BIm{L\cup R} } }
		\alpha_{x,L,R}\br{ \frac{1}{\sqrt{N-\abs{\BIm{L\cup R}}}} - \frac{1}{\sqrt{N}} }
		}
		\ket{y} \ket{L'} \ket{R} \enspace.
	\]}
	
	\noindent Thus, by the Cauchy--Schwarz inequality, we have
\begin{align*}
	\langle u \vert u \rangle
	\leq &
	\sum_{ \substack{y\in\bit{n}\\ L',R\in\RR : \abs{L'\cup R}\leq t+1} }
	\br{\sum_{ \substack{x\in\bit{n},L\in\RR:\\ L'=L \cup \st{(x,y)},\\ y\notin \BIm{L\cup R} } }
	\abs{\alpha_{x,L,R}}^2}
	\cdot A_{y,L',R} \enspace ,
\end{align*}
where
\begin{align*}
	A_{y,L',R}
	&\coloneq
	\sum_{ \substack{x\in\bit{n},L\in\RR:\\ L'=L \cup \st{(x,y)},\\ y\notin \BIm{L\cup R} } }
	\br{ \frac{1}{\sqrt{N-\abs{\BIm{L\cup R}}}} - \frac{1}{\sqrt{N}} }^2 \enspace.
\end{align*}
For $A_{y,L',R}$, we have
\begin{align*}
	A_{y,L',R}
	= &~
	\sum_{ \substack{x\in\bit{n},L\in\RR:\\ L'=L \cup \st{(x,y)},\\ y\notin \BIm{L\cup R} } }
	\br{ \frac{\sqrt{N} - \sqrt{N-\abs{\BIm{L\cup R}}}}{\sqrt{N(N-\abs{\BIm{L\cup R}})}}}^2 \nonumber\\
	\leq &~
	\sum_{ \substack{x\in\bit{n},L\in\RR:\\ L'=L \cup \st{(x,y)},\\ y\notin \BIm{L\cup R} } }
	\frac{\abs{\BIm{L\cup R}}}{N(N-\abs{\BIm{L\cup R}})} \nonumber\\
	\leq &~
	\frac{(t+1)(\abs{\BIm{L'\cup R}}-2)}{N(N-\abs{\BIm{L'\cup R}}+2)} \enspace,
\end{align*}
where the first inequality follows from the fact that, for nonnegative real numbers
$a\geq b$, one has $\sqrt{a}-\sqrt{b}\leq\sqrt{a-b}$.
The second inequality uses the following observation: once $y,L',R$ are fixed,
the number of summands satisfying $L'=L\cup\st{(x,y)}$ and
$y\notin\BIm{L\cup R}$ is at most $t+1$, and
\(
	\abs{\BIm{L'\cup R}}
	=
	\abs{\BIm{L\cup R}}+2 .
\)
Therefore,
\begin{align*}
	\langle u \vert u \rangle & \nonumber\\
	\leq &~
	\sum_{ \substack{y\in\bit{n}\\ L',R\in\RR : \abs{L'\cup R}\leq t+1} }
	\sum_{ \substack{x\in\bit{n},L\in\RR:\\ L'=L \cup \st{(x,y)},\\ y\notin \BIm{L\cup R} } }
	\abs{\alpha_{x,L,R}}^2
	\cdot
	\frac{(t+1)\abs{\BIm{L\cup R}}}{N(N-\abs{\BIm{L\cup R}})} \nonumber\\
	\leq &~
	\sum_{\substack{x\in\bit{n}\\ L,R\in\RR \text{ s.t. } \abs{L\cup R}\leq t }}
	\abs{\alpha_{x,L,R}}^2
	\cdot
	\br{\sum_{y\in\bit{n}} \delta_{y\notin \BIm{L\cup R}}}
	\cdot
	\frac{(t+1)\abs{\BIm{L\cup R}}}{N(N-\abs{\BIm{L\cup R}})} \nonumber\\
	\leq &~
	\frac{(t+1)2t}{N} \enspace,
\end{align*}
where the last inequality follows from
$\abs{\BIm{L\cup R}}\leq 2t$.
    
	As for $\ket{v}$, notice that we can rewrite it as
	\begin{align*}
		\sum_{ \substack{y\in\bit{n}\\ L',R\in\RR : \abs{L'\cup R}\leq t+1} }
		\br{
		\sum_{ \substack{x\in\bit{n},L\in\RR:\\ L'=L \cup \st{(x,y)},\\ y\in \BIm{L\cup R} } }
		\alpha_{x,L,R}\br{ - \frac{\sqrt{\num{L,(x,y)}+1}}{\sqrt{N}} }
		}
		\ket{y} \ket{L'} \ket{R}
	\end{align*}
    By the Cauchy--Schwarz inequality, we have
    \begin{equation*}
	\langle v \vert v \rangle
	\leq
	\sum_{ \substack{y\in\bit{n}\\ L',R\in\RR : \abs{L'\cup R}\leq t+1} }
	\br{
	\sum_{ \substack{x\in\bit{n},L\in\RR:\\ L'=L \cup \st{(x,y)},\\ y\in \BIm{L\cup R} } }
	\abs{\alpha_{x,L,R}}^2} \cdot B_{y,L',R}
	   \enspace,
\end{equation*}
where
\begin{equation*}
	B_{y,L',R} \coloneqq 
	\sum_{ \substack{x\in\bit{n},L\in\RR:\\ L'=L \cup \st{(x,y)},\\ y\in \BIm{L\cup R} } }
	\frac{\num{L,(x,y)}+1}{N}
	\enspace.
\end{equation*}
For $B_{y,L',R}$, we have
\begin{align*}
	B_{y,L',R} = &~ \sum_{ \substack{x\in\bit{n},L\in\RR:\\ L'=L \cup \st{(x,y)},\\ y\in \BIm{L\cup R} } }
	\frac{\num{L',(x,y)}}{N} 
	\leq ~ \frac{\sum_{ \substack{x\in\bit{n}}}\num{L',(x,y)}}{N}
	\leq ~ \frac{t+1}{N} \enspace,
\end{align*}
where the last inequality follows because, for fixed $y$,
$\sum_{ \substack{x\in\bit{n}}}\num{L',(x,y)}$ is precisely the number
of occurrences of $y$ in $L'$, which is at most $t+1$.
Therefore,
\begin{align*}
	\langle v \vert v \rangle
	\leq &~
	\sum_{ \substack{y\in\bit{n}\\ L',R\in\RR : \abs{L'\cup R}\leq t+1} }
	\br{
	\sum_{ \substack{x\in\bit{n},L\in\RR:\\ L'=L \cup \st{(x,y)},\\ y\in \BIm{L\cup R} } }
	\abs{\alpha_{x,L,R}}^2}
	\cdot
	\frac{t+1}{N}  \\
	= &~
	\sum_{\substack{x\in\bit{n}\\ L,R\in\RR \text{ s.t. } \abs{L\cup R}\leq t }} \abs{\alpha_{x,L,R}}^2
	\cdot \br{\sum_{y\in\bit{n}} \delta_{y\in \BIm{L\cup R}}} 
	\cdot \frac{t+1}{N} \\
	\leq &~ \frac{2t(t+1)}{N} \enspace. \qedhere
\end{align*}
\end{proof}

Next, we prove the main lemma in this section.

\begin{proof}[Proof of \lem{inv}]
	To prove \[
	\norm{D\reg{A} \cdot V_{\leq t} \cdot C\reg{A} \cdot Q[C,D]\reg{LR} - Q[C,D]\reg{LR} \cdot V_{\leq t} }_{\infty}\leq 32\sqrt{\frac{t(t+1)}{N}} \enspace,
	\]
	it is equivalent to show
	\[
	\norm{D\reg{A} \cdot V_{\leq t} \cdot C\reg{A} - Q[C,D]\reg{LR} \cdot V_{\leq t} \cdot Q[C,D]^\dag \reg{LR} }_{\infty}\leq 32\sqrt{\frac{t(t+1)}{N}} \enspace.
	\]
	By the triangle inequality and expanding the definition of operator $V$, we have
    {\small
	\begin{align*}
		&\norm{D\reg{A} \cdot V_{\leq t} \cdot C\reg{A} - Q[C,D]\reg{LR} \cdot V_{\leq t} \cdot Q[C,D]^\dag \reg{LR} }_{\infty} \\
		\leq &
		\underbrace{\norm{D\reg{A} \cdot V^L_{\leq t} \cdot C\reg{A} - Q[C,D]\reg{LR} \cdot V^L_{\leq t} \cdot Q[C,D]^\dag \reg{LR} }_{\infty}}_{(a)}\\
		& + 
		\underbrace{\norm{D\reg{A} \cdot \br{V^L\cdot V^R\cdot V^{R,\dag}}_{\leq t} \cdot C\reg{A} - Q[C,D]\reg{LR} \cdot \br{V^L\cdot V^R\cdot V^{R,\dag}}_{\leq t} \cdot Q[C,D]^\dag \reg{LR} }_{\infty}}_{(b)} \\
		& + 
		\underbrace{\norm{D\reg{A} \cdot \br{V^{R,\dag}}_{\leq t} \cdot C\reg{A} - Q[C,D]\reg{LR} \cdot \br{V^{R,\dag}}_{\leq t} \cdot Q[C,D]^\dag \reg{LR} }_{\infty}}_{(c)} \\
		& +
		\underbrace{\norm{D\reg{A} \cdot \br{V^L\cdot V^{L,\dag}\cdot V^{R,\dag}}_{\leq t} \cdot C\reg{A} - Q[C,D]\reg{LR} \cdot \br{V^L\cdot V^{L,\dag}\cdot V^{R,\dag}}_{\leq t} \cdot Q[C,D]^\dag \reg{LR} }_{\infty}}_{(d)} \enspace .
	\end{align*}}
	
	\noindent For $(a)$, by the unitary invariance of operator $E^L$ (\lem{invE}) and the triangle inequality, we have
	\begin{align*}
		(a) &\leq ~
		\norm{D\reg{A} \cdot V^L_{\leq t} \cdot C\reg{A} - D\reg{A} \cdot E^L_{\leq t} \cdot C\reg{A} }_{\infty} \nonumber \\
		&\qquad \quad +
		\norm{Q[C,D]\reg{LR} \cdot E^L_{\leq t} \cdot Q[C,D]^\dag \reg{LR} - Q[C,D]\reg{LR} \cdot V^L_{\leq t} \cdot Q[C,D]^\dag \reg{LR} }_{\infty} \nonumber\\
		&\leq ~2\norm{V^L_{\leq t} - E^L_{\leq t}}_{\infty} \leq ~ 4\sqrt{\frac{t(t+1)}{N}} \enspace .
	\end{align*}
	For $(b)$, notice that $\br{V^L\cdot V^R\cdot V^{R,\dag}}_{\leq t} = V^L_{\leq t}\br{ V^R\cdot V^{R,\dag}}_{\leq t}$. Therefore, we have
	{\begin{align*}
		(b)
		= &\norm{D\reg{A} \cdot V^L_{\leq t}\br{ V^R\cdot V^{R,\dag}}_{\leq t} \cdot C\reg{A} - Q[C,D]\reg{LR} \cdot V^L_{\leq t}\br{ V^R\cdot V^{R,\dag}}_{\leq t} \cdot Q[C,D]^\dag \reg{LR}}_{\infty} \\
		= &\norm{
            \begin{aligned}
			& D\reg{A} \cdot V^L_{\leq t} \cdot C\reg{A}
			\cdot C^\dag\reg{A}\br{ V^R\cdot V^{R,\dag}}_{\leq t}
			\cdot C\reg{A} \\
			&\quad -
			Q[C,D]\reg{LR} \cdot V^L_{\leq t} \cdot Q[C,D]^\dag \reg{LR}
			\cdot Q[C,D] \reg{LR}\br{ V^R\cdot V^{R,\dag}}_{\leq t}
			\cdot Q[C,D]^\dag \reg{LR}
		\end{aligned}
        }_{\infty} \\
		\leq &\norm{D\reg{A} \cdot V^L_{\leq t} \cdot C\reg{A} - Q[C,D]\reg{LR} \cdot V^L_{\leq t} \cdot Q[C,D]^\dag \reg{LR}}_{\infty} \\
		& + \norm{ C^\dag\reg{A} \cdot \br{ V^R\cdot V^{R,\dag}}_{\leq t} \cdot C\reg{A} - Q[C,D] \reg{LR} \cdot \br{ V^R\cdot V^{R,\dag}}_{\leq t} \cdot Q[C,D]^\dag \reg{LR}}_{\infty} \enspace.
	\end{align*}}
	The first term is exactly $(a)$ which is bounded by $4\sqrt{\frac{t(t+1)}{N}}$. As for the second one, note that $\br{ V^R\cdot V^{R,\dag}}_{\leq t} = V^R_{\leq t-1} \cdot V^{R,\dag}_{\leq t-1}$, we have
	\begin{align*}
		&\norm{ C^\dag\cdot\reg{A}\br{ V^R\cdot V^{R,\dag}}_{\leq t} \cdot C\reg{A} - Q[C,D] \reg{LR} \cdot\br{ V^R\cdot V^{R,\dag}}_{\leq t} \cdot Q[C,D]^\dag \reg{LR}}_{\infty} \\
		= & \norm{ C^\dag\reg{A}\cdot V^R_{\leq t-1} \cdot V^{R,\dag}_{\leq t-1} \cdot C\reg{A} - Q[C,D] \reg{LR} \cdot V^R_{\leq t-1} \cdot V^{R,\dag}_{\leq t-1} \cdot Q[C,D]^\dag \reg{LR}}_{\infty} \\
		\leq & \norm{ C^\dag\reg{A}\cdot V^R_{\leq t-1} \cdot D^\dag\reg{A} - Q[C,D] \reg{LR} \cdot V^R_{\leq t-1} \cdot Q[C,D]^\dag \reg{LR}}_{\infty} \\
		&\qquad+ \norm{ D\reg{A}\cdot V^{R,\dag}_{\leq t-1} \cdot C\reg{A} - Q[C,D] \reg{LR} \cdot V^{R,\dag}_{\leq t-1} \cdot Q[C,D]^\dag \reg{LR}}_{\infty}\\
		\leq & 8\sqrt{\frac{t(t+1)}{N}} \enspace,
	\end{align*}
	where the last inequality follows from a similar argument as $(a)$. Then, we have $(b)\leq 12\sqrt{\frac{t(t+1)}{N}}$.
	Similarly, we obtain $(c)\leq 4\sqrt{\frac{t(t+1)}{N}}$ and $(d)\leq 12\sqrt{\frac{t(t+1)}{N}}$. Thus,
	\[
		\norm{D\reg{A} \cdot V_{\leq t} \cdot C\reg{A} - Q[C,D]\reg{LR} \cdot V_{\leq t} \cdot Q[C,D]^\dag \reg{LR} }_{\infty} \leq 32\sqrt{\frac{t(t+1)}{N}}\enspace.
	\]
	The other inequality in \lem{inv} follows from a similar argument.
\end{proof}

\section{Proof of \texorpdfstring{\clm{clm18}}{Claim 56}}
\label{sec:claim}
We first need a twirling lemma.
Recall that
\[
\Pi^{\mathrm{DB}}\reg{LR} \coloneq \sum_{L,R:L\cup R\in\DBR} 
\ketbra{L}\reg{L}\otimes\ketbra{R}\reg{R}\enspace,\enspace
\Pi^{\RR^2} \reg{LR} \coloneq \sum_{L,R\in\RR} 
\ketbra{L}\reg{L}\otimes\ketbra{R}\reg{R}\enspace.
\]
\begin{lemma} [Twirling]\label{lem:twirl} 
	Let $\D\in\st{\D_1,\D_2}$ as defined in \defi{d1d2}.
	For an integer $0\leq t\leq N/4$, we have
    {\small
	\begin{align*}
	    \norm{\expect{C,D\leftarrow \D}{(C\reg{A} \otimes Q[C,D]\reg{LR})^\dag \cdot\br{
				\Pi^{\mathrm{DB}}_{\leq t, \textcolor{black}{\mathsf{LR}}} -  
				\Pi^{\dom{W}}_{\leq t, \textcolor{black}{\mathsf{LR}}}
			}\cdot (C\reg{A} \otimes Q[C,D]\reg{LR})}}_{\infty} &\leq 16t\cdot \sqrt{\frac{2t}{N}}\enspace, \\
		\norm{\expect{C,D\leftarrow \D}{(D\reg{A}^\dag \otimes Q[C,D]\reg{LR})^\dag \cdot\br{
				\Pi^{\mathrm{DB}}_{\leq t, \textcolor{black}{\mathsf{LR}}} -  
				\Pi^{\im{W}}_{\leq t, \textcolor{black}{\mathsf{LR}}}
			}\cdot (D\reg{A}^\dag \otimes Q[C,D]\reg{LR})}}_{\infty} &\leq 16t\cdot \sqrt{\frac{2t}{N}}\enspace.
	\end{align*}}
\end{lemma}

The proof of this twirling lemma is deferred to
\app{twirl}.
Now, we proceed to prove \clm{clm18}.
\begin{repclaim}{clm:clm18}
		For any integer $t\geq 0$ and $\D\in\st{\D_1,\D_2}$ as defined in \defi{d1d2}
	\begin{equation*}
		\re\Br{
			\bra{\A_t^{W, \D}}\reg{ABLRCD} \cdot \cQ\reg{LRCD} \cdot
			\br{\ket{\A_t^V}\reg{ABLR} \ket{\init(\D)}\reg{CD}}
		}\geq 1-\frac{38t^2}{N^{1/4}} \enspace .
	\end{equation*}
\end{repclaim}

\begin{proof} [Proof by induction.]
	For the base case $t=0$, 
	by \defi{cq}, \cref{def:twp} and \cref{def:vp}, we have
	\begin{align*}
		\cQ\reg{LRCD} \cdot \br{ \ket{\A_0^{V}}\reg{ABLR} 
			\ket{\init(\D)}\reg{CD}} 
		&= \cQ\reg{LRCD} \cdot \br{\ket{0^n0^m}\reg{AB}
			\ket{\set{ }}\reg{L} \ket{\set{ }}\reg{R}
			\ket{\init(\D)}\reg{CD}}
		\\
		&= \ket{0^n0^m}\reg{AB}
		\ket{\set{ }}\reg{L} \ket{\set{ }}\reg{R}
		\ket{\init(\D)}\reg{CD}
		\\
		&= \ket{\A_0^{W, \D}}\reg{ABLRCD} \enspace.
	\end{align*}
	Thus, the base case holds:
	{\small$
	\re\Br{
		\bra{\A_0^{W, \D}}\reg{ABLRCD} \cdot \cQ\reg{LRCD} \cdot
		\br{\ket{\A_0^V}\reg{ABLR} \ket{\init(\D)}\reg{CD}}
	}=1
	$}.
	
			Assuming the claim holds for some $t\geq 0$, we now
	prove it for $t+1$. Since the argument is symmetric,
	we assume the $(t+1)$-th query is a forward query, i.e.
	$b_{t+1}=0$, without loss of generality. Thus,
	\begin{align*}
		\ket{\A_{t+1}^{W, \D}}\reg{ABLRCD} &=
		\br{\cD\cdot W \cdot \cC} \cdot A_{t+1}
		\cdot \ket{\A_{t}^{W, \D}}
		\\
		\cQ\reg{LRCD} \cdot \br{ \ket{\A_{t+1}^{V}}\reg{ABLR} 
			\ket{\init(\D)}\reg{CD}} 
		&= \cQ\cdot \br{V \cdot A_{t+1} \cdot \ket{\A_t^V}	\ket{\init(\D)}}
	\end{align*}
	which gives
	\begin{align}
		\begin{split}
			&\re\Br{
				\bra{\A_{t+1}^{W, \D}} \cdot \cQ \cdot
				\br{\ket{\A_{t+1}^V}\ket{\init(\D)}}
			} \\
			&=\re \Br{
				\bra{\A_t^{W,\D}} \cdot A_{t+1}^\dag \cdot \cC^\dag \cdot W^\dag
				\cdot \cD^\dag \cdot\cQ \cdot \br{V \cdot A_{t+1} \cdot \ket{\A_t^V}	\ket{\init(\D)}}
			}
			\\
			&=\re \Br{
				\bra{\A_t^{W,\D}} \cdot A_{t+1}^\dag \cdot \cC^\dag \cdot W_{\leq t}^\dag
				\cdot \cD^\dag \cdot\cQ \cdot V_{\leq t} \cdot A_{t+1} \cdot \ket{\A_t^V}	\ket{\init(\D)}
			},
		\end{split}
		\label{clm18-eq1}
	\end{align}
	because of \fct{purispace}. Recall the notation $W_{\leq t} = W\cdot \Pi_{\leq t}$ and $V_{\leq t} = V\cdot \Pi_{\leq t}$. Then, via rewriting
	\begin{align*}
		\cQ \cdot V_{\leq t} 
								&=  \cD \cdot V_{\leq t} \cdot \cC \cdot \cQ + 
		\big(  \cQ \cdot V_{\leq t} -  \cD \cdot V_{\leq t} \cdot \cC \cdot \cQ 
		\big)\enspace,
	\end{align*}
	we further rewrite $(\ref{clm18-eq1}) = (*) + (**)$ where
	\begin{align*}
		(*) &= 
		\re\Br{\bra{\A_t^{W,\D}} \cdot A_{t+1}^\dag \cdot \cC^\dag \cdot W_{\leq t}^\dag
			\cdot V_{\leq t} \cdot \cC \cdot \cQ
			\cdot A_{t+1} \cdot \ket{\A_t^V}	\ket{\init(\D)}
		}\,, \\
		(**) &= 
		\re \left[
			\bra{\A_t^{W,\D}} \cdot A_{t+1}^\dag \cdot \cC^\dag \cdot W_{\leq t}^\dag
			\cdot \cD^\dag \cdot
			\big(\cQ \cdot V_{\leq t} -  \cD \cdot V_{\leq t} \cdot \cC \cdot \cQ \big)
			\right. \nonumber\\
			&\qquad \qquad \qquad\qquad \qquad\qquad \qquad \qquad \qquad \qquad \qquad \left.
			\cdot A_{t+1} \cdot \ket{\A_t^V}	\ket{\init(\D)}
		\right]\,.
	\end{align*}
	We first lower bound $(**)$:
	\begin{align*}
		\begin{split}
			(**)&\geq 
			-\norm{\big(\cD \cdot V_{\leq t} \cdot \cC \cdot \cQ -\cQ \cdot V_{\leq t} \big)}_{\infty}
			\\
			&\geq -\norm{\sum_{C,D}\big(
				D_A \cdot V_{\leq t} \cdot C_A \tp Q[C,D]\reg{LR} - Q[C,D]\reg{LR}\cdot V_{\leq t}
				\big)\tp \ketbra{C,D}}_{\infty}
			\\
			&\geq -\max_{C,D}{\norm{D_A \cdot V_{\leq t} \cdot C_A \tp Q[C,D]\reg{LR} - Q[C,D]\reg{LR}\cdot V_{\leq t}}_\infty}
			\\
			&\geq -32 \sqrt{\frac{t(t+1)}{N}} \enspace.
		\end{split}
	\end{align*}
	The first inequality follows from \fct{purinorm}, ensuring that the norm of 
	$\bra{\A_t^{W,\D}} \cdot A_{t+1}^\dag \cdot \cC^\dag \cdot 
	W_{\leq t}^\dag \cdot \cD^\dag$ 
	and $A_{t+1} \cdot \ket{\A_t^V}	\ket{\init(\D)}$ are at most $1$.
	The last inequality follows from \lem{inv}.
	To bound $(*)$, we first utilize the properties of $W$ and $V$ to derive:
	\begin{align}
		\begin{split}
			W_{\leq t}^\dag \cdot V_{\leq t} &=
			\br{W\cdot \Pi_{\leq t}}^\dag \cdot V \cdot \Pi_{\leq t}
			\\&= \Pi_{\leq t} \cdot W^\dag \cdot V \cdot \Pi_{\leq t}
			\\&= \Pi_{\leq t} \cdot \domProj{W} \cdot \Pi_{\leq t}
			\\&= \Pi_{\leq t} \cdot \br{\Pi^{\mathrm{DB}}-\br{\Pi^{\mathrm{DB}}-\domProj{W}}} \cdot \Pi_{\leq t}
			\\&= \Pi^{\mathrm{DB}}_{\leq t} - \br{\Pi^{\mathrm{DB}}_{\leq t}-\domProj{W}_{\leq t}} \enspace.
		\end{split}
		\label{clm18-eq2}
	\end{align}
	Then, we can rewrite $(*) = (\triangle) - (\triangle\triangle)$ using \cref{clm18-eq2} where:
	\begin{align*}
		(\triangle) &= 	\re\Br{\bra{\A_t^{W,\D}} \cdot A_{t+1}^\dag \cdot \cC^\dag \cdot 
			\Pi^{\mathrm{DB}}_{\leq t}			\cdot \cC \cdot \cQ \cdot A_{t+1} \cdot \ket{\A_t^V}\ket{\init(\D)}
		}\enspace,
	\end{align*}
    {\small
	\begin{align*}
		(\triangle\triangle) &= 	\re\Br{\bra{\A_t^{W,\D}} \cdot A_{t+1}^\dag \cdot \cC^\dag \cdot 
			\br{\Pi^{\mathrm{DB}}_{\leq t}-\domProj{W}_{\leq t}}			\cdot \cC \cdot \cQ \cdot A_{t+1} \cdot \ket{\A_t^V}\ket{\init(\D)}
		}.
	\end{align*}}
	
	\noindent Thus, to bound $(*)$ we need to separately bound $(\triangle)$ and $(\triangle\triangle)$.  First, in $(\triangle)$ we have:
	\begin{align*}
		\bra{\A_t^{W,\D}} \cdot A_{t+1}^\dag \cdot \cC^\dag \cdot 
		\Pi^{\mathrm{DB}}_{\leq t} 
		&= \bra{\A_t^{W,\D}}\cdot \Pi^{\mathrm{DB}}_{\leq t}  
		\cdot A_{t+1}^\dag \cdot \cC^\dag \\
		& =\bra{\A_t^{W,\D}}
		\cdot A_{t+1}^\dag \cdot \cC^\dag \enspace.
	\end{align*}
	Thus, by the inductive hypothesis, we have
	\begin{align*}
		\begin{split}
			(\triangle) &= \re\Br{\bra{\A_t^{W,\D}} \cdot A_{t+1}^\dag \cdot
				\cQ \cdot A_{t+1} \cdot \ket{\A_t^V}\ket{\init(\D)}}\\
			&= \re\Br{\bra{\A_t^{W,\D}}  \cdot
				\cQ \cdot \ket{\A_t^V}\ket{\init(\D)}}\\
			&\geq 1 - \frac{38t^2}{N^{1/4}} \enspace.
		\end{split}
	\end{align*}
	Then we will upper-bound $(\triangle\triangle)$:
	{\small\begin{align*}
		(\triangle\triangle) &\leq \abs{\bra{\A_t^{W,\D}} \cdot A_{t+1}^\dag \cdot \cC^\dag \cdot 
			\br{\Pi^{\mathrm{DB}}_{\leq t}-\domProj{W}_{\leq t}}			\cdot \cC \cdot \cQ \cdot A_{t+1} \cdot \ket{\A_t^V}\ket{\init(\D)}}\\
		&\leq \max_{\substack{\ket{u}\in \H\reg{ABLRCD}: \norm{\ket{u}}_2\leq 1\\ 
				\ket{v}\in \H\reg{ABLR}: \norm{\ket{v}}_2\leq 1}}
		\abs{\bra{u} \cdot
			\br{\Pi^{\mathrm{DB}}_{\leq t}-\domProj{W}_{\leq t}}			\cdot \cC \cdot \cQ \cdot \ket{v} \ket{\init(\D)}}\\
		&= \br{\max_{\substack{\ket{v}\in\H\reg{ABLR}: \\ 
					\norm{\ket{v}}_2\leq 1}}
			\bra{v}\bra{\init(\D)} \cdot \cQ^\dag \cdot \cC^\dag \cdot
			\br{\Pi^{\mathrm{DB}}_{\leq t}-\domProj{W}_{\leq t}} \cdot
			\cC \cdot \cQ \cdot \ket{v} \ket{\init(\D)}
		}^{1/2}\\
		&= \norm{\expect{C, D \gets \D}{
				\br{C\reg{A} \tp Q[C,D]\reg{LR}}^\dag
				\br{\Pi^{\mathrm{DB}}_{\leq t}-\domProj{W}_{\leq t}} \cdot
				\br{C\reg{A} \tp Q[C,D]\reg{LR}}
		}}_\infty^{1/2}\\
		&\leq \br{16t\sqrt{\frac{2t}{N}}}^{1/2} \leq \frac{6t^{3/4}}{N^{1/4}} \enspace,
	\end{align*}}
	
	\noindent where the last line follows from \lem{twirl}.
	
	Now, putting everything together we show the claim for $t+1$ to conclude:
	\begin{align*}
		\re\Br{
			\bra{\A_{t+1}^{W, \D}} \cdot \cQ \cdot
			\br{\ket{\A_{t+1}^V}\ket{\init(\D)}}
		} &= (*)+(**)\\
		&\geq (*) - 32\sqrt{\frac{t(t+1)}{N}}\\
		&\geq (\triangle)-(\triangle\triangle) - 32\sqrt{\frac{t(t+1)}{N}}\\
		&\geq 1 - \frac{38t^2}{N^{1/4}} - \frac{6t^{3/4}}{N^{1/4}} - 32\sqrt{\frac{t(t+1)}{N}} \\
		&\geq 1-\frac{1}{N^{1/4}}\br{38t^2+6t^{3/4}+32\frac{\sqrt{t(t+1)}}{N^{1/4}}}\\
		&\geq 1- \frac{1}{N^{1/4}}\br{ 38t^2+6t+32(t+1) } \\
		&\geq 1 - \frac{38(t+1)^2}{N^{1/4}} \enspace. \qedhere
	\end{align*}
\end{proof}

\section{A Twirling Lemma}
\label{sec:twirl}
In this section, we prove \lem{twirl}. We first analyze $\Pi^{\dom{W}}\reg{LR}$ and $\Pi^{\im{W}}\reg{LR}$ and then give the proof.

\subsection{Properties of \texorpdfstring{$\Pi^{\dom{W}}$ and $\Pi^{\im{W}}$}{the Domain and Image Projectors of W}}
We need some notation:
\begin{align*}
	\Pi^{\notin\mathrm{Dom}}\reg{ALR} &\coloneq \sum_{\substack{L,R\in\RR \\ x\notin \BDom{L\cup R} }} \ketbra{x}\reg{A}\otimes\ketbra{L}\reg{L}\otimes\ketbra{R}\reg{R} \enspace,\\
	\Pi^{\notin\mathrm{Im}}\reg{ALR} &\coloneq \sum_{\substack{L,R\in\RR \\ y\notin \BIm{L\cup R} }} \ketbra{y}\reg{A}\otimes\ketbra{L}\reg{L}\otimes\ketbra{R}\reg{R} \enspace,\\
	\Pi^{\mathrm{EPR}} &\coloneq \frac{1}{N}\sum_{x,y\in\bit{n}} \ketbratwo{x,x}{y,y} \enspace.
\end{align*}
We have the following lemma:
\begin{lemma}\label{lem:dom_im_W}
	{\small
    \begin{align*}
		\Pi^{\dom{W}}\reg{LR} &= \Pi^{\mathrm{DB}}\reg{LR} \cdot
		\br{ \Pi^{\notin\mathrm{Dom}}\reg{ALR} +
		\sum_{\substack{l,r\geq 0\\l+r<N/2}} \frac{N}{N-2l-2r}\cdot\Pi_{l, \textcolor{Black}{\mathsf{L}}}
		\otimes \sum_{i\in[r+1]} \Pi^{\mathrm{EPR}}_{\textcolor{Black}{\mathsf{A},\mathsf{R}^{(r+1)}_{\mathsf{X},i}}}
		} \cdot \Pi^{\mathrm{DB}}\reg{LR} \enspace,\\
		\Pi^{\im{W}}\reg{LR} &= \Pi^{\mathrm{DB}}\reg{LR} \cdot
		\br{ \Pi^{\notin\mathrm{Im}}\reg{ALR} +
		\sum_{\substack{l,r\geq 0\\l+r<N/2}} \frac{N}{N-2l-2r}\cdot\Pi_{l, \textcolor{Black}{\mathsf{R}}}
		\otimes \sum_{i\in[r+1]} \Pi^{\mathrm{EPR}}_{\textcolor{Black}{\mathsf{A},\mathsf{L}^{(r+1)}_{\mathsf{Y},i}}}
		} \cdot \Pi^{\mathrm{DB}}\reg{LR} \enspace.
	\end{align*} }

\noindent Here,
$\Pi^{\mathrm{EPR}}_{\textcolor{Black}{\mathsf{A},\mathsf{R}^{(r+1)}_{\mathsf{X},i}}}$ denotes the operator that acts on $\mathsf{A}$ and $\mathsf{R}^{(r+1)}$ such that
it applies $\Pi^{\mathrm{EPR}}$ to $\mathsf{A},\mathsf{R}^{(r+1)}_{\mathsf{X},i}$,
while applying the identity to the remainder of  $\mathsf{R}^{(r+1)}$.
\end{lemma}

\begin{proof}
	We show the first equality, and the other one follows from a similar argument.
	Note that $\Pi^{\dom{W}} = \Pi^{\dom{W^L}} + \Pi^{\im{W^R}}$. So, it suffices to show
	\begin{align*}
		\Pi^{\dom{W^L}} &= \Pi^{\mathrm{DB}}\reg{LR} \cdot
		\Pi^{\notin\mathrm{Dom}}\reg{ALR} \cdot \Pi^{\mathrm{DB}}\reg{LR} \enspace,\\
		\Pi^{\im{W^R}} &= \Pi^{\mathrm{DB}}\reg{LR} \cdot
		\br{
		\sum_{\substack{l,r\geq 0\\l+r<N/2}} \frac{N}{N-2l-2r}\cdot \Pi_{l, \textcolor{Black}{\mathsf{L}}}
		\otimes \sum_{i\in[r+1]} \Pi^{\mathrm{EPR}}_{\textcolor{Black}{\mathsf{A},\mathsf{R}^{(r+1)}_{\mathsf{X},i}}}
		} \cdot \Pi^{\mathrm{DB}}\reg{LR} \enspace.
	\end{align*}
	It is easy to see that
	\[
		\Pi^{\dom{W^L}} = \sum_{\substack{L\cup R\in\DBR \\ x\notin \dom{L\cup R} }} \ketbra{x}\reg{A}\otimes\ketbra{L}\reg{L}\otimes\ketbra{R}\reg{R} = \Pi^{\mathrm{DB}}\reg{LR} \cdot
		\Pi^{\notin\mathrm{Dom}}\reg{ALR} \cdot \Pi^{\mathrm{DB}}\reg{LR} \enspace.
	\]
	We now turn to proving the second equality. By definition of $W^R$, we know
	\[
		\Pi^{\im{W^R}} = W^R\cdot W^{R,\dag} = W^R\cdot \sum_{\substack{l,r\geq 0,\\ l+r<N/2}} \Pi_{l,r,\textcolor{Black}{\mathsf{LR}}} \cdot W^{R,\dag} = \sum_{\substack{l,r\geq 0,\\ l+r<N/2}}  W^R_{l,r}\cdot W^{R,\dag}_{l,r} \enspace.
	\]
	Therefore, we are left to show
	\[
		W^R_{l,r}\cdot W^{R,\dag}_{l,r} = \Pi^{\mathrm{DB}}_{l,r+1,\textcolor{Black}{\mathsf{LR}}} \cdot
		\br{ \frac{N}{N-2l-2r}\cdot
		\Pi_{l, \textcolor{Black}{\mathsf{L}}}
		\otimes
		\sum_{i\in[r+1]} \Pi^{\mathrm{EPR}}_{\textcolor{Black}{\mathsf{A},\mathsf{R}^{(r+1)}_{\mathsf{X},i}}}
		} \cdot \Pi^{\mathrm{DB}}_{l,r+1,\textcolor{Black}{\mathsf{LR}}} \enspace.
	\]
	Recall the operator $E^R$ in \defi{opE}. It is not hard to check
	\[
		W^R_{l,r} = \frac{\sqrt{N}}{\sqrt{N-2l-2r}}\cdot \Pi^{\mathrm{DB}}_{l,r+1,\textcolor{Black}{\mathsf{LR}}} \cdot E^R_{l,r}\enspace.
	\]
	Therefore,
	\begin{align*}
		W^R_{l,r}\cdot W^{R,\dag}_{l,r} &= \frac{N}{N-2l-2r}\cdot \Pi^{\mathrm{DB}}_{l,r+1,\textcolor{Black}{\mathsf{LR}}} \cdot E^R_{l,r} \cdot E^{R,\dag}_{l,r} \cdot \Pi^{\mathrm{DB}}_{l,r+1,\textcolor{Black}{\mathsf{LR}}} \\
		&= \frac{N}{N-2l-2r}\cdot\Pi^{\mathrm{DB}}_{l,r+1,\textcolor{Black}{\mathsf{LR}}} \cdot
		\br{
		\Pi_{l, \textcolor{Black}{\mathsf{L}}}
		\otimes
		\sum_{i\in[r+1]} \Pi^{\mathrm{EPR}}_{\textcolor{Black}{\mathsf{A},\mathsf{R}^{(r+1)}_{\mathsf{X},i}}}
		} \cdot \Pi^{\mathrm{DB}}_{l,r+1,\textcolor{Black}{\mathsf{LR}}} \enspace,
	\end{align*}
	where the second inequality is from \cite[Eq. (11.26)]{MH25}.
\end{proof}

For $l,r\geq 0$, we define
\begin{align*}
	\Pi^{\RR^2}\reg{LR}&\coloneq \sum_{L,R\in\RR} \ketbra{L}\reg{L}\otimes\ketbra{R}\reg{R}\enspace,\\
	\Pi^{\mathrm{db}}_{l,r,\textcolor{Black}{\mathsf{LR}}} &\coloneq
	\sum_{\substack{(x_1,\cdots,x_l,x'_1,\cdots,x'_l)\in\newDB_{l+r}\\(y_1,\cdots,y_l,y'_1,\cdots,y'_l)\in\newDB_{l+r}}}
	\ketbra{x_1,\cdots,x_l,y_1,\cdots,y_l}\reg{L}\\
	&\quad\quad\quad\quad\quad\quad\quad\quad\quad\quad\otimes\ketbra{x'_1,\cdots,x'_r,y'_1,\cdots,y'_r}\reg{R}\enspace,\\
	\Pi^{\mathrm{db}}\reg{LR} &\coloneq \sum_{\substack{l,r\geq 0\\l+r\leq N}} \Pi^{\mathrm{db}}_{l,r,\textcolor{Black}{\mathsf{LR}}}\enspace.
\end{align*}
It is evident that
\[
\Pi^{\mathrm{DB}}\reg{LR} = \Pi^{\RR^2}\reg{LR}\cdot \Pi^{\mathrm{db}}\reg{LR} = \Pi^{\mathrm{db}}\reg{LR}\cdot \Pi^{\RR^2}\reg{LR} \enspace.
\]
Now we define
{\small
\begin{align*}
		J^{\dom{W}}\reg{LR} &\coloneq \Pi^{\mathrm{db}}\reg{LR} \cdot
		\br{ \Pi^{\notin\mathrm{Dom}}\reg{ALR} +
		\sum_{\substack{l,r\geq 0\\l+r<N/2}} \frac{N}{N-2l-2r}\cdot\Pi_{l, \textcolor{Black}{\mathsf{L}}}
		\otimes \sum_{i\in[r+1]} \Pi^{\mathrm{EPR}}_{\textcolor{Black}{\mathsf{A},\mathsf{R}^{(r+1)}_{\mathsf{X},i}}}
		} \cdot \Pi^{\mathrm{db}}\reg{LR} \enspace,\\
		J^{\im{W}}\reg{LR} &\coloneq \Pi^{\mathrm{db}}\reg{LR} \cdot
		\br{ \Pi^{\notin\mathrm{Im}}\reg{ALR} +
		\sum_{\substack{l,r\geq 0\\l+r<N/2}} \frac{N}{N-2l-2r}\cdot\Pi_{l, \textcolor{Black}{\mathsf{R}}}
		\otimes \sum_{i\in[r+1]} \Pi^{\mathrm{EPR}}_{\textcolor{Black}{\mathsf{A},\mathsf{L}^{(r+1)}_{\mathsf{Y},i}}}
		} \cdot \Pi^{\mathrm{db}}\reg{LR} \enspace.
\end{align*}}

\noindent Then we have
\begin{align*}
	\Pi^{\dom{W}}\reg{LR} &= \Pi^{\RR^2}\reg{LR}\cdot J^{\dom{W}}\reg{LR}\cdot\Pi^{\RR^2}\reg{LR}\enspace,\\
	\Pi^{\im{W}}\reg{LR} &= \Pi^{\RR^2}\reg{LR}\cdot J^{\im{W}}\reg{LR}\cdot\Pi^{\RR^2}\reg{LR}\enspace.
\end{align*}

We will need the following lemma when proving \lem{twirl}.

\begin{lemma} \label{lem:ubound}
	For non-negative $l,r$ such that $l+r<N/2$,
	\begin{align*}
		\Pi^{\mathrm{db}}_{l,r,\textcolor{black}{\mathsf{LR}}} &- J^{\dom{W}}_{l,r,\textcolor{black}{\mathsf{LR}}}
		\preceq \sum_{i\in[l]}\eqProj_{\textcolor{black}{\mathsf{A},\mathsf{L}^{(l)}_{\mathsf{X},i}}}
		+
		\sum_{i\in[l]} \ffbProj_{\textcolor{black}{\mathsf{A},\mathsf{L}^{(l)}_{\mathsf{X},i}}}
		+
		\sum_{i\in[r]} \ffbProj_{\textcolor{black}{\mathsf{A},\mathsf{R}^{(r)}_{\mathsf{X},i}}} 
		\nonumber \\
		&+
		\frac{N}{N-2l-2r+2}\cdot
		\sum_{i\in[r]}
		\br{\br{\eqProj_{\textcolor{black}{\mathsf{A},\mathsf{R}^{(r)}_{\mathsf{X},i}}}
		- \Pi^{\mathrm{EPR}}_{\textcolor{black}{\mathsf{A},\mathsf{R}^{(r)}_{\mathsf{X},i}}}} +  2\sqrt{\frac{2(l+r)}{N}}\cdot\id\reg{ALR}} \enspace,\\
		\Pi^{\mathrm{db}}_{l,r,\textcolor{black}{\mathsf{LR}}} &- J^{\im{W}}_{l,r,\textcolor{black}{\mathsf{LR}}}
		 \preceq  
		\sum_{i\in[l]} \ffbProj_{\textcolor{black}{\mathsf{A},\mathsf{L}^{(l)}_{\mathsf{Y},i}}}
		+
		\sum_{i\in[r]}\eqProj_{\textcolor{black}{\mathsf{A},\mathsf{R}^{(r)}_{\mathsf{Y},i}}}
		+
		\sum_{i\in[r]} \ffbProj_{\textcolor{black}{\mathsf{A},\mathsf{R}^{(r)}_{\mathsf{Y},i}}} 
		\nonumber \\
		& +
		\frac{N}{N-2l-2r+2}\cdot
		\sum_{i\in[r]}
		\br{\br{\eqProj_{\textcolor{black}{\mathsf{A},\mathsf{L}^{(l)}_{\mathsf{Y},i}}}
		- \Pi^{\mathrm{EPR}}_{\textcolor{black}{\mathsf{A},\mathsf{L}^{(l)}_{\mathsf{Y},i}}}} +  2\sqrt{\frac{2(l+r)}{N}}\cdot\id\reg{ALR}}
		\enspace,
	\end{align*}
    where
	\[\eqProj = \sum_{x\in\bit{n}}\ketbra{x} \tp \ketbra{x}\enspace,\enspace\ffbProj = \sum_{x\in\bit{n}}\ketbra{x} \tp \ketbra{\bar{x}}\enspace.\]
\end{lemma}

\begin{proof}
	We prove the first inequality, and the other one follows from a similar argument.
	Notice that
	\[
		J^{\dom{W}}_{l,r,\textcolor{Black}{\mathsf{LR}}} =
		\Pi^{\mathrm{db}}_{l,r,\textcolor{Black}{\mathsf{LR}}} \cdot
		\br{ \Pi^{\notin\mathrm{Dom}}_{l,r,\textcolor{Black}{\mathsf{ALR}}} +\frac{N}{N-2l-2r+2}\cdot \sum_{i\in[r]} \Pi^{\mathrm{EPR}}_{\textcolor{Black}{\mathsf{A},\mathsf{R}^{(r)}_{\mathsf{X},i}}}
		} \cdot \Pi^{\mathrm{db}}_{l,r,\textcolor{Black}{\mathsf{LR}}}
	\]
	Therefore, we have
	\begin{align*}
		&\Pi^{\mathrm{db}}_{l,r,\textcolor{black}{\mathsf{LR}}} - J^{\dom{W}}_{l,r,\textcolor{black}{\mathsf{LR}}}\nonumber\\
		= ~&  \Pi^{\mathrm{db}}_{l,r,\textcolor{black}{\mathsf{LR}}} \cdot
		\br{ \Pi_{l,r,\textcolor{black}{\mathsf{LR}}}
		- \Pi^{\notin\mathrm{Dom}}_{l,r,\textcolor{black}{\mathsf{ALR}}} 
		-\frac{N}{N-2l-2r+2}\cdot \sum_{i\in[r]} \Pi^{\mathrm{EPR}}_{\textcolor{black}{\mathsf{A},\mathsf{R}^{(r)}_{\mathsf{X},i}}}
		} \cdot \Pi^{\mathrm{db}}_{l,r,\textcolor{black}{\mathsf{LR}}} \nonumber\\
		\preceq ~& \Pi^{\mathrm{db}}_{l,r,\textcolor{black}{\mathsf{LR}}} \cdot
		\left(
		\sum_{i\in[l]} \eqProj_{\textcolor{black}{\mathsf{A},\mathsf{L}^{(l)}_{\mathsf{X},i}}}
		+
		\sum_{i\in[l]} \ffbProj_{\textcolor{black}{\mathsf{A},\mathsf{L}^{(l)}_{\mathsf{X},i}}}
		+
		\sum_{i\in[r]} \eqProj_{\textcolor{black}{\mathsf{A},\mathsf{R}^{(r)}_{\mathsf{X},i}}}
		+
		\sum_{i\in[r]} \ffbProj_{\textcolor{black}{\mathsf{A},\mathsf{R}^{(r)}_{\mathsf{X},i}}}
		\right. \nonumber \\
		&\qquad\qquad\qquad\qquad\qquad\left.
		-~\frac{N}{N-2l-2r+2}\cdot \sum_{i\in[r]} \Pi^{\mathrm{EPR}}_{\textcolor{black}{\mathsf{A},\mathsf{R}^{(r)}_{\mathsf{X},i}}}
		\right) \cdot \Pi^{\mathrm{db}}_{l,r,\textcolor{black}{\mathsf{LR}}}\enspace.
	\end{align*}
	Since $\eqProj$, $\ffbProj$ and $\Pi^{\mathrm{db}}_{l,r}$ commute with each other,
	we have
	\begin{align*}
		\Pi^{\mathrm{db}}_{l,r,\textcolor{black}{\mathsf{LR}}} - & J^{\dom{W}}_{l,r,\textcolor{black}{\mathsf{LR}}} \nonumber\\
		\preceq ~
		&\sum_{i\in[l]} \eqProj_{\textcolor{black}{\mathsf{A},\mathsf{L}^{(l)}_{\mathsf{X},i}}}
		+
		\sum_{i\in[l]} \ffbProj_{\textcolor{black}{\mathsf{A},\mathsf{L}^{(l)}_{\mathsf{X},i}}}
		+
		\sum_{i\in[r]} \ffbProj_{\textcolor{black}{\mathsf{A},\mathsf{R}^{(r)}_{\mathsf{X},i}}} \nonumber\\
		&\quad +
		\Pi^{\mathrm{db}}_{l,r,\textcolor{black}{\mathsf{LR}}} \cdot
		\br{
		\sum_{i\in[r]} \eqProj_{\textcolor{black}{\mathsf{A},\mathsf{R}^{(r)}_{\mathsf{X},i}}}
		-\frac{N}{N-2l-2r+2}\cdot \sum_{i\in[r]} \Pi^{\mathrm{EPR}}_{\textcolor{black}{\mathsf{A},\mathsf{R}^{(r)}_{\mathsf{X},i}}}
		} \cdot \Pi^{\mathrm{db}}_{l,r,\textcolor{black}{\mathsf{LR}}}\nonumber\\
		\preceq ~&
		\sum_{i\in[l]}\eqProj_{\textcolor{black}{\mathsf{A},\mathsf{L}^{(l)}_{\mathsf{X},i}}}
		+
		\sum_{i\in[l]} \ffbProj_{\textcolor{black}{\mathsf{A},\mathsf{L}^{(l)}_{\mathsf{X},i}}}
		+
		\sum_{i\in[r]} \ffbProj_{\textcolor{black}{\mathsf{A},\mathsf{R}^{(r)}_{\mathsf{X},i}}} +
		\frac{N}{N-2l-2r+2}\nonumber\\
		&\quad\cdot
		\sum_{i\in[r]} \Pi^{\mathrm{db}}_{l,r,\textcolor{black}{\mathsf{LR}}} \cdot
		\br{\eqProj_{\textcolor{black}{\mathsf{A},\mathsf{R}^{(r)}_{\mathsf{X},i}}}
		- \Pi^{\mathrm{EPR}}_{\textcolor{black}{\mathsf{A},\mathsf{R}^{(r)}_{\mathsf{X},i}}}}
		\cdot \Pi^{\mathrm{db}}_{l,r,\textcolor{black}{\mathsf{LR}}} \nonumber\\
		\preceq ~&
		\sum_{i\in[l]}\eqProj_{\textcolor{black}{\mathsf{A},\mathsf{L}^{(l)}_{\mathsf{X},i}}}
		+
		\sum_{i\in[l]} \ffbProj_{\textcolor{black}{\mathsf{A},\mathsf{L}^{(l)}_{\mathsf{X},i}}}
		+
		\sum_{i\in[r]} \ffbProj_{\textcolor{black}{\mathsf{A},\mathsf{R}^{(r)}_{\mathsf{X},i}}} 
		+
		\frac{N}{N-2l-2r+2} \nonumber\\
		&\quad\cdot
		\sum_{i\in[r]}
		\br{\br{\eqProj_{\textcolor{black}{\mathsf{A},\mathsf{R}^{(r)}_{\mathsf{X},i}}}
		- \Pi^{\mathrm{EPR}}_{\textcolor{black}{\mathsf{A},\mathsf{R}^{(r)}_{\mathsf{X},i}}}}
		\Pi^{\mathrm{db}}_{l,r,\textcolor{black}{\mathsf{LR}}}
		\br{\eqProj_{\textcolor{black}{\mathsf{A},\mathsf{R}^{(r)}_{\mathsf{X},i}}}
		- \Pi^{\mathrm{EPR}}_{\textcolor{black}{\mathsf{A},\mathsf{R}^{(r)}_{\mathsf{X},i}}}}
		+ \lambda \cdot\id\reg{ALR}} \nonumber\\
		\preceq ~&
		\sum_{i\in[l]}\eqProj_{\textcolor{black}{\mathsf{A},\mathsf{L}^{(l)}_{\mathsf{X},i}}}
		+
		\sum_{i\in[l]} \ffbProj_{\textcolor{black}{\mathsf{A},\mathsf{L}^{(l)}_{\mathsf{X},i}}}
		+
		\sum_{i\in[r]} \ffbProj_{\textcolor{black}{\mathsf{A},\mathsf{R}^{(r)}_{\mathsf{X},i}}} 
		+
		\frac{N}{N-2l-2r+2} \nonumber\\
		&\quad\cdot
		\sum_{i\in[r]}
		\br{\br{\eqProj_{\textcolor{black}{\mathsf{A},\mathsf{R}^{(r)}_{\mathsf{X},i}}}
		- \Pi^{\mathrm{EPR}}_{\textcolor{black}{\mathsf{A},\mathsf{R}^{(r)}_{\mathsf{X},i}}}} + \lambda \cdot\id\reg{ALR}} \enspace,
	\end{align*}
	where 
	\[
        \lambda\coloneqq
		\norm{
		\Pi^{\mathrm{db}}_{l,r,\textcolor{black}{\mathsf{LR}}} \cdot
		\Delta_i
		\cdot \Pi^{\mathrm{db}}_{l,r,\textcolor{black}{\mathsf{LR}}}
		-
		\Delta_i \cdot
		\Pi^{\mathrm{db}}_{l,r,\textcolor{black}{\mathsf{LR}}} \cdot
		\Delta_i
		}_\infty\enspace,
	\]
    and
    $\Delta_i
		\coloneq
		\eqProj_{\mathsf{A},\mathsf{R}^{(r)}_{\mathsf{X},i}}
		-
		\Pi^{\mathrm{EPR}}_{\mathsf{A},
		\mathsf{R}^{(r)}_{\mathsf{X},i}}$.
	Thus, we are left to show that $\lambda$ is bounded by $2\sqrt{\frac{2(l+r)}{N}}$.
	Notice that
	\begin{align*}
		\lambda ~&
		\leq ~\norm{
		\Pi^{\mathrm{db}}_{l,r,\textcolor{black}{\mathsf{LR}}} \cdot
		\Delta_i
		\cdot \Pi^{\mathrm{db}}_{l,r,\textcolor{black}{\mathsf{LR}}}
		-
		\Delta_i \cdot \Pi^{\mathrm{db}}_{l,r,\textcolor{black}{\mathsf{LR}}}\cdot\Delta_i \cdot \Pi^{\mathrm{db}}_{l,r,\textcolor{black}{\mathsf{LR}}}
		}_\infty \nonumber\\
		&\qquad\quad\ +~
		\norm{
		 \Delta_i \cdot \Pi^{\mathrm{db}}_{l,r,\textcolor{black}{\mathsf{LR}}}\cdot\Delta_i \cdot \Pi^{\mathrm{db}}_{l,r,\textcolor{black}{\mathsf{LR}}}
		-
		\Delta_i \cdot
		\Pi^{\mathrm{db}}_{l,r,\textcolor{black}{\mathsf{LR}}}\cdot
		\Delta_i
		}_\infty \nonumber \\
		&= ~\norm{
		\Pi^{\mathrm{db}}_{l,r,\textcolor{black}{\mathsf{LR}}} \cdot
		\Delta_i^2
		\cdot \Pi^{\mathrm{db}}_{l,r,\textcolor{black}{\mathsf{LR}}}
		-
		\Delta_i \cdot \Pi^{\mathrm{db}}_{l,r,\textcolor{black}{\mathsf{LR}}}\cdot\Delta_i \cdot \Pi^{\mathrm{db}}_{l,r,\textcolor{black}{\mathsf{LR}}}
		}_\infty \nonumber\\
		&\qquad\quad\ +~
		\norm{
		 \Delta_i \cdot \Pi^{\mathrm{db}}_{l,r,\textcolor{black}{\mathsf{LR}}}\cdot\Delta_i \cdot \Pi^{\mathrm{db}}_{l,r,\textcolor{black}{\mathsf{LR}}}
		-
		\Delta_i \cdot
		\Pi^{\mathrm{db}^2}_{l,r,\textcolor{black}{\mathsf{LR}}}\cdot
		\Delta_i
		}_\infty \nonumber\\
		&\leq~ 2\norm{ \Pi^{\mathrm{db}}_{l,r,\textcolor{black}{\mathsf{LR}}} \cdot
		\Delta_i -
		\Delta_i \cdot \Pi^{\mathrm{db}}_{l,r,\textcolor{black}{\mathsf{LR}}}
		}_\infty \nonumber\\
		&= ~2\norm{ \Pi^{\mathrm{db}}_{l,r,\textcolor{black}{\mathsf{LR}}} \cdot
		\Pi^{\mathrm{EPR}}_{\textcolor{black}{\mathsf{A},\mathsf{R}^{(r)}_{\mathsf{X},i}}} -
		\Pi^{\mathrm{EPR}}_{\textcolor{black}{\mathsf{A},\mathsf{R}^{(r)}_{\mathsf{X},i}}} \cdot \Pi^{\mathrm{db}}_{l,r,\textcolor{black}{\mathsf{LR}}}
		}_\infty \enspace.
	\end{align*}
	So, it suffices to prove that
	\[
		\norm{ \Pi^{\mathrm{db}}_{l,r,\textcolor{Black}{\mathsf{LR}}} \cdot
		\Pi^{\mathrm{EPR}}_{\textcolor{Black}{\mathsf{A},\mathsf{R}^{(r)}_{\mathsf{X},i}}} -
		\Pi^{\mathrm{EPR}}_{\textcolor{Black}{\mathsf{A},\mathsf{R}^{(r)}_{\mathsf{X},i}}} \cdot \Pi^{\mathrm{db}}_{l,r,\textcolor{Black}{\mathsf{LR}}}
		}_\infty \leq \sqrt{\frac{2(l+r)}{N}} \enspace.
	\]
	If $r=0$, this holds trivially. From now on, we assume $r\geq 1$.
	Note that $\Pi^{\mathrm{db}}_{l,r,\textcolor{Black}{\mathsf{LR}}} \cdot
		\Pi^{\mathrm{EPR}}_{\textcolor{Black}{\mathsf{A},\mathsf{R}^{(r)}_{\mathsf{X},i}}}$
    is 
	\begin{align*}
		\sum_{\substack{(x,x')\in\newDB_{l+r-1}\\ (y,y')\in\newDB_{l+r}}}
		\ketbra{x,y}\reg{L}
		\otimes\ketbra{x',y'} _{\textcolor{Black}{\mathsf{R}\backslash\mathsf{R}^{(r)}_{\mathsf{X},i}}} \otimes \frac{1}{N}
		\sum_{\substack{z \in T_{x,x'}\\ w\in \bit{n}}} \ketbratwo{z,z}{w,w}_{\textcolor{Black}{\mathsf{A},\mathsf{R}^{(r)}_{\mathsf{X},i}}} \enspace.
	\end{align*}
	Here, $x\in\bit{nl}$, $x'\in\bit{n(r-1)}$, $y\in\bit{nl},\ y'\in\bit{nr}$ and $T_{x,x'}$ denotes the set of all binary strings that do not belong to the same block as any string in $x$ and $x'$.
	Similarly, $\Pi^{\mathrm{EPR}}_{\textcolor{Black}{\mathsf{A},\mathsf{R}^{(r)}_{\mathsf{X},i}}}\cdot
		\Pi^{\mathrm{db}}_{l,r,\textcolor{Black}{\mathsf{LR}}}$ is
	\begin{align*}
    \sum_{\substack{(x,x')\in\newDB_{l+r-1}\\ (y,y')\in\newDB_{l+r}}}
		\ketbra{x,y}\reg{L}
		\otimes\ketbra{x',y'} _{\textcolor{Black}{\mathsf{R}\backslash\mathsf{R}^{(r)}_{\mathsf{X},i}}} \otimes \frac{1}{N}
		\sum_{\substack{w \in T_{x,x'}\\ z\in \bit{n}}} \ketbratwo{z,z}{w,w}_{\textcolor{Black}{\mathsf{A},\mathsf{R}^{(r)}_{\mathsf{X},i}}} \enspace.
	\end{align*}
	Then, we have
	\begin{align*}
		&\Pi^{\mathrm{db}}_{l,r,\textcolor{black}{\mathsf{LR}}} \cdot 
		\Pi^{\mathrm{EPR}}_{\textcolor{black}{\mathsf{A},\mathsf{R}^{(r)}_{\mathsf{X},i}}} - \Pi^{\mathrm{EPR}}_{\textcolor{black}{\mathsf{A},\mathsf{R}^{(r)}_{\mathsf{X},i}}}\cdot
		\Pi^{\mathrm{db}}_{l,r,\textcolor{black}{\mathsf{LR}}} \nonumber \\
		= ~&\sum_{\substack{(x,x')\in\newDB_{l+r-1}\\ (y,y')\in\newDB_{l+r}}}
		\ketbra{x,y}\reg{L}
		\otimes\ketbra{x',y'} _{\textcolor{black}{\mathsf{R}\backslash\mathsf{R}^{(r)}_{\mathsf{X},i}}} \nonumber\\
		&\qquad \otimes\br{
		\frac{1}{N}
		\sum_{\substack{z \in T_{x,x'}\\ w\in \bit{n}}} \ketbratwo{z,z}{w,w}_{\textcolor{black}{\mathsf{A},\mathsf{R}^{(r)}_{\mathsf{X},i}}} 
		- \frac{1}{N}
		\sum_{\substack{w \in T_{x,x'}\\ z\in \bit{n}}} \ketbratwo{z,z}{w,w}_{\textcolor{black}{\mathsf{A},\mathsf{R}^{(r)}_{\mathsf{X},i}}}
		} \enspace.
	\end{align*}
	Therefore,
    \begin{align*}
        &\norm{ \Pi^{\mathrm{db}}_{l,r,\textcolor{black}{\mathsf{LR}}} \cdot
		\Pi^{\mathrm{EPR}}_{\textcolor{black}{\mathsf{A},\mathsf{R}^{(r)}_{\mathsf{X},i}}} -
		\Pi^{\mathrm{EPR}}_{\textcolor{black}{\mathsf{A},\mathsf{R}^{(r)}_{\mathsf{X},i}}} \cdot \Pi^{\mathrm{db}}_{l,r,\textcolor{black}{\mathsf{LR}}}
		}_\infty 
		\nonumber\\
		\leq~& \max_{x,x'} ~\norm{\frac{1}{N}
		\sum_{\substack{z \in T_{x,x'}\\ w\in \bit{n}}} \ketbratwo{z,z}{w,w}_{\textcolor{black}{\mathsf{A},\mathsf{R}^{(r)}_{\mathsf{X},i}}} 
		- \frac{1}{N}
		\sum_{\substack{w \in T_{x,x'}\\ z\in \bit{n}}} \ketbratwo{z,z}{w,w}_{\textcolor{black}{\mathsf{A},\mathsf{R}^{(r)}_{\mathsf{X},i}}}}_\infty 
	 	\enspace.
    \end{align*}
	Notice that
	\begin{align*}
		&\norm{\frac{1}{N}
		\sum_{\substack{z \in T_{x,x'}\\ w\in \bit{n}}} \ketbratwo{z,z}{w,w}_{\textcolor{Black}{\mathsf{A},\mathsf{R}^{(r)}_{\mathsf{X},i}}} 
		- \frac{1}{N}
		\sum_{\substack{w \in T_{x,x'}\\ z\in \bit{n}}} \ketbratwo{z,z}{w,w}_{\textcolor{Black}{\mathsf{A},\mathsf{R}^{(r)}_{\mathsf{X},i}}}}_\infty \\
		=~& \norm{\frac{1}{N}
		\sum_{\substack{z \in T_{x,x'}\\ w\notin T_{x,x'}}} \ketbratwo{z,z}{w,w}_{\textcolor{Black}{\mathsf{A},\mathsf{R}^{(r)}_{\mathsf{X},i}}} 
		- \frac{1}{N}
		\sum_{\substack{w \in T_{x,x'}\\ z\notin T_{x,x'}}} \ketbratwo{z,z}{w,w}_{\textcolor{Black}{\mathsf{A},\mathsf{R}^{(r)}_{\mathsf{X},i}}}}_\infty \\
		=~& \frac{\sqrt{(N-2l-2r+2)(2l+2r-2)}}{N} \norm{\ketbratwo{\phi}{\psi} - \ketbratwo{\psi}{\phi}}_\infty \enspace,
	\end{align*}
	where
	\begin{align*}
		\ket{\phi} \coloneq \frac{1}{\sqrt{N-2l-2r+2}} \sum_{z\in T_{x,x'}} \ket{z,z}\enspace \text{and}\enspace
		\ket{\psi} \coloneq \frac{1}{\sqrt{2l+2r-2}} \sum_{z\notin T_{x,x'}} \ket{z,z}\enspace.
	\end{align*}
	Since $\ket{\phi}$ and $\ket{\psi}$ are orthogonal, it is not hard to see that $\norm{\ketbratwo{\phi}{\psi} - \ketbratwo{\psi}{\phi}}_\infty=1$.
	Therefore, 
	\begin{align*}
		\norm{ \Pi^{\mathrm{db}}_{l,r,\textcolor{black}{\mathsf{LR}}} \cdot
		\Pi^{\mathrm{EPR}}_{\textcolor{black}{\mathsf{A},\mathsf{R}^{(r)}_{\mathsf{X},i}}} -
		\Pi^{\mathrm{EPR}}_{\textcolor{black}{\mathsf{A},\mathsf{R}^{(r)}_{\mathsf{X},i}}} \cdot \Pi^{\mathrm{db}}_{l,r,\textcolor{black}{\mathsf{LR}}}
		}_\infty 
		&\leq \frac{\sqrt{(N-2l-2r+2)(2l+2r-2)}}{N} \nonumber\\
		&\leq \sqrt{\frac{2(l+r)}{N}} \enspace. \qedhere
	\end{align*}
	\end{proof}

\subsection{Proof of \texorpdfstring{\lem{twirl}}{Lemma 66}}
We need the following property of twirling with Haar random unitary.
\begin{lemma}[Claim 2 in \cite{MH25}]\label{lem:twirl2}
	Let $\D$ be the Haar measure over $\ugroup{N}$. We have
	\[
		\expect{U\leftarrow \D}{\br{U\otimes \widebar{U}}^\dag\cdot \eqProj\cdot \br{U\otimes \widebar{U}}}
		= \Pi^{\mathrm{EPR}} + \frac{1}{N+1} \br{ \id - \Pi^{\mathrm{EPR}}}\enspace.
	\]
\end{lemma}

Before proving \lem{twirl}, we prove a lemma that will assist in our argument.

\begin{lemma} \label{lem:simple_twirl}
	Let $\D\in\st{\D_1,\D_2}$ as defined in \defi{d1d2}. We have for non-negative $l,r$
    {\small
	\begin{align*}
		&\norm{\expect{C,D\leftarrow \D}{(C\reg{A} \otimes Q[C,D]\reg{LR})^\dag \cdot\br{
	\Pi^{\mathrm{db}}_{l,r,\textcolor{Black}{\mathsf{LR}}} - J^{\dom{W}}_{l,r,\textcolor{Black}{\mathsf{LR}}}
	}\cdot (C\reg{A} \otimes Q[C,D]\reg{LR})}}_{\infty}\\
    &\quad\quad\quad\quad\quad\quad\quad\quad\quad\quad\quad\quad\leq \frac{4l+r}{N-1} + \frac{7rN}{N-2l-2r+2}\cdot \sqrt{\frac{2(l+r)}{N}}\enspace,\\
	&\norm{\expect{C,D\leftarrow \D}{(D\reg{A}^\dag \otimes Q[C,D]\reg{LR})^\dag \cdot\br{
	\Pi^{\mathrm{db}}_{l,r,\textcolor{Black}{\mathsf{LR}}} - J^{\im{W}}_{l,r,\textcolor{Black}{\mathsf{LR}}}
	}\cdot (D\reg{A}^\dag \otimes Q[C,D]\reg{LR})}}_{\infty}\\
	&\quad\quad\quad\quad\quad\quad\quad\quad\quad\quad\quad\quad\leq \frac{l+4r}{N-1} + \frac{7lN}{N-2l-2r+2}\cdot \sqrt{\frac{2(l+r)}{N}}\enspace.
	\end{align*}}
\end{lemma}

\begin{proof}
	We show the first inequality for $\D = \D_2$ and the rest inequalities follow from a similar argument. From \lem{ubound} and the triangle inequality, we have
    {\small
	\begin{align*}
		&\norm{\expect{C,D\leftarrow \D_2}{(C\reg{A} \otimes Q[C,D]\reg{LR})^\dag \cdot\br{
	\Pi^{\mathrm{db}}_{l,r,\textcolor{Black}{\mathsf{LR}}} - J^{\dom{W}}_{l,r,\textcolor{Black}{\mathsf{LR}}}
	}\cdot (C\reg{A} \otimes Q[C,D]\reg{LR})}}_{\infty} \\
	\leq & 
	\norm{\expect{C,D\leftarrow \D_2}{(C\reg{A} \otimes Q[C,D]\reg{LR})^\dag \cdot
	\br{
	\sum_{i\in[l]}\eqProj_{\textcolor{Black}{\mathsf{A},\mathsf{L}^{(l)}_{\mathsf{X},i}}}
	}\cdot (C\reg{A} \otimes Q[C,D]\reg{LR})}}_{\infty}\\
	& +
	\norm{\expect{C,D\leftarrow \D_2}{(C\reg{A} \otimes Q[C,D]\reg{LR})^\dag \cdot
	\br{
	\sum_{i\in[l]} \ffbProj_{\textcolor{Black}{\mathsf{A},\mathsf{L}^{(l)}_{\mathsf{X},i}}}
	}\cdot (C\reg{A} \otimes Q[C,D]\reg{LR})}}_{\infty}\\
	& +
	\norm{\expect{C,D\leftarrow \D_2}{(C\reg{A} \otimes Q[C,D]\reg{LR})^\dag \cdot
	\br{
	\sum_{i\in[r]} \ffbProj_{\textcolor{Black}{\mathsf{A},\mathsf{R}^{(r)}_{\mathsf{X},i}}} 
	}\cdot (C\reg{A} \otimes Q[C,D]\reg{LR})}}_{\infty}\\
	& + \frac{N}{N-2l-2r+2}
    \\
    &\quad\cdot
	\norm{\expect{C,D\leftarrow \D_2}{(C\reg{A} \otimes Q[C,D]\reg{LR})^\dag \cdot
	\br{
	\sum_{i\in[r]}
		\br{\eqProj_{\textcolor{Black}{\mathsf{A},\mathsf{R}^{(r)}_{\mathsf{X},i}}}
		- \Pi^{\mathrm{EPR}}_{\textcolor{Black}{\mathsf{A},\mathsf{R}^{(r)}_{\mathsf{X},i}}}}
	}\cdot (C\reg{A} \otimes Q[C,D]\reg{LR})}}_{\infty}\\
	& +  \frac{2rN}{N-2l-2r+2}\cdot\sqrt{\frac{2(l+r)}{N}}\\
	\leq & \ 
	l\cdot\sum_{x\in\bit{n}}
	\underbrace{\norm{\expect{C}{(C\otimes C)^\dag \cdot
		\ketbra{x,x}\cdot (C\otimes C)}}_{\infty}}_{(1)} \\
	&
	+\ l\cdot
	\sum_{x\in\bit{n}}
	\underbrace{\norm{\expect{C}{(C\otimes C)^\dag \cdot\ketbra{x,\bar{x}}
	\cdot (C\otimes C)}}_{\infty}}_{(2)}\\
	& +\ r\cdot
	\sum_{x\in\bit{n}}
	\underbrace{\norm{\expect{C}{(C\otimes \widebar{C})^\dag \cdot
		\ketbra{x,\bar{x}}\cdot (C\otimes \widebar{C})}}_{\infty}}_{(3)}\\
	& + \frac{rN}{N-2l-2r+2}\cdot
	\underbrace{\norm{\expect{C}{(C\otimes \widebar{C})^\dag \cdot
		\br{\eqProj - \Pi^{\mathrm{EPR}}}
	\cdot (C\otimes \widebar{C})}}_{\infty}}_{(4)} \\
	&+ \frac{2rN}{N-2l-2r+2}\cdot\sqrt{\frac{2(l+r)}{N}} \enspace.
	\end{align*} }
	
	\noindent For (1), notice that $C^\dag$ is equivalent to a $T$-step parallel Kac's walk followed by an independent random permutation, which means $C^\dag$ is drawn from a $\frac{1}{N^2}-\rss$ distribution.
	Thus, we have
	\begin{align*}
		(1) \leq \norm{\expect{U\leftarrow\mu}{(U\otimes U) \cdot
		\ketbra{x,x}\cdot (U\otimes U)^\dag}}_{\infty} + \frac{1}{N^2} = \frac{2}{N(N+1)} + \frac{1}{N^2} \enspace.
	\end{align*}
	For (2), notice that $C^\dag = C'^{\dag} \cdot P^\dag$. So,
	\begin{align*}
		(2) \leq \norm{\expect{P}{(P\otimes P)^\dag \cdot\ketbra{x,\bar{x}}\cdot (P\otimes P)}}_{\infty} &= 
		\norm{\frac{1}{N(N-1)} \sum_{x\neq y} \ketbra{x,y}}_\infty \nonumber \\
		&=~\frac{1}{N(N-1)} \enspace.
	\end{align*}
	Similarly, we have $(3)\leq \frac{1}{N(N-1)}$.
	
	We now give an upper bound on (4). We first prove that for any $x\in\bit{n}$
	{\small\begin{align} \label{eq:cu}
		\norm{
		\expect{C}{(C\otimes \widebar{C})^\dag \cdot
		\ketbra{x,x}\cdot (C\otimes \widebar{C})}
		-
		\expect{U\leftarrow\mu}{(U\otimes \widebar{U})^\dag \cdot
		\ketbra{x,x}\cdot (U\otimes \widebar{U})}
		}_{\infty} \leq \frac{4}{N^2}\enspace.
	\end{align}}
	
	\noindent Notice that $C^\dag$ is equivalent to a $T$-step parallel Kac's walk followed by a random permutation $P'$.
	Let $\ket{u}$ be the random state after $T$ steps of parallel Kac's walk starting at $\ketbra{x}$, and $\ket{v}$ be the Haar random state. We have
	\begin{align*}
		&\norm{
		\expect{C}{(C\otimes \widebar{C})^\dag \cdot
		\ketbra{x,x}\cdot (C\otimes \widebar{C})}
		-
		\expect{U\leftarrow\mu}{(U\otimes \widebar{U})^\dag \cdot
		\ketbra{x,x}\cdot (U\otimes \widebar{U})}
		}_{\infty} \nonumber\\
		\leq ~&\norm{
		\expect{C}{(C\otimes \widebar{C})^\dag \cdot
		\ketbra{x,x}\cdot (C\otimes \widebar{C})}
		-
		\expect{U\leftarrow\mu}{(U\otimes \widebar{U})^\dag \cdot
		\ketbra{x,x}\cdot (U\otimes \widebar{U})}
		}_{1} \nonumber\\
		\leq ~& \left\lVert
			\expect{P',\ket{u}}{\br{P'\otimes \widebar{P'}}
			\br{\ketbra{u}\otimes\widebar{\ketbra{u}}}\br{P'\otimes \widebar{P'}}^\dag} 
			\right. \nonumber\\
			&\qquad\qquad\qquad\qquad\qquad\left.
			- \expect{P',\ket{v}}{\br{P'\otimes \widebar{P'}}
			\br{\ketbra{v}\otimes\widebar{\ketbra{v}}}\br{P'\otimes \widebar{P'}}^\dag}
		\right\rVert_1 \nonumber\\
		\leq ~& \norm{ \expect{\ket{u}}{
		\ketbra{u}\otimes\widebar{\ketbra{u}}}
		- \expect{\ket{v}}{
		\ketbra{v}\otimes\widebar{\ketbra{v}}}
		}_1 \enspace.
	\end{align*}
	From \cite[Theorem 4]{LQSY+24}, we know that there is a joint distribution $\gamma$ among $\ket{u}$ and $\ket{v}$ such that
		$\expect{\gamma}{\norm{\ket{u}-\ket{v}}_2} \leq \frac{1}{N^2}$.
	Therefore, we have
	\begin{align*}
		&\norm{ \expect{\ket{u}}{
		\ketbra{u}\otimes\widebar{\ketbra{u}}}
		- \expect{\ket{v}}{
		\ketbra{v}\otimes\widebar{\ketbra{v}}}
		}_1 \\
		=& 
		\norm{ \expect{\gamma}{
		\ketbra{u}\otimes\widebar{\ketbra{u}}-
		\ketbra{v}\otimes\widebar{\ketbra{v}}}
		}_1\\
		\leq &\norm{ \expect{\gamma}{
		\ketbra{u}\otimes\widebar{\ketbra{u}}-
		\ketbra{v}\otimes\widebar{\ketbra{u}}}
		}_1 +
		\norm{ \expect{\gamma}{
		\ketbra{v}\otimes\widebar{\ketbra{u}}-
		\ketbra{v}\otimes\widebar{\ketbra{v}}}
		}_1 \\
		\leq &\expect{\gamma}{\norm{
		\ketbra{u}\otimes\widebar{\ketbra{u}}-
		\ketbra{v}\otimes\widebar{\ketbra{u}}
		}_1 } +
		\expect{\gamma}{ \norm{
		\ketbra{v}\otimes\widebar{\ketbra{u}}-
		\ketbra{v}\otimes\widebar{\ketbra{v}}
		}_1 }\\
		= &2\cdot\expect{\gamma}{\norm{
		\ketbra{u}-
		\ketbra{v}
		}_1 } \\
		\leq & 2\cdot\br{ \expect{\gamma}{\norm{
		\ket{u}(\bra{u}-\bra{v} )
		}_1 + \norm{
		(\ket{u}-\ket{v})\bra{v}
		}_1 } } \\
		\leq & 4\cdot \expect{\gamma}{\norm{\ket{u}-\ket{v}}_2} \leq \frac{4}{N^2} \enspace.
	\end{align*}
	This establishes Eq. \eq{cu}.
	From Eq. \eq{cu}, \lem{twirl2} and the fact that $(U\otimes \widebar{U})^\dag \cdot\Pi^{\mathrm{EPR}}
	\cdot (U\otimes \widebar{U}) = \Pi^{\mathrm{EPR}}$ for any unitary $U$, we have
	\begin{align*}
		(4)\leq \norm{\expect{U\leftarrow\mu}{(U\otimes \widebar{U})^\dag \cdot
		\br{\eqProj - \Pi^{\mathrm{EPR}}}
	\cdot (U\otimes \widebar{U})}}_{\infty} + \frac{4}{N} 
	\leq  \frac{1}{N+1} + \frac{4}{N} \leq\frac{5}{N} \enspace.
	\end{align*}
	Therefore, we have
	\begin{align*}
		&\norm{\expect{C,D\leftarrow \D_2}{(C\reg{A} \otimes Q[C,D]\reg{LR})^\dag \cdot\br{
	\Pi^{\mathrm{db}}_{l,r,\textcolor{Black}{\mathsf{LR}}} - J^{\dom{W}}_{l,r,\textcolor{Black}{\mathsf{LR}}}
	}\cdot (C\reg{A} \otimes Q[C,D]\reg{LR})}}_{\infty} \\
	\leq ~& \frac{2l}{N+1} + \frac{l}{N} + \frac{l}{N-1} +\frac{r}{N-1} + \frac{5r}{N-2l-2r+2} \nonumber\\
	&\qquad\qquad\qquad\qquad\qquad\qquad\qquad\qquad + \frac{2rN}{N-2l-2r+2}\cdot\sqrt{\frac{2(l+r)}{N}}\nonumber\\
	\leq ~& \frac{4l+r}{N-1} + \frac{7rN}{N-2l-2r+2}\cdot \sqrt{\frac{2(l+r)}{N}} \enspace. \qedhere
	\end{align*}
\end{proof}

We are now ready to prove \lem{twirl}. We restate the lemma here.
\begin{replemma}{lem:twirl}
	For an integer $0\leq t\leq N/4$ and
	$\D\in\st{\D_1,\D_2}$ as defined in \defi{d1d2},
	we have
    {\small
	\begin{align*}
	    \norm{\expect{C,D\leftarrow \D}{(C\reg{A} \otimes Q[C,D]\reg{LR})^\dag \cdot\br{
				\Pi^{\mathrm{DB}}_{\leq t, \textcolor{black}{\mathsf{LR}}} -  
				\Pi^{\dom{W}}_{\leq t, \textcolor{black}{\mathsf{LR}}}
			}\cdot (C\reg{A} \otimes Q[C,D]\reg{LR})}}_{\infty} &\leq 16t\cdot \sqrt{\frac{2t}{N}}\enspace, \\
		\norm{\expect{C,D\leftarrow \D}{(D\reg{A}^\dag \otimes Q[C,D]\reg{LR})^\dag \cdot\br{
				\Pi^{\mathrm{DB}}_{\leq t, \textcolor{black}{\mathsf{LR}}} -  
				\Pi^{\im{W}}_{\leq t, \textcolor{black}{\mathsf{LR}}}
			}\cdot (D\reg{A}^\dag \otimes Q[C,D]\reg{LR})}}_{\infty} &\leq 16t\cdot \sqrt{\frac{2t}{N}}\enspace.
	\end{align*}}
\end{replemma}

\begin{proof}
	We prove the first inequality and the other follows from a similar argument.
	Note that
	\[
		\Pi^{\mathrm{DB}}_{\leq t, \textcolor{Black}{\mathsf{LR}}} -  
	\Pi^{\dom{W}}_{\leq t, \textcolor{Black}{\mathsf{LR}}}
	= \Pi^{\RR^2}\reg{LR}\cdot
	\br{ \Pi^{\mathrm{db}}_{\leq t, \textcolor{Black}{\mathsf{LR}}} - J^{\dom{W}}_{\leq t, \textcolor{Black}{\mathsf{LR}}}} \cdot\Pi^{\RR^2}\reg{LR} \enspace,
	\]
	and $\Pi^{\RR^2}$ commutes with $Q[C,D]\reg{LR}$ since $\Pi^{\RR^2}$ is the sum of projectors onto the symmetric subspaces.
	We have
	\begin{align*}
		&\norm{\expect{C,D\leftarrow \D}{(C\reg{A} \otimes Q[C,D]\reg{LR})^\dag \cdot\br{
	\Pi^{\mathrm{DB}}_{\leq t, \textcolor{Black}{\mathsf{LR}}} -  
	\Pi^{\dom{W}}_{\leq t, \textcolor{Black}{\mathsf{LR}}}
	}\cdot (C\reg{A} \otimes Q[C,D]\reg{LR})}}_{\infty}\\
	\leq &\norm{\expect{C,D\leftarrow \D}{(C\reg{A} \otimes Q[C,D]\reg{LR})^\dag \cdot\br{
	 \Pi^{\mathrm{db}}_{\leq t, \textcolor{Black}{\mathsf{LR}}} - J^{\dom{W}}_{\leq t, \textcolor{Black}{\mathsf{LR}}}
	}\cdot (C\reg{A} \otimes Q[C,D]\reg{LR})}}_{\infty}\\
	=&\max_{\substack{l,r\geq 0\\l+r\leq t}}
	\norm{\expect{C,D\leftarrow \D}{(C\reg{A} \otimes Q[C,D]\reg{LR})^\dag \cdot\br{
	\Pi^{\mathrm{db}}_{l,r,\textcolor{Black}{\mathsf{LR}}} - J^{\dom{W}}_{l,r,\textcolor{Black}{\mathsf{LR}}}
	}\cdot (C\reg{A} \otimes Q[C,D]\reg{LR})}}_{\infty} \enspace,
	\end{align*}
	where the equality holds because $\Pi^{\mathrm{db}}_{\leq t,\textcolor{Black}{\mathsf{LR}}}$ and $ J^{\dom{W}}_{\leq t, \textcolor{Black}{\mathsf{LR}}}$ are block diagonal with respect to $l$ and $r$. Then from \lem{simple_twirl}, we have
	\begin{align*}
		&\norm{\expect{C,D\leftarrow \D}{(C\reg{A} \otimes Q[C,D]\reg{LR})^\dag \cdot\br{
	\Pi^{\mathrm{DB}}_{\leq t, \textcolor{Black}{\mathsf{LR}}} -  
	\Pi^{\dom{W}}_{\leq t, \textcolor{Black}{\mathsf{LR}}}
	}\cdot (C\reg{A} \otimes Q[C,D]\reg{LR})}}_{\infty}\\
	\leq &\max_{\substack{l,r\geq 0\\l+r\leq t}}~ \frac{4l+r}{N-1} + \frac{7rN}{N-2l-2r+2}\cdot \sqrt{\frac{2(l+r)}{N}} \\
	\leq &\max_{\substack{l,r\geq 0\\l+r\leq t}}~ \frac{3l+t}{N-1} + \frac{7rN}{N-2t+2}\cdot \sqrt{\frac{2t}{N}} \\
	\leq &\frac{t}{N-1} + \frac{7tN}{N-2t+2}\cdot \sqrt{\frac{2t}{N}} 
	\leq \frac{8tN}{N-2t+2}\cdot \sqrt{\frac{2t}{N}} \leq 16t\cdot \sqrt{\frac{2t}{N}}\enspace,
	\end{align*}
	where the second inequality is from $l+r\leq t$, the third inequality holds since the maximal value is achieved at $r=t$, and the last one is from $t\leq N/4$.
\end{proof}

\end{document}